\documentclass[conference]{sig-alternate-05-2015}

\CopyrightYear{2016}
\setcopyright{acmcopyright}
\conferenceinfo{PODS'16,}{June 26-July 01, 2016, San Francisco, CA, USA}
\isbn{978-1-4503-4191-2/16/06}\acmPrice{\$15.00}
\doi{http://dx.doi.org/10.1145/2902251.2902275}
\usepackage{etoolbox}
\makeatletter
\patchcmd{\maketitle}{\@copyrightspace}{}{}{}
\makeatother

\usepackage[normalem]{ulem}

\numberofauthors{3}
\author{
\alignauthor Alin Deutsch\\
\affaddr{UC San Diego}\\
\email{deutsch@cs.ucsd.edu}
\alignauthor Yuliang Li\\
\affaddr{UC San Diego}\\
\email{yul206@eng.ucsd.edu }
\alignauthor Victor Vianu \\
\affaddr{UC San Diego \& INRIA Saclay} \\ % Perhaps 
\email{vianu@cs.ucsd.edu}
}

\title{Verification of Hierarchical Artifact Systems}
\date{}

\newtheorem{theorem}{Theorem}
\newtheorem{corollary}[theorem]{Corollary}
\newtheorem{lemma}[theorem]{Lemma}

\newtheorem{definition}[theorem]{Definition}

\newtheorem{fact}[theorem]{Fact}

\newtheorem{example}[theorem]{Example}

\usepackage{graphicx}

\usepackage[usenames]{color}
\usepackage{enumitem}
\setlist[itemize]{leftmargin=*}
\setlist[enumerate]{leftmargin=*}

\newcommand{\alin}[1]{{\it\small\textcolor{green}{[[[ {#1}\ --alin ]]]}}}
\newcommand{\victor}[1]{{\it\small\textcolor{red}{[[[ {#1}\ --victor ]]]}}}
\newcommand{\yuliang}[1]{{\it\small\textcolor{blue}{[[[ {#1}\ --yuliang ]]]}}}
\newcommand{\reviewer}[1]{{\it\small\textcolor{cyan}{[[[ {#1}\ --reviewers ]]]}}}
\renewcommand{\alin}[1]{}
\renewcommand{\victor}[1]{}
\renewcommand{\yuliang}[1]{}
\renewcommand{\reviewer}[1]{}

\newcommand{\eat}[1]{}

\newcommand{\calc}{{\cal C}}
\newcommand{\calb}{{\cal B}}

\newcommand{\calk}{{\cal K}}

\newcommand{\cala}{{\cal A}}
\newcommand{\calh}{{\cal H}}

\newcommand{\call}{{\cal L}}
\newcommand{\calp}{{\cal P}}
\newcommand{\calr}{{\cal R}}
\newcommand{\calv}{{\cal V}}
\newcommand{\cals}{{\cal S}}
\newcommand{\calt}{{\cal T}}
\newcommand{\cale}{{\cal E}}
\newcommand{\cali}{{\cal I}}

\newcommand{\lora}{\longrightarrow}
\newcommand{\goto}[1]{\stackrel{#1}{\lora}}
\newcommand{\db}{\mathcal{DB}}

\newcommand{\ltlfo}{\ensuremath{\text{LTL-FO} } }
\newcommand{\hltlfo}{\ensuremath{\text{HLTL-FO} } }

\newcommand{\ainit}{\ensuremath{\mathtt{init}}}
\newcommand{\aactive}{\ensuremath{\mathtt{active}}}
\newcommand{\aclosed}{\ensuremath{\mathtt{closed}}}

\newcommand{\anull}{\ensuremath{\mathtt{null}}}

\newcommand{\buchi}{B\"{u}chi}

\newcommand{\T}{\emph{Tree}}
\newcommand{\Sym}{\ensuremath{\text{\bf \em Sym} } }

\newcommand{\Tree}{\ensuremath{\text{\bf \em Tree} } }
\newcommand{\xttcup}{\bar{x}^T_{T^\uparrow_c}}
\newcommand{\xttcdown}{\bar{x}^T_{T^\downarrow_c}}

\newcommand{\varid}{\ensuremath{\emph{VAR}_{id}}}
\newcommand{\varnum}{\ensuremath{\emph{VAR}_{\mathbb{R}}}}
\newcommand{\IND}{\ensuremath{\text{\bf IND} } }
\newcommand{\FD}{\ensuremath{\text{\bf FD} } }

\newcommand{\ts}{\ensuremath{\emph{TS}(T) } }
\newcommand{\tsib}{\ensuremath{\emph{TS}_{\emph{ib}}(T) } }

\newcommand{\aflightid}{\ensuremath{\mathtt{flight\_id}}}
\newcommand{\ahotelid}{\ensuremath{\mathtt{hotel\_id}}}

\newcommand{\aticketprice}{\ensuremath{\mathtt{ticket\_price}}}

\newcommand{\ahotelprice}{\ensuremath{\mathtt{hotel\_price}}}
\newcommand{\alowerprice}{\ensuremath{\mathtt{discount\_price}}}
\newcommand{\ahigherprice}{\ensuremath{\mathtt{unit\_price}}}

\newcommand{\aamountpaid}{\ensuremath{\mathtt{amount\_paid}}}
\newcommand{\anewamountpaid}{\ensuremath{\mathtt{new\_amount\_paid}}}
\newcommand{\ahotelamountpaid}{\ensuremath{\mathtt{hotel\_amount\_paid}}}
	
\newcommand{\aamountrefunded}{\ensuremath{\mathtt{amount\_refunded}}}
\newcommand{\astatus}{\ensuremath{\mathtt{status}}}

\newcommand{\dbhotels}{\ensuremath{\mathtt{HOTELS}}}
\newcommand{\dbflights}{\ensuremath{\mathtt{FLIGHTS}}}

\begin{document}
\maketitle

\begin{abstract}
Data-driven workflows, of which IBM's Business Artifacts are a prime exponent, have 
been successfully deployed in practice, adopted in industrial standards, and
have spawned a rich body of research in academia, focused primarily on static analysis.  
The present work represents a significant advance on the problem of artifact verification, by considering 
a much richer and more realistic model than in previous work, 
incorporating core elements of IBM's successful Guard-Stage-Milestone model. 
In particular, the model features task hierarchy, concurrency, and richer artifact data.
It also allows database key and foreign key dependencies,
as well as arithmetic constraints.
The results show decidability of verification and establish its complexity, making use of novel techniques including
a hierarchy of Vector Addition Systems and a variant of quantifier elimination tailored to our context.
\end{abstract}


\section{Introduction} \label{sec:intro}

The past decade has witnessed the evolution of workflow specification frameworks from the traditional
process-centric approach towards data-awareness.
Process-centric formalisms focus on control flow while
under-specifying the underlying data and its manipulations by the process
tasks, often abstracting them away completely.
In contrast, data-aware formalisms treat data as first-class
citizens. A notable exponent of this class is IBM's {\em business artifact model} 
pioneered in~\cite{Nigam03:artifacts}, successfully deployed in practice ~\cite{Bhatt-2005:Artifacts-pharm,Bhatt-2007:artifacts-customer-engagements,IGF-case-study:BPM-2009,Cordys-case-management,IBM-case-mgmt} and adopted in industrial standards.
Business artifacts have also spawned a rich body of research in academia, dealing with issues ranging from formal semantics 
to static analysis (see related work). 

In a nutshell, business artifacts (or simply ``artifacts'') model key business-relevant
entities, which are updated by a set of services that implement
business process tasks, specified declaratively by pre-and-post conditions. 
A collection of artifacts and services is called an {\em artifact system}. 
IBM has developed several variants of artifacts, of which the most recent is 
Guard-Stage-Milestone (GSM) \cite{Damaggio:BPM11,GSM:DEBS-2011}. 
The GSM approach provides rich structuring mechanisms for services, including 
parallelism, concurrency and hierarchy, and has been incorporated in the OMG standard 
for Case Management Model and Notation (CMMN) \cite{OMG:CMMN:Beta1,GSM-CMMN:2012}.

Artifact systems deployed in industrial settings typically specify very complex workflows 
that are prone to costly bugs, whence the need for verification of critical properties.
Over the past few years, we have embarked upon a study of the verification problem for artifact systems.
Rather than relying on general-purpose software verification tools suffering from well-known limitations,
our aim is to identify practically relevant classes of artifact systems and properties for which {\em fully automatic} verification is possible.
This is an ambitious goal, since artifacts are infinite-state systems due to the presence of unbounded data. 
Our approach relies critically on the declarative nature of service specifications and brings into play a novel marriage 
of database and computer-aided verification techniques. 
 
In previous work~\cite{DHPV:ICDT:09,tods12}, we studied the verification problem for a bare-bones variant of artifact systems,
without hierarchy or concurrency, in which each artifact consists of a flat tuple of evolving values and the services 
are specified by simple pre-and-post conditions on the artifact and database.
More precisely, we considered the problem of statically checking whether all runs of an
artifact system satisfy desirable properties expressed in $\ltlfo$, an 
extension of linear-time temporal logic where propositions are interpreted as $\exists$FO sentences 
on the database and current artifact tuple. 
In order to deal with the resulting infinite-state system, 
we developed in \cite{DHPV:ICDT:09} a symbolic approach allowing a reduction to finite-state model checking and
yielding a {\sc pspace} verification algorithm for the simplest variant of the
model (no database dependencies and uninterpreted data domain).  
In \cite{tods12} we extended our approach to allow for database dependencies and numeric data 
testable by arithmetic constraints. Unfortunately, decidability was obtained subject to 
a rather complex semantic restriction on the artifact system and property (feedback freedom), 
and the verification algorithm has non-elementary complexity. 

The present work represents a significant advance on the artifact verification problem on several fronts.
We consider a much richer and more realistic model, called {\em Hierarchical Artifact System} (HAS), 
abstracting core elements of the GSM model. In particular, the model features task hierarchy, concurrency, and richer artifact data
(including updatable artifact relations). 
We consider properties expressed in 
a novel {\em hierarchical} temporal logic, $\hltlfo$, that is well-suited to the model.
Our main results establish the complexity of checking $\hltlfo$ properties
for various classes of HAS, highlighting the impact of various features on verification.  
The results require qualitatively novel techniques, because the reduction to finite-state model checking
used in previous work is no longer possible. Instead, the richer model requires the use of 
a hierarchy of Vector Addition Systems with States (VASS) \cite{vass}. 
The arithmetic constraints are handled using quantifier elimination techniques, adapted to our setting. 
 
We next describe the model and results in more detail.
A HAS consists of a database and a hierarchy (rooted tree) of {\em tasks}. 
Each task has associated to it local evolving data consisting of a tuple of artifact variables and an updatable artifact relation.
It also has an associated set of {\em services}.
Each application of a service is guarded by a pre-condition on the database and local data and
causes an update of the local data, specified by a post condition (constraining the next artifact tuple) 
and an insertion or retrieval of a tuple from the artifact relation. In addition, a task may invoke a child task with a tuple of parameters, 
and receive back a result if the child task completes. A run of the artifact system consists of an infinite sequence 
of transitions obtained by any valid interleaving of concurrently running task services.

In order to express properties of HAS's we introduce {\em hierarchical} $\ltlfo$ ($\hltlfo$).
Intuitively, an $\hltlfo$ formula uses as building blocks $\ltlfo$ formulas acting on runs of individual tasks, called local runs, 
referring only to the database and local data, and can recursively state $\hltlfo$ properties on runs resulting from calls to children tasks.
The language $\hltlfo$ closely fits the computational model and is also motivated on technical grounds discussed in the paper.
A main justification for adopting $\hltlfo$ is that $\ltlfo$ (and even LTL) properties are undecidable for HAS's.

Hierarchical artifact systems as sketched above provide powerful extensions to the variants we previously studied, 
each of which immediately leads to undecidability of verification 
if not carefully controlled.  Our main contribution is to put forward a package of restrictions that
ensures decidability while capturing a significant subset of the GSM model.
This requires a delicate balancing act aiming to
limit the dangerous features while retaining their most useful aspects.
In contrast to \cite{tods12}, this is achieved without the need for unpleasant
semantic constraints such as feedback freedom.  
The restrictions are discussed in detail in the paper, and shown to be necessary by undecidability results.

% end eat

The complexity of verification under various restrictions is summarized in 
Tables \ref{tab:complexity1} (without arithmetic) and \ref{tab:complexity2} (with arithmetic). 
As seen, the complexity ranges from {\sc pspace} to non-elementary for various packages of features. 
The non-elementary complexity (a tower of exponentials whose height is the depth of the hierarchy) 
is reached for HAS with cyclic schemas, artifact relations and arithmetic. 
For acyclic schemas, which include the widely used Star (or Snowflake) schemas \cite{starschema1, starschema2}, the complexity ranges from {\sc pspace} (without 
arithmetic or artifact relations) to double-exponential space (with both arithmetic and artifact relations). 
This is a significant improvement over the previous algorithm of \cite{tods12}, which even for acyclic schemas 
has non-elementary complexity in the presence of arithmetic 
(a tower of exponentials whose height is the square of the total number of artifact variables in the system).  

The paper is organized as follows. The HAS model is presented in Section \ref{sec:framework}. We present its syntax and semantics,
including a representation of runs as a tree of local task runs, that factors out interleavings of independent concurrent tasks.  
An example HAS modeling a simple travel booking process is provided in the appendix. % Section \ref{sec:example}.
The temporal logic $\hltlfo$ is introduced in Section \ref{sec:ltl-fo}, together with a corresponding
extension of B\"uchi automata to trees of local runs. 
In Section \ref{sec:verification} we prove the decidability of verification 
without arithmetic, and establish its complexity.  To this end, we develop a symbolic representation 
of HAS runs and a reduction of model checking to state reachability problems in a set of nested VASS (mirroring the task hierarchy).
In Section \ref{sec:arithmetic} we show how the verification results can be extended in the presence of arithmetic.
Section \ref{sec:undecidability} traces the boundary of decidability, showing that the main restrictions adopted in defining the HAS model 
cannot be relaxed.
Finally, we discuss related work in Section \ref{sec:related} and conclude.
% More details and proofs are provided in the extended appendix of the full version \cite{arxiv} of this paper.
The appendix provides more details and proofs, together with our running example.
% An appendix provides details, proofs and an extended example.

\section{Framework}\label{sec:framework}

\reviewer{
\begin{itemize}
\item why for example we cannot "simulate" a hierarchical artifact system by a non-hierarchical one
\item it would have been nice to get some self-contained intuition on the formal framework.
\item Move examples and motivations to main body of the paper.
\item Explain why there are deviations compared to the nonhierarchical case. For example, the notion of database is different (with the keys and foreign keys).
\item Explain why verification of a HLTL-specification on a hierarchical artifact system cannot be reduced to verification of an LTL-specification on an artifact system such as in the [24] paper
\item I struggled to see what you can (and cannot) do with the relational instance within a task (even after looking at the example in the appendix). I would suggest having a few sentences that just explain how tuples can flow in and out of these instances.
\item possibility of flattening
\end{itemize}
}

%The arithmetic constraints considered here are over domain $\mathbb{R}$, the real
%numbers. While databases could use non-numeric data, we assume for uniformity, and
%without loss of generality, that all structures are over $\mathbb{R}$. 

%We denote by $\calc$ an infinite set of relation symbols,each of which has a fixed interpretation 
%as the set of solutions of a finite set of linear inequalities with integer
%coefficients. By slight abuse, we sometimes use the same notation for a relation 
%symbol in $\calc$ and its fixed interpretation.

In this section we present the syntax and semantics of Hierarchical Artifact Systems (HAS's). 
We begin with the underlying database schema. 
\vspace{-1mm}
\begin{definition}
A \textbf{database schema} $\db$ is a finite set of relation symbols, where 
each relation $R$ of $\db$ has an associated sequence of distinct attributes 
containing the following: 
\vspace{-1mm}
\begin{itemize}\itemsep=0pt\parskip=0pt
\item a key attribute $\emph{ID}$ (present in all relations), 
\item a set of foreign key attributes $\{F_1, \dots, F_m\}$, and
\item a set of non-key attributes $\{A_1, \dots, A_n\}$ disjoint from \\
$\{\emph{ID}, F_1, \dots, F_m\}$.
\end{itemize}
\vspace{-1mm}
To each foreign key attribute $F_i$ of $R$ is associated a relation $R_{F_i}$ of $\db$
and the inclusion dependency $R[F_i] \subseteq R_{F_i}[\emph{ID}]$.
It is said that $F_i$ references $R_{F_i}$.
\end{definition} \vspace{-1mm}

\reviewer{motivation for non-key attributes being all numeric} \victor{Added a paragraph below}
\reviewer{It would also be nice to understand the reason for insisting on a single ID key for all relations and that non-numeric attributes can only be foreign keys to such IDs and nothing else.} \victor{See Discussion}
\reviewer{page 3: The notion of linearly cyclic looks similar to ``flatness''. Is linearly cyclic a standard terminology?}
\victor{Alin, do you know what "flatness" is ?}

The domain $Dom(A)$ of each attribute $A$ depends on its type. The domain of all non-key attributes is
numeric, specifically $\mathbb{R}$. 
The domain of each key attribute is a countable infinite domain disjoint from $\mathbb{R}$.
For distinct relations $R$ and $R'$, $Dom(R.\emph{ID}) \cap Dom(R'.\emph{ID}) = \emptyset$.
The domain of a foreign key attribute $F$ referencing $R$ is $Dom(R.\emph{ID})$. We denote by
$\emph{DOM}_{id} = \cup_{R \in \db} Dom(R.\emph{ID})$. 
Intuitively, in such a database schema, each
tuple is an object with a \emph{globally} unique id. This id does not appear
anywhere else in the database except as foreign keys referencing it.
An {\em instance} of a database schema $\db$ is a mapping $D$ associating to each relation symbol $R$ a finite relation
$D(R)$ of the same arity of $R$, whose tuples provide, for each attribute, a value from its domain. In addition,
$D$ satisfies all key and inclusion dependencies associated with the keys and foreign keys of the schema. 
The active domain $D$, denoted $\texttt{adom}(D)$, consists of all elements of $D$ (id's and reals). 
% We denote $\emph{adom}_{id}(D) = \emph{adom}(D) \cap \emph{DOM}_{id}$ and $\emph{adom}_{\mathbb{R}}(D) = \emph{adom}(D) \cap \mathbb{R}$.
%\yuliang{I think $\emph{adom}_{id}(D)$ and $\emph{adom}_{\mathbb{R}}(D)$ are not used.}
A database schema $\db$ is {\em acyclic} if there are no cycles in the references induced by foreign keys.
More precisely, consider the labeled graph FK whose nodes are the relations of the schema and in which there is an edge
from $R_i$ to $R_j$ labeled with $F$ if $R_i$ has a foreign key attribute $F$ referencing $R_j$. The schema $\db$ is {\em acyclic} if the graph FK
is acyclic, and it is {\em linearly-cyclic} if each relation $R$ is contained in at most one simple cycle.

The assumption that the ID of each relation is a single attribute is made for simplicity, and multiple-attribute IDs
can be easily handled. The fact that the domain of all non-key attributes is numeric is also harmless. 
Indeed, an uninterpreted domain on which only equality can be used can be easily simulated. 
Note that the keys and foreign keys used on our schemas are special cases of the dependencies used in \cite{tods12}.
The limitation to keys and foreign keys is one of the factors leading to improved complexity of verification and still
captures most schemas of practical interest. 

We next proceed with the definition of tasks and services, described informally in the introduction.
The definition imposes various restrictions needed for decidability of verification.
These are discussed and motivated in Section \ref{sec:undecidability}.

Similarly to the database schema, we consider two infinite, disjoint sets $\varid$ of ID variables and $\varnum$ of
numeric variables. We associate to each variable $x$ its {\em domain} $Dom(x)$.
If $x \in \varid$, then $Dom(x) = \{\anull\} \cup \emph{DOM}_{id}$, 
where $\anull \not\in \emph{DOM}_{id} \cup \mathbb{R}$
($\anull$ plays a special role that will become clear shortly).
If $x \in \varnum$, then $Dom(x) = \mathbb{R}$. An {\em artifact variable} is a variable in $\varid \cup \varnum$.
If $\bar{x}$ is a sequence of artifact variables, a {\em valuation} of $\bar{x}$ 
is a mapping $\nu$ associating to each variable in $\bar{x}$ an element of its domain $Dom(x)$. 

\vspace{-2mm}
\begin{definition}
A \textbf{task schema} over database schema $\db$ is a triple $T = \langle \bar{x}^T, S^T, \bar{s}^T \rangle$ where 
$\bar{x}^T$ is a sequence of distinct artifact variables, $S^T$ is a relation symbol not in $\db$ with associated arity $k$, and 
$\bar{s}^T$ is a sequence of $k$ distinct id variables in $\bar{x}^T$. 
\end{definition}
\vspace{-1mm}

We denote by $\bar x^T_{id} = \bar x^T \cap \varid$ and $\bar x^T_{\mathbb{R}} = \bar x^T \cap \varnum$.
We refer to $S^T$ as the {\em artifact relation} or {\em set} of $T$.
\vspace{-2mm}
\begin{definition}
An \textbf{artifact schema} is a tuple $\mathcal{A} = \langle \calh, \db \rangle$ 
where $\mathcal{DB}$ is a database schema and $\calh$ is a rooted tree of task schemas over $\db$
with pairwise disjoint sets of artifact variables and distinct artifact relation symbols.
\end{definition}\vspace{-1mm}

\reviewer{disjoint $\rightarrow$ pairwise disjoint}
\victor{Fixed}

The rooted tree $\calh$ defines the {\em task hierarchy}. 
Suppose the set of tasks is $\{T_1, \ldots, T_k\}$. 
For uniformity, we always take task $T_1$ to be the root of $\calh$. 
We denote by $\preceq_\calh$ (or simply $\preceq$ when $\calh$ is understood) the 
partial order on $\{T_1, \dots, T_k\}$ induced by $\calh$ (with $T_1$ the minimum). For a node $T$ of $\calh$, 
we denote by $\emph{tree(T)}$ the subtree of $\calh$ rooted at $T$, 
$\emph{child}(T)$ the set of children of $T$ (also called {\em subtasks} of $T$), 
$\emph{desc}(T)$ the set of descendants of $T$ (excluding $T$). Finally, $\emph{desc}^*(T)$ denotes
$\emph{desc}(T) \cup \{T\}$.
We denote by $\cals_\calh$ (or simply $\cals$ when $\calh$ is understood) the relational schema $\{S^{T_i} \mid 1 \leq i \leq k\}$.
An instance of $\cals$ is a mapping associating to each $S^{T_i} \in \cals$ a finite relation over $\emph{DOM}_{id}$
of the same arity.
\vspace{-2mm}
\begin{definition}
An \textbf{instance} of an artifact schema $\mathcal{A} = \langle \calh, \db \rangle $ is a tuple 
$\bar I = \langle \bar \nu, stg, D, \bar S \rangle$ where 
$D$ is a finite instance of $\mathcal{DB}$, $\bar S$ a finite instance of $\cals$,
$\bar{\nu}$ a valuation of 
$\bigcup_{i=1}^k \bar{x}^{T_i}$, 
and $stg$ (standing for ``stage'') a mapping of $\{T_1, \dots, T_k\}$ to 
$\{\ainit, \aactive, \aclosed \}$. 
\reviewer{need to make specific that $stg$ stands for stage}
\victor{Done.}
\end{definition}
\vspace{-1mm}

The stage $stg(T_i)$ of a task $T_i$ has the following intuitive meaning in the context of a run of its parent: 
$\ainit$ indicates that $T_i$ has not yet been called within the run,  
$\aactive$ says that $T_i$ has been called and has not returned its answer, and $\aclosed$ indicates that $T_i$
has returned its answer. As we will see, a task $T_i$ can only be called once within a given run of its parent.
However, it can be called again in subsequent runs.   

We denote by $\calc$ an infinite set of relation symbols, each of which has a fixed interpretation 
as the set of real solutions of a finite set of polynomial inequalities with integer
coefficients. By slight abuse, we sometimes use the same notation for a relation 
symbol in $\calc$ and its fixed interpretation.
For a given artifact schema $\mathcal{A} = \langle \calh, \mathcal{DB} 
\rangle$ and a sequence $\bar x$ of variables, 
a {\em condition} on $\bar x$ is a quantifier-free FO formula over 
$\mathcal{DB} \cup \calc \cup \{=\}$ \reviewer{need explanation of $\calc$} \victor{Added above.}
whose variables are included in $\bar{x}$.
The special constant $\anull$ can be used in equalities with ID variables.
For each atom $R(x, y_1, \dots, y_m, z_1, \dots, z_n)$ of relation
$R(\emph{ID}, A_1, \dots, A_m, F_1, \dots, F_n) \in \db$, 
$\{x, z_1, \dots, z_n\} \subseteq \varid$ and $\{y_1, \dots, y_m\} \subseteq \varnum$.
Atoms over $\calc$ use only numeric variables.
If $\alpha$ is a condition on $\bar x$, 
$D$ is an instance of $\mathcal{DB}$ and $\nu$ a valuation of $\bar x$, 
we denote by $D \cup \calc \models \alpha(\nu)$ the fact that $D \cup 
\calc$ satisfies $\alpha$ with valuation $\nu$ with standard semantics. 
For an atom $R(\bar{y})$ in $\alpha$ where $R \in \db$ and $\bar{y} \subseteq \bar{x}$,
if $\nu(y) = \anull$ for any $y \in \bar{y}$, then $R(\bar{y})$ is false.
%\victor{We need to spell out the semantics when $\nu$ has nulls.} \yuliang{Fixed.}
%When $\alpha(\nu)$ is evaluated, each relational atom $R(x, y_1, \dots, y_m, z_1, \dots, z_n)$ is \emph{true} iff 
%$(\nu(x), \nu(y_1), \dots, \nu(y_m), \nu(z_1), \dots, \nu(z_n))$ is in $D(R)$. The standard
%semantics is used for $\calc$ and $\{=\}$.
%\yuliang{As an extension, we can have $S^T( \bar{c} )$ for a sequence
%of constants $\bar{c}$ in conditions.}
%\victor{Let's decide later on this.}

% and by CQ$^\neg$ the
%formulas built from literals (positive and negated atoms) using only conjunction and existential quantification.
%For each $x_i$, $\emph{desc}(x_i)$ denotes the descendants of $x_i$ in $\calh$ (excluding $x_i$)
%and $\emph{col}(x_i)$ denotes the set of variables colinear with $x_i$, i.e. descendants or ancestors.

We next define services of tasks. We start with internal services, which update
the artifact variables and artifact relation of the task. \reviewer{artifact tuple is not defined.}  \yuliang{fixed}
\vspace{-1mm}
\begin{definition}
Let $T = \langle \bar{x}^T, S^T, \bar{s}^T \rangle$ 
be a task of an artifact schema $\mathcal{A}$.
An \textbf{internal service} $\sigma$ of $T$
is a tuple $\langle \pi, \psi, \delta \rangle$ where:
\vspace*{-1.5mm}
\begin{itemize}\itemsep=0pt\parskip=0pt
\item $\pi$ and $\psi$, called \emph{pre-condition} and \emph{post-condition}, 
respectively, are conditions over $\bar{x}^T$
\item $\delta \subseteq \{+S^T(\bar{s}^T), -S^T(\bar{s}^T)\}$ is a set of {\em set updates};  
$+S^T(\bar{s}^T)$ and $-S^T(\bar{s}^T)$ are 
called the \textbf{insertion} and \textbf{retrieval} of $\bar{s}^T$, respectively. 
\reviewer{$\delta$ is from a two element set. And explain why do we need empty $\delta$ and $\delta$ with both insert and retrieve}
\reviewer{explanation of retrieve}
\end{itemize}
\end{definition}
\vspace{-1mm}

Intuitively, $+S^T(\bar{s}^T)$ causes an insertion of the {\em current} value of $\bar{s}^T$ into $S^T$,
while $-S^T(\bar{s}^T)$ causes the removal of some non-deterministically chosen current tuple of $S^T$
and its assignment as the {\em next} value of $\bar{s}^T$. In particular, if $\delta = \{+S^T(\bar{s}^T), -S^T(\bar{s}^T)\}$,
the tuple inserted by $+S^T(\bar{s}^T)$ and the one retrieved by $-S^T(\bar{s}^T)$ are generally distinct, 
but may be the same as a degenerate case (see definition of the semantics below).

As will become apparent, although pre-and-post conditions are quantifier-free, $\exists$FO conditions can be simulated 
by adding variables to $\bar{x}^T$. \reviewer{At this point in the paper
we have seen no semantics, so how can we observe this?} \victor{Reformulated}

An internal service of a task $T$ specifies transitions that only 
modify the variables $\bar{x}^T$ of $T$ and the contents of $S^T$. 
Interactions among tasks are specified using two kinds of special services, 
called the \emph{opening-services} and \emph{closing-services}.
\reviewer{The paragraph  starting "Intuitively..." before Definition 7 and the sentence after
Definition 7 starting with "Intuitively" were both useful, but it would be nice if the conditions outlined there could be mapped to ones
in the formal definition.}

\vspace{-2mm}
\begin{definition}
Let $T_c$ be a child of a task $T$ in $\cala$. \\
(i) The \textbf{opening-service} $\sigma_{T_c}^o$ of $T_c$ is a tuple $\langle 
\pi, f_{in} \rangle$, where $\pi$ is a condition over $\bar{x}^{T}$, and
$f_{in}$ is a partial 1-1 mapping \reviewer{explain why it has to be 1-1} from $\bar{x}^{T_c}$ to $\bar{x}^{T}$ (called the input variable mapping).  
We denote $dom(f_{in})$ by $\bar{x}_{in}^{T_c}$, called the \textbf{input variables} of $T_c$,
and $range(f_{in})$ by  $\bar x^{T}_{{T_c}^\downarrow}$ (the variables of $T$ passed as input to $T_c$). \\
(ii) The \textbf{closing-service} $\sigma_{T_c}^c$ of $T_c$ is a tuple $\langle 
\pi, f_{out} \rangle$, 
where $\pi$ is a condition over $\bar{x}^{T_c}$,
and $f_{out}$ is a partial 1-1 mapping from $\bar{x}^{T}$ to $\bar{x}^{T_c}$ (called the output variable mapping). 
%\yuliang{I think $f_{out}$ is also 1-1.}
We denote $dom(f_{out})$ by $\bar{x}^{T}_{{T_c}^\uparrow}$, referred to as the \textbf{returned variables} from $T_c$.
It is required that 
$\bar{x}_{{T_c}^\uparrow}^{T} \cap {\bar x}_{in}^{T} = \emptyset$ \yuliang{this restriction is said twice} \victor{Where?}.
We denote by $\bar{x}^{T_c}_{ret}$ the \textbf{to-be-returned variables} (or return variables),
defined as $range(f_{out})$. 
%if $T_c$ is a non-root task and otherwise $\emptyset$.
\end{definition}
\vspace{-2mm}

Intuitively, the opening-service  $\langle \pi, f_{in} \rangle$ of a task $T_c$ specifies the condition $\pi$ that the parent task $T$ 
has to satisfy in order to open $T_c$. When $T_c$ is opened, a subset of the variables of $T$
are \yuliang{typo fixed} 
sent to $T_c$ according to the mapping $f_{in}$. 
Similarly, the closing-service $\langle \pi, f_{out} \rangle$ specifies the 
condition $\pi$ that $T_c$ has to satisfy in order to be closed and return to $T$. When $T_c$ is closed, a
subset of $\bar{x}^{T_c}$ is sent back to $T$, as specified by $f_{out}$.

\reviewer{motivation of this restriction.} 
\victor{Explained later.}

For uniformity of notation, we also equip the root task $T_1$ with a service $\sigma^o_{T_1}$ with pre-condition {\em true} that initiates 
the computation by providing a valuation to a designated subset $\bar x^{T_1}_{in}$ of $\bar x^{T_1}$ (the input variables of $T_1$),
and a service $\sigma^c_{T_1}$ whose pre-condition is \emph{false} 
(so it never occurs in a run). \yuliang{need to unify $\sigma^o_{T_1}$ and the global pre-condition $\Pi$}
For a task $T$ we denote by $\Sigma_T$ the set of its internal services,
$\Sigma_T^{oc} = \Sigma_T \cup \{\sigma^o_T, \sigma^c_T\}$, 
$\Sigma_T^{\emph{obs}} = \Sigma_T^{oc} \cup \{\sigma^o_{T_c}, \sigma^c_{T_c} \mid T_c \in \emph{child}(T)\}$, and
$\Sigma_T^\delta = \Sigma_T \cup \{\sigma^o_T\} \cup \{\sigma^c_{T_c} \mid T_c \in \emph{child}(T)\}$.
Intuitively, $\Sigma_T^{\emph{obs}}$ consists of the services observable in runs of task $T$ 
and $\Sigma_T^\delta$ consists of services whose application can modify the variables $\bar x^T$. 
\reviewer{move the definition of $\Sigma_T$ after the following definition?}

\vspace{-2mm}
\begin{definition}
A \emph{Hierarchical Artifact System} (HAS) is a triple $\Gamma = \langle \mathcal{A}, \Sigma, \Pi \rangle$, 
where $\mathcal{A}$ is an artifact schema, $\Sigma$ is a set of services over $\mathcal{A}$ including 
$\sigma^o_T$ and  $\sigma^c_T$ for each task $T$ of $\mathcal{A}$, 
and $\Pi$ is a condition over $\bar x^{T_1}_{in}$  (where $T_1$ is the root task).
\victor{Changed. Note that this is not the precondition of $\sigma_{T_1}^o$.}
\end{definition}
\vspace*{-1mm}

We next define the semantics of HAS's. Intuitively, a run of a HAS on a database $D$ consists of 
an infinite sequence of transitions among HAS instances (also referred to as configurations, or snapshots), 
starting from an initial artifact tuple satisfying pre-condition $\Pi$. 
At each snapshot, each active task $T$ can open a subtask $T_c$ if the pre-condition of the opening service of $T_c$ holds, and the values of
a subset of $\bar{x}^T$ is passed to $T_c$ as its input variables. $T_c$ can be closed if the pre-condition of
its closing service is satisfied. When $T_c$ is closed, the values of a subset of $\bar{x}^{T_c}$ are sent to $T$ as
$T$'s returned variables from $T_c$. An internal service of $T$ can only be applied after all active subtasks of $T$ have returned their answer. 

\vspace{2mm}
Because of the hierarchical structure, and the locality of task specifications, 
the actions of concurrently active children of a given task are independent of each other and can be arbitrarily interleaved.
To capture just the essential information, factoring out the arbitrary interleavings, we first define the notion
of {\em local run} and {\em tree of local runs}. Intuitively, a local run of a task consists of a sequence 
of services of the task, together with the transitions they cause on the task's local artifact variables and relation. 
The tasks's input and output are also specified. A tree of local runs captures the relationship between 
the local runs of tasks and those of their subtasks, including the passing of inputs and results.
Then the runs of the full artifact system simply consist of all legal interleavings
of transitions represented in the tree of local runs, lifted to full HAS instances (we refer to these as {\em global runs}).  
We begin by defining instances of tasks and local transitions.
For a mapping $M$, we denote by $M[a \mapsto b]$ the mapping that sends $a$ to $b$ and agrees with $M$ everywhere else.
\vspace{-2mm}
\begin{definition}
Let $T =  \langle \bar{x}^T, S^T, \bar{s}^T \rangle$ be a task in $\Gamma$ and $D$ a database instance over $\mathcal{DB}$.
An {\em instance} of $T$ is a pair $(\nu, S)$ where $\nu$ is a valuation of $\bar x^T$ and $S$ an instance of $S^T$.
For instances $I = (\nu, S)$ and $I' = (\nu', S')$ of $T$ and a service $\sigma \in \Sigma^{obs}_T$,
there is a local transition $I \goto{\sigma} I'$ if the following holds.
If $\sigma$ is an internal service $(\pi,\psi)$, then:
\vspace*{-1mm}
\begin{itemize}\itemsep=0pt\parskip=0pt
\item $D \cup \calc \models \pi(\nu)$ and $D \cup \calc \models \psi(\nu')$
\item $\nu'(y) = \nu(y)$ for each $y$ in $\bar{x}^T_{in}$
\item if $\delta = \{+S^T(\bar{s}^T)\}$, then $S' = S \cup \{\nu(\bar{s}^T)\}$,
\item if $\delta = \{-S^T(\bar{s}^T)\}$, then
$\nu'(\bar{s}^T) \in S$ and $S' =  S - \{\nu'(\bar{s}^T)\}$,
\item if $\delta = \{+S^T(\bar{s}^T), -S^T(\bar{s}^T)\}$, then
$\nu'(\bar{s}^T) \in S \cup \{\nu(\bar{s}^T)\}$ and $S' = (S \cup \{\nu(\bar{s}^T)\}) - \{\nu'(\bar{s}^T)\}$, 
\item if $\delta = \emptyset$ then $S' = S$.
\end{itemize}
\vspace*{-1mm}

If $\sigma =  \sigma^o_{T_c} = \langle \pi, f_{in} \rangle$ is the opening-service for a child $T_c$ of $T$
then $D \cup \calc \models \pi(\nu)$, $\nu' = \nu$ and $S' = S$.
If $\sigma = \sigma^c_{T_c}$ then $S = S'$, $\nu'|(\bar x^T - \bar x^T_{T_c \uparrow}) =
\nu|(\bar x^T - \bar x^T_{T_c \uparrow})$ and $\nu'(z) = \nu(z)$ for every $z \in \bar x^T_{T_c \uparrow} \cap \varid$ for which $\nu(z) \neq \anull$.  Finally, if $\sigma = \sigma^c_T$ then $I' = I$.
\end{definition}

\vspace*{-2mm}
We now define local runs.

\vspace{-2mm}
\begin{definition}
Let $T =  \langle \bar{x}^T, S^T, \bar{s}^T \rangle$ be a non-root task in $\Gamma$ and $D$ a database instance over $\mathcal{DB}$.
A {\em local run} of $T$ over $D$ is a triple 
$\rho_T = (\nu_{in}, \nu_{out}, \{(I_i,\sigma_i)\}_{0 \leq i < \gamma})$, where:
\vspace*{-1mm}
\begin{itemize}\itemsep=0pt\parskip=0pt
\item $\gamma \in \mathbb{N} \cup \{\omega\}$
\item for each $i \geq 0$, $I_i$ is an instance of $T$ and $\sigma_i \in \Sigma^{obs}_T$
\item $\nu_{in}$ is a valuation of $\bar x^T_{in}$ 
\item $\sigma_0 = \sigma^o_T$ and $S_0 = \emptyset$,
\item $\nu_0|{\bar x}^T_{in} = \nu_{in}$, 
$\nu_0(z) = \anull$ for $z \in \varid- \bar{x}^T_{in}$ and  
$\nu_0(z) = 0$ for $z \in \varnum -  \bar{x}^T_{in}$ 
\item if for some $i$, $\sigma_i = \sigma^c_T$ then $\gamma \in \mathbb{N}$ and $i = \gamma-1$ 
(and $\rho_T$ is called a {\em returning} local run)
\item $\nu_{out} = \nu_{\gamma-1} |\bar{x}^T_{ret}$ if $\rho_T$ is a returning run and $\bot$ otherwise
\item a \emph{segment} of $\rho_T$ is a subsequence $\{(I_i,\sigma_i)\}_{i \in J}$, where
$J$ is a maximal interval $[a, b] \subseteq \{i \mid 0 \leq i < \gamma\}$
such that no $\sigma_j$ is an internal service of $T$ for $j \in [a + 1, b]$. A segment $J$
is {\em terminal} if $\gamma \in \mathbb{N}$ and
$b = \gamma - 1$ (and is called returning if $\sigma_{\gamma-1} = \sigma_T^c$ and blocking otherwise). 
Segments of ${\rho}_T$ must satisfy the following
properties. For each child $T_c$ of $T$ there is at most one $i \in J$ such that $\sigma_i = \sigma^o_{T_c}$.
If $J$ is not blocking and such $i$ exists, there is exactly one $j \in J$ for which $\sigma_j = \sigma^c_{T_c}$, and $j > i$.
If $J$ is blocking, there is {\em at most} one such $j$.
%\yuliang{fixed the definition of segments for local runs.} \victor{OK}
\item for every $0 < i < \gamma$, $I_{i-1} \goto{\sigma_i} I_i$. 
\end{itemize}
\vspace*{-1mm}
Local runs of the root task $T_1$ are defined as above, except that $\nu_{in}$ is a valuation of $\bar x^{T_1}_{in}$ such that
$D \cup \calc \models \Pi$, and $\nu_{out} = \bot$ (the root task never returns).
\end{definition}
\vspace*{-1mm}

For a local run as above, we denote $\gamma(\rho_T) = \gamma$.
Note that by definition of segment, a task can call each of its children tasks at most once between two consecutive services in $\Sigma_T^{oc}$
and all of the called children tasks must complete within the segment, unless it is blocking.  
These restrictions are essential for decidability and are discussed in Section \ref{sec:undecidability}.
%Given a local run $\rho_T = (\nu_{in}, \nu_{out}, \{(I_i,\sigma_i)\}_{0 \leq i < \gamma})$, we associate to each $i$ and child $T_c$ of $T$ 
%a {\em stage} $stg_i(T_c) \in \{\ainit,\aactive,\aclosed\}$ as follows. If $\sigma_i \in \Sigma_T$ then $stg_i(T_c) = \ainit$. 
%Otherwise, let $J$ be the segment to which $i$ belongs. Then $stg_i(T_c) = \ainit$ if $T_c$ has not been called within $J$,
%$stg_i(T_c) = \aactive$ if $T_c$ has been called within $J$ and has not yet returned, and $stg_i(T_c) = \aclosed$ if $T_c$ has returned within $J$.
%Thus, in between two internal services of $T$, each $T_c$ starts in stage $\ainit$, then may be called and transitions to $\aactive$, 
%and finally may return and move to stage $\aclosed$.  
\victor{I commented out the definition of stage for local runs, which is very close to that defined earlier for global runs. 
Let's see if this works.}

\reviewer{
page 4: Why do you require that an id variable is not overwritten unless currently null? Is it crucial for decidability? Is it more natural than overwriting? \\
page 4: Why do you forbid internal services when a child task is running? Is it natural? Why cant an internal service of parent execute in parallel with child task?}

Observe that local runs take arbitrary inputs and allow for arbitrary return values from its children tasks.
The valid interactions between the local runs of a tasks and those of its children is captured by the notion of {\em tree of local runs}.   
\vspace*{-2mm}
\begin{definition}
A {\em tree of local runs} is a directed labeled tree $\Tree$ in which each node
is a local run ${\rho}_T$ for some task $T$, and every edge connects a local run of
a task $T$ with a local run of a child task $T_c$ and is labeled with a non-negative integer $i$
(denoted $i({\rho}_{T_c})$). In addition, the following properties are satisfied.
Let ${\rho}_T = (\nu^T_{in}, \nu^T_{out}, \{(I_i,\sigma_i)\}_{0 \leq i < \gamma})$ be a node of $\Tree$, where $I_i = (\nu_i, S_i)$, $i \geq 0$.  Let
$i$ be such that $\sigma_i = \sigma^o_{T_c}$ for some child $T_c$ of $T$. There exists a unique edge
labeled $i$ from ${\rho}_T$ to a node ${\rho}_{T_c} = (\nu_{in}, \nu_{out}, \{(I'_i,\sigma'_i)\}_{0 \leq i < \gamma'})$
of $\Tree$, and the following hold:
\vspace*{-1mm}
\begin{itemize}\itemsep=0pt\parskip=0pt
\item $\nu_{in} = f_{in} \circ \nu_i$ 
where\footnote{Composition is left-to-right.} $f_{in}$ is the input variable mapping of $\sigma_{T_c}^o$
\item ${\rho}_{T_c}$ is a returning run iff there exists $j > i$ such that $\sigma_j = \sigma^c_{T_c}$;
let $k$ be the minimum such $j$. Then $\nu_k(z) = \nu_{out}(f_{out}(z))$ for every $z \in \bar x^{T}_{{T_c}^\uparrow}$
for which $\nu_{k-1}(z) = \anull$, where $f_{out}$ is the output mapping of $\sigma_{T_c}^c$.
\end{itemize}
Finally, for every node $\rho_T$ of $\Tree$, if ${\rho}_T$ is blocking then there exists a child of
${\rho}_T$ that is not returning (so infinite or blocking).
\end{definition}
\vspace*{-1mm}

Note that a tree of local runs may generally be rooted at a local run of any task of $\Gamma$. 
We say that $\Tree$ is {\em full} if it is rooted at a local run of $T_1$.

We next turn to global runs. A global run of $\Gamma$ 
on database instance $D$ over $\mathcal{DB}$
is an infinite sequence $\rho = \{ (I_i, \sigma_i) \}_{i \geq 0}$,
where each $I_i$ is an instance $(\nu_i, stg_i, D, S_i)$ of $\cala$ and $\sigma_i \in \Sigma$,
resulting from a tree of local runs by interleaving its transitions, lifted to full HAS instances
(see Appendix for the formal definition).  
For a tree of local runs $\Tree$, we denote by $\call(\Tree)$ the set of all global runs 
induced by the legal interleavings of $\Tree$. 

%end eat

\section{Hierarchical LTL-FO}\label{sec:ltl-fo}

In order to specify temporal properties of HAS's we use an extension of LTL (linear-time temporal logic).
Recall that LTL is propositional logic augmented with temporal operators $\mathbf{X}$ (next), $\mathbf{U}$ (until),
{\bf G} (always) and {\bf F} (eventually) (e.g., see \cite{Eme90}). Their semantics is reviewed in Appendix \ref{app:ltl}. 
% \yuliang{This needs to be removed. ``Their semantics is reviewed in the full version.''}
An extension of LTL in which propositions are interpreted as FO sentences has previously been defined 
to specify properties of sequences of structures \cite{S00}, and in particular of runs of artifact systems \cite{DHPV:ICDT:09,tods12}.
The extension is denoted by $\ltlfo$. In order to specify properties of HAS's, we shall use a variant of $\ltlfo$, called  
{\em hierarchical} $\ltlfo$, denoted $\hltlfo$.
Intuitively, an $\hltlfo$ formula uses as building blocks $\ltlfo$ formulas acting on local runs of individual tasks, 
referring only to the database and local data, and can recursively state $\hltlfo$ properties on runs resulting from calls to children tasks.
This closely mirrors the hierarchical execution of tasks, and is a natural fit for this computation model.  
In addition to its naturaleness, the choice of $\hltlfo$ has several technical justifications. First, verification of 
$\ltlfo$ (and even LTL) properties is not possible for HAS's. 
\vspace*{-1mm}
\begin{theorem} \label{thm:ltlfo}
It is undecidable, given an $\ltlfo$ formula $\varphi$ and a HAS $\Gamma =  \langle \cala, \Sigma, \Pi \rangle$, 
whether $\Gamma \models \varphi$.
Moreover, this holds even for LTL formulas over $\Sigma$ (restricting the sequence of services in a global run).
\end{theorem}

The proof, provided in Appendix \ref{app:ltlfo}, is by reduction from repeated state reachability in VASS 
with resets and bounded lossiness, whose undecidability follows from \cite{mayr2003undecidable}. \yuliang{This needs to be removed.}

Another technical argument in favor of $\hltlfo$ is that it only expresses
properties that are invariant under interleavings of independent tasks.
Interleaving invariance is not only a natural soundness condition, but also
allows more efficient model checking by {\em partial-order reduction} \cite{peled94}. 
Moreover, $\hltlfo$ enjoys a pleasing completeness property: it expresses, in a reasonable sense, 
{\em all } interleaving-invariant $\ltlfo$ properties of HAS's.
The proof is non-trivial, building on completeness results 
for propositional temporal logics on Mazurkiewicz traces \cite{gastin04,gastin06} (see Appendix \ref{sec:complete}). \yuliang{This also needs to be fixed.}

We next define $\hltlfo$. 
Propositions in $\hltlfo$ are interpreted as conditions\footnote{For consistency with previous notation, we denote the logic HLTL-FO
although the FO interpretations are restricted to be quantifier free.}
on artifact instances
in the run, or recursively as $\hltlfo$ formulas on runs of invoked children tasks. 
The different conditions may share some universally quantified global variables. 

\vspace{-2mm}
\begin{definition}
Let $\Gamma = \langle \cala, \Sigma, \Pi \rangle$
be an artifact system where $\cala = \langle \calh, \mathcal{DB} \rangle$. 
Let $\bar y$ be a finite sequence of variables in $\varid \cup \varnum$ disjoint from $\{\bar{x}^T \mid T \in \calh \}$, called {\em global variables}.
We first define recursively the set $\Psi(T, \bar y)$ of {\em basic $\hltlfo$ formulas} with global variables $\bar y$,
for each task $T \in \calh$. The set $\Psi(T, \bar y)$ consists of all formulas $\varphi_f$ obtained as follows:
\vspace*{-1mm}
\begin{itemize}\itemsep=0pt\parskip=0pt
\item $\varphi$ is an LTL formula with propositions $P \cup \Sigma_T^{obs}$
where $P$ is a finite set of proposition disjoint from $\Sigma$;
\item Let $\Phi$ be the set of conditions on $\bar{x}_T \cup \bar y$ extended by allowing atoms of the form
$S^T(\bar z)$ in which all variables in $\bar z$ are in $\bar y \cap \varid$;
$f$ is a function from $P$ to\footnote{$[\psi]_{T_c}$ is an expression whose meaning is explained below.} $\Phi \cup \{[\psi]_{T_c} \mid \psi \in \Psi (T_c, \bar y), T_c \in child(T)\}$;
\item $\varphi_f$ is obtained by replacing each $p \in P$ with $f(p)$;
\end{itemize}
An $\hltlfo$ formula over $\cala$ is an expression $\forall \bar y [\varphi_f]_{T_1}$ where $\varphi_f$ is in $\Psi(T_1, \bar y)$.
\end{definition}

In an $\hltlfo$ formula of task $T$, each proposition is mapped to either a quantifier-free FO formula referring 
to the variables and set of task $T$, or an 
$\hltlfo$ formula of a child task of $T$. The intuition is the following. A proposition mapped to a
quantifier-free FO formula holds in a given configuration of $T$ if the formula is true
in that configuration. A proposition mapped to an expression $[\psi]_{T_c}$ holds 
in a given configuration if $T$ makes a call to $T_c$ and the run of $T_c$ resulting from the call
satisfies $\psi$. 

\vspace*{-2mm}
\begin{example} Let $T_1$ be a root task with child tasks $T_2$ and $T_3$.
The $\hltlfo$ formula (with no global variables)
$$
\varphi = [ 
~\mathbf{F}[\psi_2]_{T_2} 
\rightarrow 
\mathbf{G} (\sigma^o_{T_3} \rightarrow [\psi_3]_{T_3})
]_{T_1}
$$
states that whenever $T_1$ calls child task $T_2$ 
and $T_2$'s local run satisfies property $\psi_2$, then if $T_3$ is also called (via the opening service $\sigma^o_{T_3}$), its local run
must satisfy property $\psi_3$. 
\end{example}
\vspace*{-2mm}
See Appendix \ref{app:example-hltl} for a concrete $\hltlfo$ property of similar structure,
in the context of our example for the HAS model.

Since $\hltlfo$ properties depend on local runs of tasks and their relationship to local runs of their descendants, their semantics is naturally defined using the full trees of local runs.
We first define satisfaction by a local run of $\hltlfo$ formulas with no global variables. 
This is done recursively. Let \Tree be a full tree of local runs of $\Gamma$ over some database $D$.
Let $\varphi_f$ be a formula in $\Psi(T, \langle \rangle)$ (no global variables).
Recall that $\varphi$ is a propositional LTL formula over $P \cup \Sigma_T^{obs}$. Let 
$\rho_T = (\nu_{in}, \nu_{out}, \{(I_i,\sigma_i)\}_{i < \gamma})$ be a local run of $T$ in $\Tree$.
A proposition $\sigma \in \Sigma_T^{obs}$ holds in $(I_j, \sigma_{j})$ if $\sigma = \sigma_j$.  
Consider $p \in P$ and $f(p)$. If $f(p)$
is an FO formula, the standard definition applies. If $f(p) = [\psi]_{T_c}$, then $(I_j, \sigma_{j})$ satisfies
$[\psi]_{T_c}$ iff $\sigma_j = \sigma^0_{T_c}$ and the local run of $T_c$ connected to $\rho_T$ in \Tree by an edge labeled $j$  
satisfies $\psi$. The formula $\varphi_f$ is satisfied if the sequence of truth values of its propositions via $f$
satisfies $\varphi$. Note that $\rho_T$ may be finite, in which case a finite variant of the LTL semantics is used \cite{ltl-finite}
(see Appendix \ref{app:ltl}). \yuliang{Needs to be fixed.}

A full tree of local runs satisfies $\varphi_f \in \Psi(T_1, \langle \rangle)$ if
its root (a local run of $T_1$) satisfies $\varphi_f$. 
Finally, let $\varphi_f(\bar y)$ be a formula in $\Psi(T_1, \bar y)$.
Then $\forall \bar y [\varphi_f(\bar y)]_{T_1}$ is satisfied by $\Tree$, denoted $\Tree \models \forall \bar y [\varphi_f(\bar y)]_{T_1}$,
if for every valuation $\nu$ of $\bar y$, 
$\Tree$ satisfies $\varphi_{f^\nu}$ where $f^\nu$ is obtained from $f$ by replacing each $y$ in $f(p)$ by $\nu(y)$ 
for every $p \in P$.  Note that $\varphi_{f^\nu} \in \Psi(T_1, \langle \rangle)$. 
Finally, $\Gamma$ satisfies $\forall \bar y [\varphi_f(\bar y)]_{T_1}$, denoted $\Gamma \models \forall \bar y [\varphi_f(\bar y)]_{T_1}$,
if $\Tree \models \forall \bar y [\varphi_f(\bar y)]_{T_1}$ for every database instance $D$ 
and tree of local runs $\Tree$ of $\Gamma$ on $D$. % \yuliang{added database instance}

The semantics of $\hltlfo$ on trees of local runs of a HAS also induces a semantics on the global runs of the HAS. 
Let $\forall \bar y [\varphi_f(\bar y)]_{T_1}$ be an $\hltlfo$ formula and $\rho \in \call(\Tree)$, where  $\Tree$
is a full tree of local runs of $\Gamma$. We say that $\rho$ satisfies $\forall \bar y [\varphi_f(\bar y)]_{T_1}$ if
$\Tree$ satisfies $\forall \bar y [\varphi_f(\bar y)]_{T_1}$. This is well defined in view of the following
easily shown fact: if $\rho \in \call(\Tree_1) \cap \call(\Tree_2)$ then $\Tree_1 = \Tree_2$.

\vspace{2mm}
\noindent
{\bf Simplifications}~
Before proceeding, we note that several simplifications to $\hltlfo$ formulas and HAS specifications
can be made without impact on verification. First,
although useful at the surface syntax, the global variables, as well as set atoms,
can be easily eliminated from the $\hltlfo$ formula to be verified (Lemma \ref{lem:noglobal} in Appendix \ref{app:simplification}).
It is also useful to note that one can assume, without loss of generality,
two simplifications on artifact systems regarding the interaction of tasks with their subtasks:
(i) for every task $T$, the set of variables passed to subtasks is disjoint with the set of variables returned by subtasks, and
(ii) all variables returned by subtasks are non-numeric (Lemma \ref{lem:simplification} in Appendix \ref{app:simplification}).
In view of the above, we henceforth consider only properties with no global variables or set atoms,
and artifact systems simplified as described. \yuliang{Needs to be fixed.}

\vspace{4mm}
\noindent
{\bf Checking $\hltlfo$ Properties Using Automata}\\
We next show how to check $\hltlfo$ properties of trees of local runs of artifact systems. 
Before we do so, recall the standard construction of a B\"{u}chi automaton $B_\varphi$
corresponding to an LTL formula $\varphi$ \cite{VW:LICS:86,SVW87}.
The automaton $B_\varphi$ has exponentially many states and accepts 
precisely the set of $\omega$-words that satisfy $\varphi$. 
Recall that we are interested in evaluating LTL formulas $\varphi$ on both infinite {\em and} finite runs.
It is easily seen that for the $B_\varphi$ obtained by the standard construction 
there is a subset $Q^{\emph{fin}}$ of its states such that $B_\varphi$
viewed as a finite-state automaton with final states $Q^{\emph{fin}}$ 
accepts precisely the finite words that satisfy $\varphi$ (details omitted).

Consider now an artifact system $\Gamma$ and let $\varphi = [\xi]_{T_1}$ be
an $\hltlfo$ formula over $\Gamma$.  Consider a full tree \Tree of local runs.
For task $T$, denote by $\Phi_T$ the set of sub-formulas $[\psi]_{_T}$ occurring in
$\varphi$ and by $2^{\Phi_T}$ the set of truth assignments to these formulas. 
For each $T$ and $\eta \in 2^{\Phi_T}$, 
let $B(T, \eta)$ be the \buchi \ automaton
constructed from the formula
$$\left( \wedge_{\psi \in \Phi_T, \eta(\psi)=1} ~\psi \right) \land 
\left( \wedge_{\psi \in \Phi_T, \eta(\psi)=0} ~\neg \psi \right)$$
and define $\calb_\varphi = \{B(T,\eta) \mid T \in \calh, \eta \in 2^{\Phi_T}\}$.

We now define acceptance of \Tree by $\calb_\varphi$.
An {\em adornment} of \Tree is a mapping $\alpha$ associating to each edge from $\rho_T$ to $\rho_{T_c}$
a truth assignment in $2^{\Phi_{T_c}}$.  \Tree is accepted by $\calb_\varphi$
if there exists an adornment $\alpha$ such that:
\vspace*{-1mm}
\begin{itemize}\itemsep=0pt\parskip=0pt
\item for each local run $\rho_T$ of $T$ with no outgoing edge and incoming edge with adornment $\eta$,
$\rho_T$ is accepted by $B(T,\eta)$
\item for each local run $\rho_T$ of $T$ with incoming edge labeled by $\eta$,  $\alpha(\rho_T)$ is accepted by
$B(T, \eta)$, where $\alpha(\rho_T)$ extends $\rho_T$ by assigning to each configuration $(\rho_j, \sigma^o_{T_c})$ the truth assignment
in $2^{\Phi_{T_c}}$ adorning its outgoing edge labeled $j$. (Recall that in configurations $(I_j, \sigma_j)$ for which
$\sigma_j \neq \sigma^o_{T_c}$, all formulas in $\Phi_{T_c}$ are {\em false} by definition.) 
\item $\alpha(\rho_{T_1})$ is accepted by the \buchi\ automaton $B_\xi$ where $\alpha(\rho_{T_1})$ is defined as above.
\end{itemize}    

The following can be shown.
\vspace*{-1mm}
\begin{lemma} \label{lem:proj}
A full tree of local runs \Tree satisfies $\varphi = [\xi]_{T_1}$ iff \Tree is accepted by $\calb_\varphi$.
\end{lemma}

\section{Verification Without Arithmetic} \label{sec:verification}

\reviewer{
\begin{itemize}
\item I was expecting some words on lower bounds, which could answer the following questions: Is it possible to devise a different decision procedure for this verification problem with a better complexity? Is there a milder restriction (for example lifting the restriction imposed on logic) that would still give decidability?
\item Since the specification is already interleaving invariant, can the verification problem be reduced to verification of non-hierarchical artifact systems?
\item Artifact systems with data dependencies and arithmetic was studies by the same authors in [18,19].  How does the present framework compare with [18,19] in terms of expressiveness?
\item connection with DCDS
\item Explain the role of the schema restrictions (unary keys etc.) My impression (from reading C.3) is that the need for counters on isomorphism types is still there even under your most drastic restriction, acyclicity. If this is the case, my suggestion is that you reverse your current ordering and give the argument for acyclic schemas only in the body. This would allow the reader to focus on the VASS construction without having to worry about the more complex data abstraction used to get rid of cyclicity.
\item Can you simulate a VASS using your machines?
\end{itemize}}

\newcommand{\tauin}{\tau_{in}}
\newcommand{\tauout}{\tau_{out}}
\newcommand{\nuin}{\nu_{in}}
\newcommand{\nuout}{\nu_{out}}
\newcommand{\II}{\{(I_i,\sigma_i)\}_{0 \leq i < \gamma}}
\newcommand{\IIw}{\{(I_i,\sigma_i)\}_{0 \leq i < \omega}}
\newcommand{\RS}{\{(\rho_i, \sigma_i)\}_{0 \leq i < \gamma}}
\newcommand{\RSw}{\{(\rho_i, \sigma_i)\}_{0 \leq i < \omega}}
\newcommand{\trt}{\tilde{\rho}_T}
\newcommand{\trtc}{\tilde{\rho}_{T_c}}
\newcommand{\trtp}{\tilde{\rho}_{T_p}}
\newcommand{\ID}{\emph{ID}}

\newcommand{\Ret}{\ensuremath{\texttt{Retrieve} } }
\newcommand{\Prop}{\ensuremath{\texttt{Prop} } }
\newcommand{\merge}{\ensuremath{\texttt{merge} } }
\newcommand{\Add}{\ensuremath{\texttt{Add} } }
\newcommand{\reach}{\ensuremath{\texttt{Reach} } }
\newcommand{\adom}{\ensuremath{\texttt{adom} } }
\newcommand{\vmap}{\ensuremath{\texttt{vmap} } }
\newcommand{\connect}{\ensuremath{\texttt{connect} } }
\newcommand{\type}{\ensuremath{\texttt{type} } }

\newcommand{\e}{\epsilon}
\newcommand{\GP}{\emph{GP}}
\newcommand{\hexp}{\text{-}\exp}

In this section we consider verification for the case when the artifact system and the $\hltlfo$ 
property have no arithmetic constraints. We show in Section \ref{sec:arithmetic} how our approach can be extended
when arithmetic is present.

The roadmap to verification is the following. Let $\Gamma$ be a HAS and $\varphi = [\xi]_{T_1}$ an $\hltlfo$ formula over $\Gamma$.
To verify that every tree of local runs of $\Gamma$ satisfies $\varphi$, we check that there is no tree of local runs 
satisfying $\neg \varphi = [\neg \xi]_{T_1}$, or equivalently, accepted by $\calb_{\neg\varphi}$. 
Since there are infinitely many trees of local runs of $\Gamma$ due so the unbounded data domain, and each tree can be infinite,
an exhaustive search is impossible.  We address this problem by developing a symbolic representation 
of trees of local runs, called {\em symbolic tree of runs}. 
The symbolic representation is subtle for several reasons. First, unlike the representations in \cite{DHPV:ICDT:09,tods12}, 
it is not finite state. This is because summarizing the relevant information about artifact relations requires keeping track of the number of tuples of
various isomorphism types. Second, the symbolic representation does not capture the full information about the actual runs, 
but just enough for verification. Specifically, we show that for every $\hltlfo$ formula $\varphi$,
there exists a tree of local runs accepted by $\calb_\varphi$ iff there exists a symbolic tree of runs accepted by $\calb_\varphi$.  
We then develop an algorithm to check the latter.
The algorithm relies on reductions to state reachability problems in Vector Addition Systems with States (VASS) \cite{vass}.

\victor{Discussion of flat simulation}
One might wonder whether there is a simpler approach to verification of HAS, that reduces it to verification of a flat system
(consisting of a single task). This could indeed be done in the absence of artifact relations, by essentially 
concatenating the artifact tuples of the tasks along the hierarchy that are active at any given time, 
and simulating all transitions by internal services. 
However, there is strong evidence that this is no longer possible when tasks are equipped with artifact relations.
First, a naive simulation using a single artifact relation would require 
more powerful updating capabilities than available in the model.  
Moreover, Theorem \ref{thm:ltlfo} shows that LTL is undecidable for hierarchical systems, 
whereas the results in this section imply that it is decidable for flat ones (as it coincides with HLTL for single tasks). 
While this does not rule out a simulation, 
it shows that there can be no effective simulation natural enough to be
extensible to LTL properties. A reduction to the model of \cite{tods12} is even less plausible, because of the lack of
artifact relations.  Note that, even if a reduction were possible, the results of \cite{tods12} would be of no help in obtaining our
lower complexities for verification, since the algorithm
provided there is non-elementary in all cases.

We next embark upon the development outlined above.

\subsection{Symbolic Representation} \label{sec:symrep}

We begin by defining the symbolic analog of a local run, called {\em local symbolic run}.
The symbolic tree of runs is obtained by connecting the local symbolic runs similarly to the way
local runs are connected in trees of local runs.

Each local symbolic run is a sequence of symbolic representations 
of an actual instance within a local run of a task $T$. The representation has the following ingredients:
\vspace*{-1mm}
\begin{enumerate}\itemsep=0pt\parskip=0pt  
\item the equality type of the artifact variables of $T$ 
and the elements in the database reachable from them by navigating foreign keys up to a specified depth $h(T)$. 
This is called the \emph{$T$-isomorphism type} of the variables.
\item the $T$-isomorphism type of the input and return variables (if representing a returning local run)
\item for each $T$-isomorphism type of the set variables of $T$ together with the input variables, the net number of insertions 
of tuples of that type in $S^T$. 
\end{enumerate}\vspace*{-1mm}

Intuitively, (1) and (2) are needed in order to ensure that the assumptions made about the database while navigating via foreign keys
in tasks and their subtasks are consistent. The depth $h(T)$ is chosen to be sufficiently large to ensure the consistency. 
(3) is required in order to make sure that a retrieval from $S^T$ of a tuple with a given $T$-isomorphism type is allowed only when 
sufficiently many tuples of that type have been inserted in $S^T$. 

We now formally define the symbolic representation, starting with $T$-isomorphism type.
Let $\bar x^T$ be the variables of $T$.
We define $h(T)$ as as follows.
Let FK be the foreign key graph of the schema $\db$ and $F(n)$ be the maximum number of distinct paths of length at most $n$
starting from any relation $R$ in FK. 
Let $h(T) = 1 + |\bar x^T| \cdot F(\delta)$ where $\delta = 1$ if $T$ is a leaf task and $\delta = \max_{T_c \in child(T)} h(T_c)$
otherwise. 
\reviewer{this is only discussed in the appendix.}
\victor{Removed the sentence.}

We next define expressions that denote navigation via foreign keys starting from the set of id variables $\bar x^T_{id}$ of $T$.
\yuliang{added reminder for $\bar{x}^T_{id}$} 
For each $x \in \bar x^T_{id}$ and $R \in \db$, let $x_R$ be a new symbol.
An expression is a sequence $\xi_1.\xi_2.\ldots\xi_m$, $\xi_1 = x_R$ for some  $x \in \bar x^T_{id}$
and $R \in \db$, $\xi_j$ is a foreign key in some relation of $\db$ for $2 \leq j < m$, $\xi_m$ is a foreign key or a numeric attribute, 
$\xi_2$ is an attribute of $R$, and for each $i$, $2 < i \leq m$, if $\xi_{i-1}$ is a foreign key referencing $Q$ then $\xi_i$ is
an attribute of $Q$. We define the length of $\xi_1.\xi_2.\ldots\xi_m$ as $m$.
A {\em navigation set} $\cale_T$ is a set of expressions such that:
\vspace*{-1mm}
\begin{itemize}\itemsep=0pt\parskip=0pt
\item for each $x \in \bar x^T_{id}$ there is at most one $R \in \db$ for which the expression $x_R$ is in $\cale_T$;
\item every expression in $\cale_T$ is of the form $x_R.w$ where $x_R \in \cale_T$, and has length $\leq h(T)$;
\item if $e \in \cale_T$ then every expression $e.s$ of length $\leq h(T)$ extending $e$ is also in $\cale_T$.
\end{itemize}
\vspace*{-1mm}

Note that $\cale_T$ is closed under prefix. \reviewer{prefix or suffix?}  \yuliang{I think it should be prefix.}
We can now define $T$-isomorphism type.
Let $\cale_T^+ = \cale_T \cup \bar x^T \cup \{\anull, 0\}$.
The {\em sort} of $e \in \cale_T^+$ is numeric if $e \in \bar x^T_{\mathbb{R}} \cup \{0\}$ or $e = w.a$ where $a$ is a numeric attribute;
its sort is $\anull$ if $e = \anull$ or $e = x \in \bar x^T_{id}$ and $x_R \not\in \cale_T$ for all $R \in \db$; and its sort is
ID$(R)$ for $R \in \db$ if $e = x_R$, or $e = x \in \bar x^T_{id}$ and $x_R \in \cale_T$, or $e = w.f$ where $f$ is a 
foreign key referencing $R$.

\vspace{-2mm}
\begin{definition}
A $T$-isomorphism type $\tau$ consists of a navigation set $\cale_T$ together with an equivalence relation 
$\sim_\tau$ over $\cale_T^+$ such that:
\vspace*{-1mm}
\begin{itemize}\itemsep=0pt\parskip=0pt
\item if $e \sim_\tau f$ then $e$ and $f$ are of the same sort;
\item for every $\{x, x_R\} \subseteq \cale_T^+$, $x \sim_\tau x_R$; 
\item for every $e$ of sort $\anull$, $e \sim_\tau \anull$;
\item if $u \sim_\tau v$ and $u.f, v.f \in \cale_T$ then $u.f \sim_\tau v.f$.
\end{itemize}
\end{definition}
\vspace*{-2mm}

We call an equivalence relation $\sim_\tau$ as above an {\em equality type} for $\tau$.
The relation $\sim_\tau$ is extended to tuples componentwise.

Note that $\tau$ provides enough information to evaluate conditions over $\bar x^T$.
Satisfaction of a condition $\varphi$ by an isomorphism type $\tau$, denoted $\tau \models \varphi$, is defined as follows:
\vspace*{-1mm}
\begin{itemize}\itemsep=0pt\parskip=0pt
\item $x = y$ holds in $\tau$ iff $x \sim_\tau y$, 
\item $R(x, y_1, \dots, y_m, z_1, \dots, z_n)$ holds in $\tau$ for relation
$R(id, a_1, \\ \dots, a_m, f_1, \dots, f_n)$ iff $\{x_R.a_1, \dots, x_R.a_m, x_R.f_1, \dots, \\ x_R.f_n\} \subseteq \cale_T$, and
$(y_1, \dots, y_m, z_1, \dots, z_m) \sim_\tau (x_R.a_1, \dots, \\ x_R.a_m, x_R.f_1, \dots, x_R.f_n)$
\item Boolean combinations of conditions are standard.
\end{itemize}
\vspace*{-1mm}

Let $\tau$ be a $T$-isomorphism type with navigation set $\cale_T$ and equality type $\sim_\tau$.
The projection of $\tau$ onto a subset of variables $\bar{z}$ of $\bar x^T$ is defined
as follows. Let $\cale_T | \bar{z} = \{x_R.e \in \cale_T | x \in \bar{z} \}$ and
$\sim_\tau|\bar{z}$ be the projection of $\sim_\tau$ onto $\bar z \cup \cale_T |\bar{z} \cup \{\anull, 0\}$. 
The projection of $\tau$ onto $\bar{z}$, denoted as $\tau |\bar{z}$, is a $T$-isomorphism type with navigation set
$\cale_T | \bar{z}$ and equality type $\sim_\tau | \bar{z}$. Furthermore, the projection of
$T$-isomorphism onto $\bar{z}$ upto length $k$, denoted as $\tau | (\bar{z}, k)$, is defined as
$\tau | \bar{z}$ with all expressions in $\cale_T | \bar{z}$ with length more than $k$ removed.

We apply variable renaming to isomorphism types as follows.
Let $f$ be a 1-1 partial mapping from $\bar x^T$ to $\varid \cup \varnum$ such that $f(\bar x^T_{id}) \subseteq \varid$,
$f(\bar x^T_{\mathbb{R}}) \subseteq \varnum$ and $f(\bar x^T) \cap \bar x^T = \emptyset$.
For a $T$-isomorphism type $\tau$ with navigation set $\cale_T$, 
$f(\tau)$ is the isomorphism type obtained as follows. Its navigation set is obtained by replacing in $\cale_T$ each variable
$x$ and $x_R$ in $\cale_T$ with $f(x)$ and $f(x)_R$, for $x \in dom(f)$. The relation $\sim_{f(\tau)}$ is the image of 
$\sim_\tau$ under the same substitution.

As seen above, a $T$-isomorphism type captures all information needed
to evaluate a condition on $\bar x_T$. However, the set $S^T$ can contain unboundedly
many tuples, which cannot be represented by a finite equality type.
This is handled by keeping a set of counters for projections of $T$-isomorphism types on the variables relevant to $S^T$,
that is, $(\bar{x}^T_{in} \cup \bar{s}^T)$.
We refer to the projection of a $T$-isomorphism type onto $(\bar{x}^T_{in} \cup \bar{s}^T)$ as a $TS$-isomorphism type,
and denote by $\ts$ the set of $TS$-isomorphism types of $T$.
We will use counters to record the number of tuples in $S^T$ of each $TS$-isomorphism type. 
\reviewer{The juxtaposition of these
sentences is confusing, since the first is a definition which does not include
any counters. The counters are a forward pointer to the notion that is about
to be defined.}
\victor{Changed slightly.}

We can now define symbolic instances.

\vspace{-2mm}
\begin{definition}
A symbolic instance $I$ of task $T$ is a tuple $(\tau, \bar{c})$
where $\tau$ is a $T$-isomorphism type and $\bar{c}$ is a vector of integers where
each dimension of $\bar{c}$ corresponds to a $TS$-isomorphism type. 
\end{definition}

We denote by $\bar{c}(\hat{\tau})$ the value of the dimension of $\bar c$ corresponding to the $TS$-isomorphism type $\hat{\tau}$
and by $\bar{c}[\hat{\tau} \mapsto a]$ the vector obtained from $\bar{c}$
by replacing $\bar c(\hat{\tau})$ with $a$.

\vspace{-2mm}
\begin{definition}
A {\em local symbolic run} $\tilde{\rho}_T$ of task $T$ is a tuple $(\tau_{in}, \tau_{out}, \{(I_i, \sigma_i)\}_{0 \leq i < \gamma})$, where:
\vspace*{-1mm}
\begin{itemize}\itemsep=0pt\parskip=0pt
\item each $I_i$ is a symbolic instance $(\tau_i, \bar{c}_i)$ of $T$
\item each $\sigma_i$ is a service in $\Sigma_T^{obs}$
\item $\gamma \in \mathbb{N} \cup \{\omega\}$ 
(if $\gamma = \omega$ then $\tilde{\rho}_T$ is infinite, otherwise it is finite)
\item $\tau_{in}$, called the input isomorphism type, is a $T$-isomorphism type projected to $\bar{x}^T_{in}$. And $\tauin \models \Pi$ if $T = T_1$.
\yuliang{added the global precondition}
\item at the first instance $I_0$, $\tau_0 | \bar{x}^T_{in} = \tau_{in}$,
for every $x \in \bar{x}^{T}_{id} - \bar{x}^{T}_{in}$, $x \sim_{\tau_0} \anull$, 
and for every $x \in \bar{x}^{T}_{\mathbb{R}} - \bar{x}^{T}_{in}$, $x \sim_{\tau_0} 0$.
Also $\bar{c}_0 = \bar{0}$ and $\sigma_0 = \sigma_T^o$. 
\item if for some $i$, $\sigma_i = \sigma^c_T$ then $\tilde{\rho}_T$ is finite and $i = \gamma-1$ (and $\tilde{\rho}_T$ is called a {\em returning} run)
\item $\tau_{out}$ is $\bot$ if $\tilde{\rho}_T$ is infinite or finite but $\sigma_{\gamma-1} \neq \sigma^c_T$,
and it is $\tau_{\gamma-1} | (\bar{x}^T_{in} \cup \bar{x}^T_{ret})$ otherwise
\item a \emph{segment} of $\tilde{\rho}_T$ is a subsequence 
$\{(I_i,\sigma_i)\}_{i \in J}$, where $J$ is a maximal interval $[a, b] \subseteq \{i \mid 0 \leq i < \gamma\}$
such that no $\sigma_j$ is an internal service of $T$ for $j \in [a + 1, b]$. A segment $J$ 
is {\em terminal} if $\gamma \in \mathbb{N}$ and
$b = \gamma - 1$. Segments of $\tilde{\rho}_T$ must satisfy the following  
properties. For each child $T_c$ of $T$ there is at most one $i \in J$ such that $\sigma_i = \sigma^o_{T_c}$.
If $J$ is not terminal and such $i$ exists, there is exactly one $j \in J$ for which $\sigma_j = \sigma^c_{T_c}$, and $j > i$.   
If $J$ is terminal, there is {\em at most} one such $j$.
\item for every $0 < i < \gamma$, $I_{i}$ is a {\bf successor} \reviewer{highlight successor?} of $I_{i-1}$ under $\sigma_i$ (see below). 
\end{itemize}
\end{definition}
\vspace*{-2mm}

The successor relation is defined next.
We begin with some preliminary definitions.
A $TS$-isomorphism type $\hat{\tau}$ is \emph{input-bound} if for every $s \in \bar{s}^T$, 
$s \not\sim_{\hat{\tau}} \anull$ implies that there exists an expression $x_R.w$ in $\hat{\tau}$
such that $x \in \bar{x}^T_{in}$ and $x_R.w \sim_{\hat{\tau}} s$.
We denote by $\tsib$ the set of input-bound types in $\ts$.
For $\hat{\tau},\hat{\tau}' \in \ts$, update $\delta$ of the form $\{+S^T(\bar{s}^T)\}$ or
$\{-S^T(\bar{s}^T)\}$ and mapping $\bar{c}_{ib}$ from $\tsib$ to $\{0,1\}$, we define 
the mapping $\bar{a}(\delta, \hat{\tau}, \hat{\tau}', \bar{c}_{ib})$ from $\ts$ to $\{-1,0,1\}$ as follows
($\bar a_0$ is the mapping sending $\ts$ to $0$): 
\vspace*{-1mm}
\begin{itemize}\itemsep=0pt\parskip=0pt
\item if $\delta = \{+S^T(\bar{s}^T)\}$, then
$\bar{a}(\delta, \hat{\tau}, \hat{\tau}', \bar{c}_{ib})$ is $\bar a_0[\hat{\tau} \mapsto 1]$ 
if $\hat{\tau}$ is not input-bound, and   
$\bar{a}_0[\hat{\tau} \mapsto (1 - \bar{c}_{ib}(\hat{\tau}))]$ otherwise
\item if $\delta = \{-S^T(\bar{s}^T)\}$, then 
$\bar{a}(\delta, \hat{\tau}, \hat{\tau}', \bar{c}_{ib}) = \bar{a}_0[\hat{\tau}' \mapsto -1]$
\item if $\delta$ is $\{+S^T(\bar{s}^T),-S^T(\bar{s}^T)\}$  then 
\vspace{-1mm}
$$\bar{a}(\delta, \hat{\tau}, \hat{\tau}', \bar{c}_{ib}) = \bar{a}(\delta^+, \hat{\tau}, \hat{\tau}', \bar{c}_{ib}) + \bar{a}(\delta^-, \hat{\tau}, \hat{\tau}', \bar{c}_{ib})$$ \vspace{-1mm}
where $\delta^+ = \{+S^T(\bar{s}^T)\}$ and $\delta^- = \{-S^T(\bar{s}^T)\}$.
\end{itemize}
Intuitively, the vector $\bar{a}(\delta, \hat{\tau}, \hat{\tau}', \bar{c}_{ib})$ specifies how the current counters need to be modified
to reflect the update $\delta$. The input-bound $TS$-isomorphism types require special handling because consecutive insertions 
necessarily collide so the counter's value cannot go beyond $1$.

For symbolic instances $I = (\tau, \bar{c})$ 
and $I' = (\tau', \bar{c}')$, $I'$ is a successor of $I$
by applying service $\sigma'$ iff:
\vspace*{-1mm}
\begin{itemize}\itemsep=0pt\parskip=0pt
\item If $\sigma'$ is an internal service $\langle \pi, \psi, \delta \rangle$, 
then for $\hat{\tau} = \tau|(\bar{x}_{in}^T \cup \bar{s}^T)$ and $\hat{\tau}' = \tau'|(\bar{x}_{in}^T \cup \bar{s}^T)$, 
  \vspace*{-1mm}
  \begin{itemize}\itemsep=0pt\parskip=0pt
  \item $\tau | \bar{x}^T_{in} = \tau' | \bar{x}^T_{in}$,
  \item $\tau \models \pi$ and $\tau' \models \psi$,
  \item $\bar{c}' \geq \bar{0}$ and $\bar{c}' = \bar{c} + \bar{a}(\delta, \hat{\tau}, \hat{\tau}', \bar{c}_{ib})$, 
   where $\bar{c}_{ib}$ the restriction of $\bar{c}$ to $\tsib$.
  \end{itemize}
  \vspace*{-1mm}
\item If $\sigma'$ is an opening service $\langle \pi, f_{in} \rangle$ of 
  subtask $T_c$, then $\tau = \tau' \models \pi$ and $\bar{c}' = \bar{c}$.
\item If $\sigma'$ is a closing service of subtask $T_c$, then for
  $\bar{x}^T_{const} = \bar{x}^T - \{x \in \xttcup | x \sim_\tau \anull \}$, 
  $\tau' | \bar{x}^T_{const} = \tau  | \bar{x}^T_{const}$ and $\bar{c}' = \bar{c}$.
\item If $\sigma'$ is the closing service $\sigma^c_T = \langle \pi, f_{out} \rangle$ of $T$, then $\tau \models \pi$ and $(\tau, \bar{c}) = (\tau', \bar{c}')$. 
\end{itemize}
\vspace*{-1mm}

Note that there is a subtle mismatch between transitions in actual local runs and in symbolic runs.
In the symbolic transitions defined above, a service inserting a tuple in $S^T$ {\em always} causes the correspoding counter to increase 
(except for the input-bound case).
However, in actual runs, an inserted tuple may collide with an already existing tuple in the set, 
in which case the number of tuples does {\em not} increase.
Symbolic runs do not account for such collisions (beyond the input-bound case), which raises the danger that they might overestimate the number of available tuples and
allow impossible retrievals. Fortunately, the proof of Theorem \ref{thm:actual-symbolic} shows that collisions can be ignored 
at no peril. More specifically, it follows from the proof that 
for every actual local run with collisions satisfying an $\hltlfo$ property there exists 
an actual local run without collisions that satisfies the same property. The intuition is the following.
First, given an actual run with collisions, one can modify it so that only new tuples are inserted in the artifact relation,
thus avoiding collisions. However, this raises a challenge, since it may require augmenting the database with new tuples.
If done naively, this could result in an infinite database.
The more subtle observation, detailed in the proof of Theorem \ref{thm:actual-symbolic}, is that
only a bounded number of new tuples must be created, thus keeping the database finite.
\reviewer{a one line intuition behind why collisions do not matter is welcome}
\victor{Added a brief explanation, please read.}

\vspace{-2mm}
\begin{definition}
A {\em symbolic tree of runs} is a directed labeled tree $\Sym$ in which each node
is a local symbolic run $\tilde{\rho}_T$ for some task $T$, and every edge connects a local symbolic run of 
a task $T$ with a local symbolic run of a child task $T_c$ and is labeled with a non-negative integer $i$
(denoted $i(\tilde{\rho}_{T_c})$). In addition, the following properties are satisfied.
Let $\tilde{\rho}_T = (\tau_{in}, \tau_{out}, \{(I_i,\sigma_i)\}_{0 \leq i < \gamma})$
be a node of $\Sym$. 
Let $i$ be such that $\sigma_i = \sigma^o_{T_c}$ for some child $T_c$ of $T$. There exists a unique edge
labeled $i$ from $\tilde{\rho}_T$ to a node $\tilde{\rho}_{T_c} = (\tau'_{in}, \tau'_{out}, \\ \{(I'_i,\sigma'_i)\}_{0 \leq i < \gamma'})$ 
of $\Sym$, and the following hold:
\vspace*{-1mm}
\begin{itemize}\itemsep=0pt\parskip=0pt
\item $\tau'_{in} = f^{-1}_{in}(\tau_i) | (\bar{x}^{T_c}_{in}, h(T_c))$ 
where $f_{in}$ is the input variable mapping of $\sigma_{T_c}^o$
\item  $\tilde{\rho}_{T_c}$ is a returning run iff there exists $j > i$ such that $\sigma_j = \sigma^c_{T_c}$;
let $k$ be the minimum such $j$. 
Let $\bar{x}_r = \xttcdown$ and
$\bar{x}_{w} = \{x | x \in \xttcup, x \sim_{\tau_{k-1}} \anull \}$. 
Then $\tau_k | (\bar{x}_r \cup \bar{x}_w, h(T_c)) = ((f_{in} \circ f^{-1}_{out})(\tau_{out})) | (\bar{x}_r \cup \bar{x}_w)$
%Then $\tau_k | (\bar{x}_w \cup \bar{x}_r, h(T_c)) = ((f_{in}^{-1} \circ f_{out})(\tau_{out})) | (\bar{x}_w \cup \bar{x}_r)$
where $f_{out}$ is the output variable mapping of $\sigma_{T_c}^c$. 
\end{itemize}
\vspace*{-1mm}
For every local symbolic run $\tilde{\rho}_T$ where $\gamma \neq \omega$ and $\tauout = \bot$, there exists a child of
$\tilde{\rho}_T$ which is not returning.
\end{definition}
\vspace*{-1mm}

Now consider an $\hltlfo$ formula $\varphi = [\xi]_{T_1}$ over $\Gamma$. Satisfaction of 
$\varphi$ by a symbolic tree of runs is defined analogously to satisfaction by local runs, 
keeping in mind that as previously noted, isomorphism types
of symbolic instances of $T$ provide enough information to evaluate conditions over $\bar x^T$. The definition of acceptance by
the automaton $\calb_\varphi$, and Lemma \ref{lem:proj}, are also immediately extended 
to symbolic trees of runs. We state the following.
\vspace*{-1mm}
\begin{lemma} \label{lem:symbuchi} 
A symbolic tree of runs $\Sym$ over $\Gamma$ satisfies $\varphi$ iff $\Sym$ is accepted by $\calb_\varphi$.
\end{lemma}  
\vspace*{-1mm}

The key result enabling the use of symbolic trees of runs is the following (see Appendix for proof). \yuliang{need to be fixed.}
\vspace*{-1mm}
\begin{theorem} \label{thm:actual-symbolic}
For an artifact system $\Gamma$ and $\hltlfo$ property $\varphi$,
there exists a tree of local runs $\Tree$ accepted by $\calb_\varphi$,
iff there exists a symbolic tree of runs $\Sym$ accepted by $\calb_\varphi$.
\end{theorem}
\vspace*{-1mm}

The {\em only-if} part is relatively straightforward,
but the {\em if} part is non-trivial. The construction of an accepted tree of local runs from an accepted symbolic tree of runs $\Sym$
is done in two stages. First, an accepted tree of local runs over an {\em infinite} database is constructed, using a global equality type
that extends the local equality types by taking into account connections across instances 
resulting from the propagation of input variables and insertions and retrievals of tuples from $S^T$, 
and subject to satisfaction of the key constraints.
In the second stage, the infinite database is turned into a finite one  
by carefully merging data values, while avoiding any inconsistencies.
%One of the subtleties is showing that the mismatch between symbolic
%and actual transitions discussed above, leading to the possible overestimation by the counters in symbolic runs of the number of
%tuples available in artifact relations, is not dangerous.

\subsection{Symbolic Verification}

In view of Theorem \ref{thm:actual-symbolic}, we can
now focus on the problem of checking the existence of a symbolic tree of runs satisfying a given $\hltlfo$ property.
To begin, we define a notion that captures the functionality of each task and allows a modular approach to 
the verification algorithm.  Let $\varphi$ be an $\hltlfo$ formula over $\Gamma$, and recall the automaton $\calb_\varphi$ 
and associated notation from Section \ref{sec:ltl-fo}.
We consider the relation $\calr_T$ between input and outputs of each task, 
defined by its symbolic runs that satisfy a 
given truth assignment $\beta$ to the formulas in $\Phi_T$. More specifically, we denote by $\calh_T$ the restriction of $\calh$ to $T$ 
and its descendants, and $\Gamma_T$ the corresponding HAS, with precondition {\em true}. 
The relation $\calr_T$ consists of the set of triples $(\tau_{in}, \tau_{out}, \beta)$
for which there exists a symbolic tree of runs $\Sym_T$ of $\calh_T$ such that:
\vspace*{-1mm}
\begin{itemize}\itemsep=0pt\parskip=0pt
\item $\beta$ is a truth assignment to $\Phi_T$
\item $\Sym_T$ is accepted by $\calb_\beta$
\item the root of $\Sym_T$ is $ \tilde{\rho}_{T} = (\tau_{in}, \tau_{out}, \{(I_i,\sigma_i)\}_{0 \leq i < \gamma})$ 
\end{itemize}\vspace*{-1mm}
Note that there exists a symbolic tree of runs $\Sym$ over $\Gamma$ satisfying $\varphi = [\xi]_{T_1}$ iff
$(\tau_{in}, \bot, \beta) \in \calr_{T_1}$ for some $\tau_{in}$ satisfying the precondition of $\Gamma$, and $\beta(\xi) = 1$.
Thus, if $\calr_{T}$ is computable for every $T$, then satisfiability of $[\xi]_{T_1}$ by some symbolic tree of runs over $\Gamma$ is decidable,
and yields an algorithm for model-checking $\hltlfo$ properties of HAS's.

%It is easy to see that for every $T$, $\calr_T$ depends only on the relations $\calr_{T_c}$ of its children tasks.

%\begin{lemma} \label{lem:subtrees}
%Consider an adorned symbolic tree of runs $\Sym_T$ with an edge $(i, \alpha)$ from the root $\tilde{\rho}_{T}$
%to a subtree $\emph{tree}(\tilde{\rho}_{T_c})$ rooted at $\tilde{\rho}_{T_c}$ and accepted by $\calb_\alpha$. 
%Consider another symbolic tree of runs $\emph{tree}(\hat{\rho}_{T_c})$ rooted at $\hat{\rho}_{T_c}$ 
%and also accepted by $\calb_\alpha$, such that $\tilde{\rho}_{T}$ and $\hat{\rho}_{T}$ have the same input and output.
%Let $\hat{\Sym}_T$ be obtained from $\Sym_T$ be replacing 
%$\emph{tree}(\tilde{\rho}_{T_c})$ with $\emph{tree}(\hat{\rho}_{T_c})$. Then for every truth assignment 
%$\beta$ to $\Phi_T$, $\Sym_T$ is accepted by $\calb_\beta$ iff $\hat{\Sym}_T$ is accepted by $\calb_\beta$.
%\end{lemma}

We next describe an algorithm that computes the relations $\calr_T(\tauin, \tauout, \beta)$ recursively.
The algorithm uses as a key tool Vector Addition Systems with States (VASS) \cite{vass,vass-repeated-reachability}, 
which we review next.

%\vspace{4mm}
%\noindent {\bf Review of VASS}  
A VASS $\calv$ is a pair $(Q,A)$ where $Q$ is a finite set of {\em states} and $A$ is a finite set of {\em actions} of the form
$(p,\bar a,q)$ where $\bar a \in \mathbb{Z}^d$ for some fixed $d > 0$, and $p,q \in Q$. A run  
of $\calv = (Q,A)$ is a finite sequence $(q_0, \bar z_0) \ldots (q_n,\bar z_n)$ where 
$\bar z_0 = \bar 0$ and for each $i \geq 0$, $q_i \in Q$, 
$\bar z_i \in \mathbb{N}^d$, and $(q_i, \bar a, q_{i+1}) \in A$ for some $\bar a$ such that $\bar z_{i+1} = \bar z_i + \bar a$.  
We will use the following decision problems related to VASS.
\vspace*{-1mm}
\begin{itemize}\itemsep=0pt\parskip=0pt
\item {\em State Reachability}: For given states $q_0, q_f \in Q$, is there a run $(q_0, \bar z_0) \ldots (q_n,\bar z_n)$
of $\calv$ such that $q_n = q_f$ ?
\item {\em State Repeated Reachability}:  For given states $q_0,q_f \in Q$, is there a run 
$(q_0, \bar z_0) \ldots (q_m, \bar z_m) \ldots (q_n,\bar z_n)$
of $\calv$ such that $q_m = q_n = q_f$ and $\bar z_m \leq \bar z_n$ ?
\end{itemize}\vspace*{-1mm}
Both problems are known to be {\sc expspace}-complete \cite{vass76,vass78,vass-repeated-reachability} \reviewer{need to cite both rackoff 78 and lipton 76} \yuliang{fixed}. In particular, 
\cite{vass-repeated-reachability} shows that for a $n$-states, $d$-dimensional VASS where every dimension of each action has constant size, 
the state repeated reachability problem can be solved in $O(( \log n ) 2^{c \cdot d \log d})$ non-deterministic space for some constant $c$.
The state reachability problem has the same complexity.

%can be solved by reducing to the repeated reachability on the VASS obtained by
%adding transitions with $\bar{a} = \bar{0}$ from each state to itself and thus it can be solved in the
%same space complexity.

\begin{table*}[!ht]
\centering
\vspace*{-4mm}
\begin{tabular}{|c|c|c|c|}
\hline
 & Acyclic & Linearly-Cyclic & Cyclic \\ \hline
w/o. Artifact relations & $c \cdot N^{O(1)}$ & $O(N^{c \cdot h}) $ & $h\hexp(O(N))$\\ \hline
w.   Artifact relations & $O(\exp(N^{c}))$ & $O(2\hexp(N^{c \cdot h}))$ & $(h + 2)\hexp(O(N))$ \\ \hline
\end{tabular}
\vspace*{-2mm}
\caption{\small Space complexity of verification without arithmetic 
($N$: size of $(\Gamma, \varphi)$; $h$: depth of hierarchy; $c$: constants depending on the schema)}
\vspace*{-4mm}
\label{tab:complexity1}
\end{table*}

\vspace{2mm}
\noindent 
{\bf VASS Construction}
Let $T$ be a task, and suppose that relations $\calr_{T_c}$ have been computed for all children $T_c$ of $T$.
We show how to compute $\calr_T$ using an associated VASS.
For each truth assignment $\beta$ of ${\Phi_T}$, we construct a VASS $\calv(T, \beta) = (Q,A)$ as follows. 
The states in $Q$ are all tuples $(\tau, \sigma, q, \bar{o},\bar c_{ib})$ 
where $\tau$ is a $T$-isomorphism type, $\sigma$ a service, 
$q$ a state of $B(T, \beta)$, and $\bar c_{ib}$ a mapping from $\tsib$ to $\{0,1\}$.
The vector $\bar o$ indicates the current stage of each child $T_c$ of $T$ ($\ainit$, $\aactive$ or $\aclosed$)
and also specifies the outputs of $T_c$ (an isomorphism type or $\bot$). 
That is, $\bar o$ is a partial mapping associating to
some of the children $T_c$ of $T$ the value $\bot$,
a $T_c$-isomorphism type projected to $\bar{x}^{T_c}_{in} \cup \bar{x}^{T_c}_{ret}$ or the value $\aclosed$.
Intuitively, $T_c \not\in \emph{dom}(\bar{o})$ means that $T_c$ is in the $\ainit$ state,
and $\bar{o}(T_c) = \bot$ indicates that $T_c$ has been called but will not return.
If $\bar{o}(T_c)$ is an isomorphism type $\tau$, this indicates that $T_c$ has been called, has not yet returned,
and will return the isomorphism type $\tau$. When $T_c$ returns, 
$\bar{o}(T_c)$ is set to $\aclosed$, and $T_c$ cannot be called again before an internal service of $T$ is applied. 

The set of actions $A$ consists of all triples $(\alpha, \bar a, \alpha')$ where
$\alpha = (\tau, \sigma, q, \bar{o}, \bar c_{ib})$, $\alpha' = (\tau', \sigma', q', \bar{o}',\bar c_{ib}')$, $\delta'$ is the update of $\sigma'$, and the following hold:
\vspace*{-1mm}
\begin{itemize}\itemsep=0pt\parskip=0pt
\item $\tau'$ is a successor of $\tau$ by applying service $\sigma'$;
%\item $\bar{m}'$ is a successor of $\bar{m}$ by applying service $\sigma'$.
\item $\bar a = \bar{a}(\delta', \hat{\tau}, \hat{\tau}', \bar c_{ib})$ (defined in Section \ref{sec:symrep}), where  $\hat{\tau} = \tau|(\bar{x}_{in}^T \cup \bar{s}^T)$ 
and $\hat{\tau}' = \tau'|(\bar{x}_{in}^T \cup \bar{s}^T)$
\item $\bar c'_{ib} = \bar c_{ib} + \bar a$ 
\item if $\sigma'$ is an internal service, $\emph{dom}(\bar o') = \emptyset$.
\item If $\sigma' = \sigma_{T_c}^o$, then $T_c \not\in \emph{dom}(\bar{o})$ and for  \\
$\tauin^{T_c} = f_{in}^{-1}(\tau | (\xttcdown, h(T_c))),$
for some output $\tauout^{T_c}$ of $T_c$ and truth assignment $\beta^{T_c}$ to $\Phi_{T_c}$,
tuple $(\tauin^{T_c}, \tauout^{T_c}, \beta^{T_c})$ is in $\calr_{T_c}$.
Note that $\tauout^{T_c}$ can be $\bot$, which indicates that this call to $T_c$ does not return.
Also, $\bar o' = \bar o[T_c \mapsto \tauout^{T_c}]$.
\item If $\sigma' = \sigma_{T_c}^c$, then 
$\bar{o}(T_c) = (f_{out} \circ f_{in}^{-1}) (\tau' | (\xttcdown \cup \xttcup, h(T_c)))$ and
 $\bar{o}' = \bar{o}[T_c \mapsto \aclosed]$.
\item $q'$ is a successor of $q$ in $B(T, \beta)$ by evaluating $\Phi_T$ using $(\tau', \sigma')$. 
If $\sigma' = \sigma_{T_c}^o$, formulas in $\Phi_{T_c}$ are assigned the truth values defined by $\beta^{T_c}$. 
\end{itemize}
\vspace*{-1mm}

An {\em initial} state of $\calv(T, \beta)$ is a state of the form
$v_0 =$ \\ $(\tau_0, \sigma_0, q_0, \bar{o}_0, \bar c^0_{ib})$ where 
$\tau_0$ is an initial $T$-isomorphism type (i.e., 
for every $x \in \bar{x}^{T}_{id} - \bar{x}^{T}_{in}$, $x \sim_{\tau_0} \anull$,
and for every $x \in \bar{x}^{T}_{\mathbb{R}} - \bar{x}^{T}_{in}$, $x \sim_{\tau_0} 0$),
$\sigma_0 = \sigma_T^o$,
$q_0$ is the successor of some initial state of $B(T, \beta)$ under $(\tau_0, \sigma_0)$,
$\emph{dom}(\bar{o}_0) = \emptyset$, and $\bar c^0_{ib} = \bar 0$.

\vspace{2mm}
\noindent
{\bf Computing $\calr_T(\tauin, \tauout, \beta)$ from $\calv(T, \beta)$} \\
Checking whether $(\tauin, \tauout, \beta)$ is in $\calr_T$ can be done using a (repeated) reachability test on
$\calv(T, \beta)$, as stated in the following key lemma (see Appendix for proof). \yuliang{need to be fixed.}
\begin{lemma} \label{lem:vass-correctness}
$(\tauin, \tauout, \beta) \in \calr_T$ iff
there exists an initial state 
$v_0 = (\tau_0, \sigma_0, q_0, \bar{o}_0, \bar c^0_{ib})$ of $\calv(T, \beta)$ for which $\tau_0 | \bar{x}^T_{in} = \tau_{in}$
and the following hold:
\vspace*{-1mm}
\begin{itemize}\itemsep=0pt\parskip=0pt
\item If $\tauout \neq \bot$, then there exists state 
$v_n = (\tau_n, \sigma_n, q_n, \bar{o}_n, \bar c^n_{ib})$ where
$\tauout = \tau_n | (\bar{x}^T_{in} \cup \bar{x}^T_{ret})$, 
$\sigma_n = \sigma_T^c$, $q_n \in Q^{fin}$ where
$Q^{fin}$ is the set of accepting states of $B(T, \beta)$ for finite runs, such that $v_n$ is reachable from $v_0$.
A path from $(v_0, \bar{0})$ to $(v_n, \bar{z}_n)$ is called a \textbf{returning path}.
\item If $\tauout = \bot$, then one of the following holds:
  \vspace*{-1mm}
  \begin{itemize}\itemsep=0pt\parskip=0pt
  \item there exists a state $v_n = (\tau_n, \sigma_n, q_n, \bar{o}_n, \bar c^n_{ib})$ in which $q_n \in Q^{inf}$ where $Q^{inf}$ is the set of accepting states of $B(T, \beta)$ for infinite runs, such that
  $v_n$ is repeatedly reachable from $v_0$. A path 
$(v_0, \bar 0) \ldots (v_n, \bar z_n) \ldots (v_n,\bar z_n')$ where $\bar z_n \leq \bar z_n'$ is called a \textbf{lasso path}.
  \item There exists state $v_n = (\tau_n, \sigma_n, q_n, \bar{o}_n, \bar c^n_{ib})$ in which $\bar{o}_n(T_c) = \bot$ for some child $T_c$ of $T$ and $q_n  \in Q^{fin}$, such that
  $v_n$ is reachable from $v_0$. The path from $(v_0, \bar{0})$ to $(v_n, \bar{z}_n)$ is called a \textbf{blocking path}.
  \end{itemize}
  \vspace*{-1mm}
\end{itemize}
\end{lemma}

\noindent
{\bf Complexity of Verification}
We now have all ingredients in place for our verification algorithm. Let $\Gamma$ be a HAS and $\varphi = [\xi]_{T_1}$ an $\hltlfo$ 
formula over $\Gamma$.  In view of the previous development, $\Gamma \models \varphi$ iff
$[\neg\xi]_{T_1}$ is {\bf not} satisfiable by a symbolic tree of runs of $\Gamma$. We outline a non-deterministic algorithm for 
checking satisfiability of $[\neg\xi]_{T_1}$, 
and establish its space complexity $O(f)$, where $f$ is a function of the relevant parameters. 
The space complexity of verification (the complement) is then $O(f^2)$ by Savitch's theorem \cite{sipser}.
  
Recall that $[\neg\xi]_{T_1}$ is satisfiable by a symbolic tree of runs of $\Gamma$ iff 
$(\tau_{in}, \bot, \beta) \in \calr_{T_1}$ for some $\tau_{in}$ satisfying the precondition of $\Gamma$, and $\beta(\neg\xi) = 1$.
By Lemma \ref{lem:vass-correctness}, membership in $\calr_{T_1}$ can be reduced to 
state (repeated) reachability in the VASS $\calv(T_1, \beta)$. 
For a given VASS, (repeated) reachability is decided by non-deterministically generating runs of the VASS up to a certain length,
using space $O( \log n \cdot 2^{c \cdot d \log d})$ 
where $n$ is the number of states, $d$ is the vector dimension and $c$ is a constant \cite{vass-repeated-reachability}.
The same approach can be used for the VASS $\calv(T_1, \beta)$, with the added complication that generating transitions
requires membership tests in the relations $\calr_{T_c}$'s for $T_c \in child(T_1)$. These in turn become
(repeated) reachability tests in the corresponding VASS. Assuming that $n$ and $d$ are upper bounds for the number of states
and dimensions for all $\calv(T, \beta)$ with $T \in \calh$, 
this yields a total space bound of $O( h \log n \cdot 2^{c \cdot d \log d})$
for membership testing in $\calv(T_1, \beta)$, where $h$ is the depth of $\calh$.

In our construction of $\calv(T, \beta)$, the vector dimension $d$ is the number of $TS$-isomorphism types.
The number of states $n$ is at most the product of the number of $T$-isomorphism types,
the number states in $B(T, \beta)$, the number of all possible $\bar{o}$ and the number of possible states of $\bar{c}_{ib}$.
The worst-case complexity occurs for HAS with unrestricted schemas (cyclic foreign keys) and artifact relations.
To understand the impact of the foreign key structure and 
artifact relations, we also consider the complexity for acyclic and linear-cyclic schemas, and without artifact relations.  
A careful analysis yields the following (see Appendix \ref{sec:complexity-no-arith-app}). \yuliang{need to be fixed.}
For better readability, 
we state the complexity for HAS over a fixed schema (database and maximum arity of artifact relations). The impact of the schema is
detailed in Appendix \ref{sec:complexity-no-arith-app}. \yuliang{need to be fixed.}

\vspace*{-1mm}
\begin{theorem} \label{thm:no-arithmetic} 
Let $\Gamma$ be a HAS over a fixed schema and $\varphi$ an $\hltlfo$ formula over $\Gamma$. 
The deterministic space complexity of checking
whether $\Gamma \models \varphi$ is summarized in Table \ref{tab:complexity1}.
\footnote{$k\hexp$ is the tower of exponential functions of height $k$.}
\end{theorem}
\vspace*{-1mm}
%\victor{I don't know how to keep the table on this or the next page. It is now at the end.}
%\yuliang{TODO: add star schema as a special case}

Note that the worst-case space complexity is non-elementary, as for feedback-free systems \cite{tods12}.
However, the height of the tower of exponentials in \cite{tods12} is the square of the total number of artifact variables of the system,
whereas in our case it is the depth of the hierarchy, likely to be much smaller.

\section{Verification with Arithmetic} \label{sec:arithmetic}
\newcommand{\qe}{\ensuremath{\texttt{QE} } }
\newcommand{\poly}{\ensuremath{\texttt{proj} } }
\newcommand{\calert}{\cale^T_{\mathbb{R}} }
\newcommand{\calertc}{\cale^{T_c}_{\mathbb{R}} }
\newcommand{\calertct}{\cale_{\mathbb{R}}^{T_c \rightarrow T} }

\reviewer{relation with hybrid systems?}

\begin{table*}[!ht]
\centering
\vspace*{-4mm}
\begin{tabular}{|c|c|c|c|}
\hline
 & Acyclic & Linearly-Cyclic & Cyclic \\ \hline
w/o. Artifact relations & $O(\exp(N^{c \cdot h}))$ & $O(\exp(N^{c \cdot h^2}))$ & $(h + 1)\hexp(O(N))$ \\ \hline
w.   Artifact relations & $O(2\hexp(N^{c \cdot h}))$ & $O(2\hexp(N^{c \cdot h^2}))$ & $(h + 2)\hexp(O(N)))$ \\ \hline
\end{tabular}
\vspace*{-2mm}
\caption{\small Space complexity of verification with arithmetic ($N$: size of $(\Gamma, \varphi)$; $h$: depth of hierarchy; $c$: constants depending on the schema)}
\vspace*{-4mm}
\label{tab:complexity2}
\end{table*}

We next outline the extension of our verification algorithm to handle HAS and $\hltlfo$ properties whose conditions
use arithmetic constraints expressed as polynomial inequalities with integer coefficients over the numeric variables 
(ranging over $\mathbb{R}$).
We note that one could alternatively limit the arithmetic constraints to linear inequalities with integer coefficients 
(and variables ranging over $\mathbb{Q}$), with the same complexity results. These are sufficient for many applications.

The seed idea behind our approach is that, in order to determine whether 
the arithmetic constraints are satisfied, we do not need to keep track of actual valuations 
of the task variables and the numeric navigation expressions they anchor (for which the search space would be infinite).
Instead, we show that these valuations can be partitioned into a finite set of equivalence classes
with respect to satisfaction of the arithmetic constraints, which we then incorporate into the 
isomorphism types of Section~\ref{sec:verification}, extending the algorithm presented there.
This however raises some significant technical challenges, which we discuss next. 

Intuitively, this approach uses the fact that a finite set of polynomials $\calp$ 
partitions the space into a bounded number of {\em cells} containing points 
located in the same region ($=0, <0, >0$)  with respect to every polynomial $P \in \calp$.
%Our verification algorithm uses the cells induced by the polynomials appearing in arithmetic constraints
%(as well as others constructed from them, detailed shortly) to represent, for each task, the 
%numeric expression valuations that are equivalent with respect to satisfaction of the arithmetic constraints
%(numeric expressions are numeric attribute variables or expressions $x_R.w$ where $w$ ends with numeric a attribute).
Isomorphism types are extended to include a cell,
%, and the algorithm performs essentially the same search as described in
%the arithmetic-free case (Section~\ref{sec:verification}).
%Cell information 
which determines which arithmetic constraints are satisfied in the conditions of services 
and in the property.
%The formal definition of cells can be found in Appendix \ref{sec:cells}.
In addition to the requirements detailed in Section~\ref{sec:verification}, 
we need to enforce cell compatibility across symbolic service calls.
For instance, when a task executes an internal service, the corresponding symbolic transition from cell $c$ to $c'$ is 
possible only if the projections of $c$ and $c'$ on the subspace corresponding to the task's input variables have non-empty intersection 
(since input variables are preserved).
%, their valuation must simultaneously belong to both cells, and indeed to all cells in the sequence of transitions leading up to $c$). 
Similarly, when the opening or closing service of a child task is called, 
compatibility is required between the parent's and the child's cell 
on the shared variables, which amounts again to non-empty intersection between cell projections.
%Indeed, results of the arithmetic constraint checks in the parent are
%inherited by the child and can influence the outcome of arithmetic checks it performs. 
%More subtly, the arithmetic constraints check in the child task must be reflected in the parent cell,
%even after the child task returns and its local variables are no longer in scope. 
%For instance assume that the child task can return only if its input from the parent is less than 
%a local variable of the child, which in turn is less than 5. If the return takes place, 
%then the corresponding parent variable is constrained to be less than 5.
This suggests the following first-cut (and problematic) 
attempt at a verification algorithm:  
once a local transition imposes new constraints, represented by a cell $c'$, these constraints are propagated {\em back} 
to previously guessed cells, refining them via 
intersection with $c'$. If an intersection becomes empty, the candidate symbolic run constructed so far has no corresponding actual run and the search is pruned. 
The problem with this attempt is that it is incompatible with the way we deal with sets in Section~\ref{sec:verification}: the contents of sets
are represented by associating counters to the isomorphism types of their elements. Since extended isomorphism types include cells, retroactive
cell intersection invalidates the counters and the results of previous VASS reachability checks.

We develop an alternative solution that avoids retroactive cell intersection altogether.
More specifically, for each task, our algorithm extends isomorphism types with cells guessed 
from a {\em pre-computed} set constructed by following the task hierarchy bottom-up and 
including in the parent's set those cells obtained by appropriately projecting the children's cells on shared variables and expressions.
Only non-empty cells are retained. We call the resulting cell collection the Hierarchical Cell Decomposition (HCD). 

The key benefit of the HCD
is that it arranges the space of cells so that consistency of a symbolic run can 
be guaranteed by performing simple local compatibility tests on the cells 
involved in each transition. Specifically, (i) in the case of internal service calls, 
the next cell $c'$ must {\em refine} the current cell $c$ on the shared variables  
(that is, the projection of $c'$ must be contained in the projection of $c$);
(ii) in the case of child task opening/closing services, the parent cell $c$ must refine the child cell $c'$. 
This ensures that in case (i) the intersection with $c'$ of all relevant previously guessed cells is non-empty 
(because we only guess non-empty cells and $c'$ refines all prior guesses), and
in case (ii) the intersection with the child's cell $c'$ is a no-op for the parent cell. 
Consequently, retroactive intersection can be skipped as it can never lead to empty cells.

A natural starting point for constructing the HCD is to gather for each task 
all the polynomials appearing in its arithmetic constraints 
(or in the property sub-formulas referring to that task), and associate sign conditions to each. This turns out to be insufficient.
For example, the projection from the child cell can impose on the parent variables new constraints which do not appear explicitly in the parent task.
%Similarly, the projection from any cell to the input variables/expressions can also produce new constraints.
It is a priori not obvious that the constrained cells can be represented symbolically, let alone efficiently computed.
The tool enabling our solution is the Tarski-Seidenberg Theorem~\cite{tarski-seidenberg}, which ensures that the projection of a cell is representable
by a union of cells defined by a set of polynomials (computed from the original ones) and sign conditions for them.
The polynomials can be efficiently computed using quantifier elimination.
%Moreover, the classic technique of real quantifier elimination provides an efficient means to compute these polynomials.

Observe that a bound on the number of newly constructed polynomials yields a bound on the number of cells in the
HCD, which in turn implies a bound on the number of distinct extended isomorphism types 
manipulated by the verification algorithm, ultimately yielding decidability of verification. 
%by following a proof similar to that of Theorem \ref{thm:actual-symbolic}. 
A naive analysis produces a bound on the number of cells that is hyperexponential 
in the height of the task hierarchy, because the number of polynomials can proliferate at this rate
when constructing all possible projections, and $p$ polynomials may produce $3^p$ cells.
Fortunately, a classical result from real algebraic geometry (\cite{basu1996number}, reviewed in Appendix \ref{sec:algebraic-geometry}) \yuliang{need to be fixed.}
bounds the number of distinct {\em non-empty} cells to only exponential in the number of variables 
(the exponent is independent of the number of polynomials). 
This yields an upper bound of the number of cells (and also the number of extended isomorphism types) 
which is singly exponential in the number of numeric expressions and doubly exponential in the height of the hierarchy $\calh$.
%Appendix~\ref{ap:arithmetic} details and formalizes the above discussion. 
%First, quantifier elimination and real algebraic geometry are reviewed (Appendix \ref{sec:qe} and \ref{sec:algebraic-geometry}).
%We then define the cells associated to a task (Appendix~\ref{sec:cells}) for verification purpose
%and their construction using Hierarchical Cell Decomposition (Appendix~\ref{sec:cell-decomp}). 
%We next describe a natural extension of isomorphism types and symbolic runs with cells 
%constructed from the hierarchical decomposition in Appendix~\ref{sec:extended-tau}. 
%The correctness of the extension is proved in \ref{sec:connecting-runs}).
We state below our complexity results for verification with arithmetic, 
relegating details (including a fine-grained analysis) to Appendix~\ref{app:arithmetic}. \yuliang{need to be fixed.}

%\textbf{Remark: } Though our complexity results are established on the fact that 
%the arithmetic constraints are polynomial inequalities over $\mathbb{R}$, 
%the same set of results can be obtained for arithmetic constraints over $\mathbb{Q}$ in models where only
%linear constraints are allowed, which are sufficiently expressive in many applications.
%\yuliang{Added the above remark on rational numbers. Please check.}
%\victor{Added instead a sentence atthe beginning of the section.}

\vspace*{-2mm}
\begin{theorem} \label{thm:arithmetic}
Let $\Gamma$ be a HAS over a fixed database schema
and $\varphi$ an $\hltlfo$ formula over $\Gamma$. 
If arithmetic is allowed in $(\Gamma, \varphi)$, then
the deterministic space complexity of checking whether $\Gamma \models \varphi$ is summarized in Table \ref{tab:complexity2}.
%where $N$ is the size of $(\Gamma,\varphi)$ and $h$ is the depth of $\calh$.
\end{theorem}
\vspace*{-2mm}

\section{Restrictions and Undecidability} \label{sec:undecidability}

We briefly review the main restrictions imposed on the HAS model and motivate them by showing that
they are needed to ensure decidability of verification. Specifically, recall that the following restrictions are placed in the model:
\vspace*{-1.5mm}
\begin{enumerate}\itemsep=0pt\parskip=0pt
\item in an internal transition of a given task (caused by an internal service),
only the input parameters of the task are explicitly propagated from one artifact tuple to the next
\item each task may overwrite upon return only $\anull$
variables in the parent task 
\item the artifact variables of a task storing the values returned by its subtasks are disjoint from the task's input variables
\item  an internal transition can take place only if all active subtasks have returned
\item each task has just one artifact relation
\item the artifact relation of a task is reset to empty every time the task closes
\item the tuple of artifact variables whose value is inserted or retrieved from a task's artifact relation
is fixed 
%\item  the tuples of variables passed as input to subtasks and returned to parent tasks are fixed 
\item each subtask may be called at most once between internal transitions of its parent
\end{enumerate} \yuliang{one restriction is removed.}
\vspace*{-1mm}
These restrictions are placed in order to control the data flow and recursive computation in the system.
Lifting any of them leads to undecidability of verification, as stated informally next.

\vspace*{-1mm}
\begin{theorem} \label{thm:restrictions}
For each $i, 1 \leq i \leq 8$, let HAS$^{(i)}$ be defined identically to HAS but without restriction $(i)$ above.
It is undecidable, given a HAS$^{(i)}$ $\Gamma$ and an $\hltlfo$ formula $\varphi$ over $\Gamma$, whether $\Gamma \models \varphi$.
\end{theorem} 
\vspace*{-1mm}

The proofs of undecidability for (1)-(7) are by reduction from the Post Correspondence Problem (PCP) \cite{Post47,sipser}.
They make no use of arithmetic, so undecidability holds even without arithmetic constraints. The only undecidability 
result relying on arithmetic is (8). Indeed, restriction (8) can be lifted in the absence of numeric variables, with no impact on 
decidability or complexity of verification. % \victor{Is the complexity part true? It seems so.}
This is because restriction (2) ensures that even if a subtask is called repeatedly, only a bounded number of calls
have a non-vacuous effect. % \victor{Let's discuss this to be sure it's right.} 

The proofs using a reduction from the PCP rely on the same main idea: removal of the restriction allows
to extract from the database a path of unbounded length in a labeled graph, 
and check that its labels spell a solution to the PCP. 
For illustration, the proof of undecidability for (2) using this technique is sketched in Appendix \ref{app:undecidability}.
\yuliang{needs to be fixed.}

We claim that the above restrictions remain sufficiently permissive to capture
a wide class of applications of practical interest.
This is confirmed by numerous examples of practical business processes modeled as artifact systems, that we encountered
in our collaboration with IBM (see \cite{tods12}). The restrictions limit the recursion and data flow among tasks and services.
In practical workflows, the required recursion is rarely powerful enough to allow unbounded propagation of data among services.
Instead, as also discussed in \cite{tods12}, recursion is often due to two scenarios:

\vspace*{-1mm}
\begin{itemize}\itemsep=0pt\parskip=0pt
\item  allowing a certain task to undo and retry an unbounded number of times,
with each retrial independent of previous ones, and depending only on
a context that remains unchanged throughout the retrial phase (its input parameters).
A typical example is repeatedly
providing credit card information until the payment goes through,
while the order details remain unchanged.
\item  allowing a task to batch-process an unbounded collection of records,
each processed independently, with unchanged input parameters
(e.g. sending invitations to an event to all attendants on the list, for the
same event details).
\end{itemize}
\vspace*{-1mm}

Such recursive computation can be expressed with the above restrictions,
which are satisfied by our example provided in Appendix \ref{app:example-specification}. \yuliang{fixed}

\section{Related Work}\label{sec:related}
We have already discussed our own prior related work in the introduction. 
We summarize next other related work on verification of artifact systems.

%\smallskip
%\noindent{\bf Artifact-centric processes with no database.}
Initial work on formal analysis of artifact-based business processes
in restricted contexts has investigated reachability
\cite{GBS:SOCA07,GeredeSu:ICSOC2007}, general temporal
constraints \cite{GeredeSu:ICSOC2007}, and the existence of complete
execution or dead end \cite{BGHLS-2007:artifacts-analysis}.
For each considered problem, verification is
generally undecidable; decidability results were obtained only
under rather severe restrictions, e.g., restricting all pre-conditions to be
``true'' \cite{GBS:SOCA07}, restricting to bounded
domains \cite{GeredeSu:ICSOC2007,BGHLS-2007:artifacts-analysis},
or restricting the pre- and post-conditions to be propositional,
and thus not referring to data values \cite{GeredeSu:ICSOC2007}.
\cite{CGHS:ICSOC:09} adopts an artifact model variation
with arithmetic operations but no database.
%It proposes a criterion for comparing the expressiveness of specifications using the notion of
%{\em dominance}, based on the input/output pairs of business processes.
Decidability relies on restricting runs to boun\-ded length.
\cite{ZSYQ:TASE:09} addresses the problem of the existence of a run
that satisfies a temporal property, for a restricted case
with no database and only propositional LTL properties.
All of these works model no underlying database, sets (artifact relations), task hierarchy, or arithmetic.

%\smallskip
%\noindent {\bf Artifact-centric processes with underlying database.}
A recent line of work has tackled verification of  artifact-centric
processes with an underlying relational data\-base.
\cite{DBLP:conf/icsoc/BelardinelliLP11,DBLP:conf/ijcai/BelardinelliLP11,DBLP:conf/kr/BelardinelliLP12,DBLP:conf/icsoc/BelardinelliLP12,DBLP:journals/ijcis/GiacomoMR12} evolve the business
process model and property language, culminating in~\cite{DBLP:conf/pods/HaririCGDM13}, which addresses verification of first-order $\mu$-calculus (hence branching time)
properties over business processes expressed in a framework that is equivalent to artifact systems whose input is provided by external services. 
%Both deterministic and non-deterministic services are considered in~\cite{DBLP:conf/pods/HaririCGDM13} and decidability of verification is shown for run-bounded, 
%respectively state-bounded systems (the active domain size is bounded globally across the run, or locally at each step). 
%Depending on the nature of external services, decidability also requires the restriction of quantification across states in the run, yielding distinct FO $\mu$-calculus fragments. 
%\cite{DBLP:conf/icsoc/HaririCMSS13} introduces guidelines for rewriting artifact systems into equivalent state-bounded form to enable decidable verification, while 
\cite{DBLP:journals/jair/BelardinelliLP14,DBLP:conf/aaai/CalvaneseDM15} extend the results of~\cite{DBLP:conf/pods/HaririCGDM13} to artifact-centric multi-agent systems 
where the property language is a version of first-order branching-time temporal-epistemic logic expressing the knowledge of the agents. 
This line of work uses variations of a business process model called DCDS (data-centric dynamic systems), 
which is sufficienty expressive to capture the GSM model, as shown in~\cite{Montali:ICSOC:13}. In their unrestricted form,
DCDS and HAS have similar expressive power. However, the difference lies in the tackled verification problem and in 
the restrictions imposed to achieve decidability.
We check satisfaction of linear-time properties for every possible choice of initial database instance, 
whereas the related line checks branching-time properties and assumes that the initial database is given.
None of the related works address arithmetic. In the absence of arithmetic, the restrictions introduced
for decidability are incomparable (neither subsumes the other).

\eat{
\alin{this is already covered in introduction, good candidate for dropping given space crunch}
A third line of research started with publication~\cite{DHPV:ICDT:09} which considers an artifact model in which the infinite domain 
is equipped with a dense linear order and the property language is first-order LTL (no arithmetic or hierarchies are supported, and sets can be manipulated in significantly more restricted fashion). 
Follow-up work~\cite{tods12} adds support for arithmetic constraints (like the ones used here) 
as well as integrity constraints (tgds and egds) on the database, and is thus the most closely related work we are aware of.
Neither sets nor hierarchy are  supported in~\cite{tods12}.
As for arithmetic, the most important difference between the current paper and  ~\cite{tods12} lies in the verification techniques, with the current paper
improving the upper bound spectacularly: 
the bound in~\cite{tods12} is hyper-exponential in the squared number of artifact variables, even for acyclic schema without integrity constraints.
This discrepancy is not caused by lax complexity analysis (there are examples for which the worst case running time is realized), but by the superiority of the cell-refinement-based algorithm
presented here.
}

%There is a plethora of literature on artifact-specific issues beyond automatic verification, ranging from
%modeling to formalisation, runtime monitoring, and synthesis. We refer the reader to the survey~\cite{DBLP:journals/sigmod/DeutschHV14}.
Beyond artifact systems, there is a plethora of literature on data-centric processes, dealing with
various static analysis problems and also with runtime monitoring and synthesis. We discuss the most related works here and refer the reader to the surveys~\cite{DBLP:conf/pods/CalvaneseGM13,DBLP:journals/sigmod/DeutschHV14} for more.
Static analysis for semantic web services is considered in
\cite{McIlraith-www2002}, but in a context restricted to finite domains.
%\cite{ASV:TODS:09}
%studied automatic verification in the context of business processes based on Active XML documents.
The works~\cite{jcss,S00,AVFY98} are ancestors of~\cite{DHPV:ICDT:09}
from the context of verification of electronic commerce applications.
Their models could conceptually (if not naturally) be encoded in HAS 
but correspond only to 
particular cases supporting no arithmetic, sets, or hierarchies. 
Also, they limit external inputs  to essentially come from the active domain 
of the database, thus ruling out fresh values introduced during the run.

\eat{
\smallskip
\noindent
%{\bf Infinite-state systems.}
\alin{
Artifact systems are a particular case of infinite-state systems.
Research on automatic verification of infinite-state systems has
also focused on extending classical model checking techniques
(e.g., see \cite{Burkartetal01} for a survey).
However, in much of this work the emphasis is on studying
recursive control rather than data,
which is either ignored or finitely abstracted.
More recent work has been focusing specifically on data as a source of
infinity. This includes augmenting recursive procedures with integer
parameters \cite{Bouajjani&Habermehl&Mayr03}, rewriting systems with data
\cite{Bouajjani&Habermehl&Jurski&Sighireanu07},
Petri nets with data associated to tokens \cite{Lazicetal07},
automata and logics over infinite alphabets
\cite{Bouyer&Petit&Therien03,Bouyer02,Neven&Schwentick&Vianu04,Demri:2009:LFQ,JL07,BMSSD06,Bouajjani&Habermehl&Jurski&Sighireanu07},
and temporal logics manipulating data \cite{Demri:2009:LFQ}.
However, the restricted use of data and the particular properties verified
have limited applicability to the business process setting we target here.
}
}

\section{Conclusion} \label{sec:conclusion}
We showed decidability of verification for a rich artifact model capturing core elements of
IBM's successful GSM system:
task hierarchy, concurrency, database keys and foreign keys, arithmetic constraints, and richer artifact data.  
The extended framework requires the use of novel techniques including nested Vector Addition Systems and 
a variant of quantifier elimination tailored to our context. 
We improve significantly on previous work on verification of artifact systems with arithmetic~\cite{tods12},
which only exhibits non-elementary upper bounds regardless of the schema shape,
even absent artifact relations. In contrast, for acyclic and linearly-cyclic schemas, 
even in the presence of arithmetic and artifact relations, 
our new upper bounds are elementary
(doubly-exponential in the input size and triply-exponential in the depth of the hierarchy).
This brings the verification algorithm closer to practical relevance, particularly since
its complexity gracefully reduces to {\sc pspace} (for acyclic schema) and 
{\sc expspace} in the hierarchy depth (for 
linearly-cyclic schema) when arithmetic and artifact relations are not present.
The sole remaining case of nonelementary complexity occurs for arbitrary cyclic schemas.
%is a tower of exponentials whose height is
%roughly the depth of the hierarchy, which in general is much smaller compared to the size of the entire specification. 
%(see also the discussion Appendix \ref{sec:complexity-app})
%\yuliang{Didn't we decide removing that?}
%\victor{We talked about it but I don't think we decided.}
Altogether, our results provide substantial new insight and techniques for the automatic verification of realistic artifact systems.

\vspace*{2mm}
\noindent
\emph{Acknowledgement} This work was supported in part by the National Science Foundation under award IIS-1422375.

\bibliographystyle{abbrv}
\bibliography{reference,artifacts,art2}
\appendix

\section{Examples} \label{app:example}

In this section we provide an example of HAS modeling a simple travel booking business process similar to Expedia \cite{expedia}.
We also show an example property that the process should satisfy, using $\hltlfo$.

\subsection{Example Hierarchical Artifact System} \label{app:example-specification}

The artifact system captures a process
where a customer books flights and/or makes hotel reservations. 
%\yuliang{Given the more detailed high-level description below, this paragraph seems redundant.}
%\victor{I think the repetition is OK, they are at different levels of detail.}
The customer starts with constructing a trip by
%\victor{Let's avoid he/she everywhere. How about just using "she". Nobody complains "she" is sexist :)} \yuliang{Good idea.}
adding a flight and/or hotel reservation to it.
During this time, the customer has the choice to store the trip as a candidate or retrieve 
a previously stored trip. Once the customer has made a decision,
she can proceed to book the trip. If a hotel reservation is made together with certain flights, 
a discount price may be applied to the hotel reservation.
In addition, the hotel reservation can be made by itself, together with the flight, or even after the flight
is purchased. After submitting a valid payment, the customer is able to cancel the
flight and/or the hotel reservation and receive a refund. 
If the customer cancels the purchase of a flight, 
she cannot receive the discount on the hotel reservation.

The Hierarchical artifact system has the following database schema:
\vspace{-0.5mm}
\begin{itemize}\itemsep=0pt\parskip=0pt
\item
\dbflights $\mathtt{(\underline{id}, price, comp\_hotel\_id)}$\\
\dbhotels $\mathtt{(\underline{id}, unit\_price, discount\_price)}$
\end{itemize}
\vspace{-0.5mm}

In the schema, the $\mathtt{id}$'s are key attributes, 
$\mathtt{price}$, $\mathtt{unit\_price}$, $\mathtt{discount\_price}$
are non-key attributes, and $\mathtt{comp\_hotel\_id}$ is a foreign key attribute
satisfying the dependency \\ $\dbflights[comp\_hotel\_id] \subseteq \dbhotels[id]$.

Intuitively, each flight stored in the $\dbflights$ table has a hotel compatible for discount.
If a flight is purchased together with a compatible hotel reservation, 
a discount is applied on the hotel reservation. Otherwise, the full price needs to be paid. 

%\yuliang{fixed.}
%\alin{change Figure~\ref{fig:hierarchy} to say ``ChooseDate''} 
%\yuliang{done.}

The artifact system has 6 tasks: ``T1: \textbf{ManageTrips}'', 
``T2: \textbf{AddHotel}'', ``T3: \textbf{AddFlight}'', ``T4: \textbf{BookInitialTrip}'', ``T5: \textbf{Cancel}'' 
and ``T6: \textbf{AlsoBookHotel}'', which form the hierarchy represented in Figure \ref{fig:hierarchy}.
%\vspace*{-2mm}
\begin{figure}[!ht]
\centering
\includegraphics[scale=0.5]{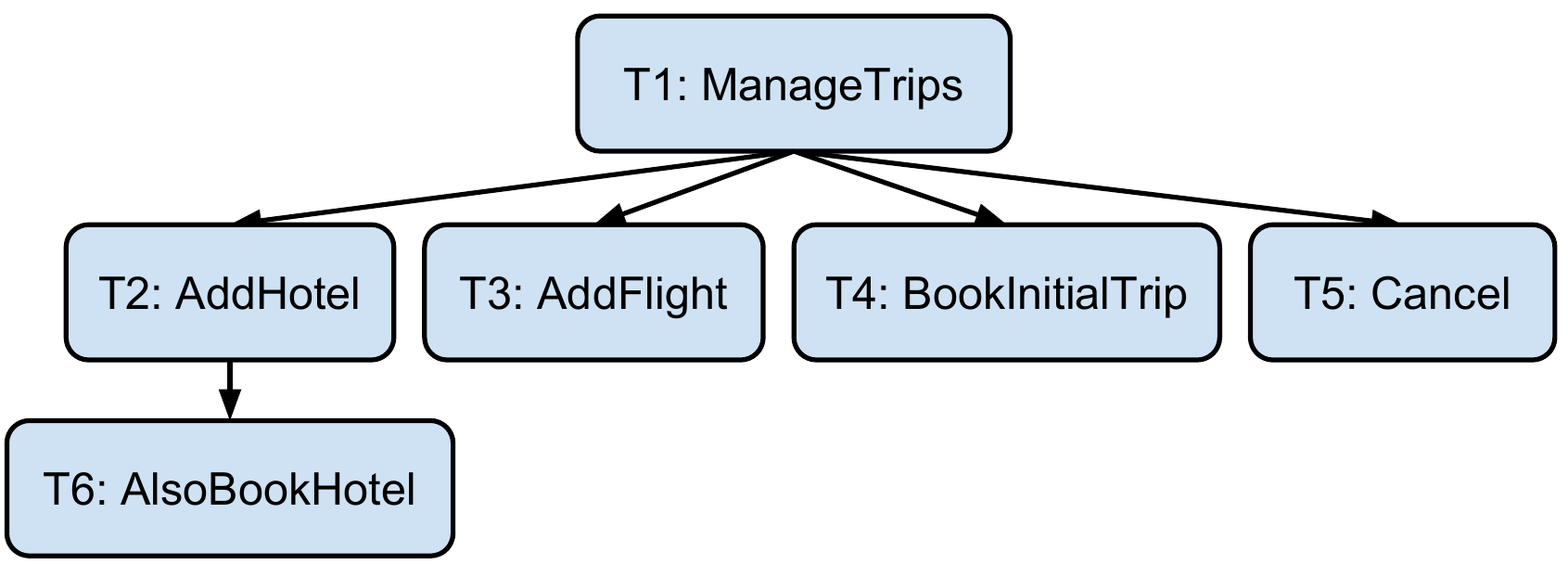}
\vspace*{-6mm}
\caption{Tasks Hierarchy} 
\label{fig:hierarchy}
\vspace*{-4mm}
\end{figure}

%\yuliang{Added the following high-level description of each task.}
%\victor{Looks good.}
The process can be described informally as follows. The customer starts with task \textbf{ManageTrips}, 
where the customer can add a flight and/or hotel to the trip by 
calling the \textbf{AddHotel} or the \textbf{AddFlight} tasks. 
%\alin{Customers can book flight, but also hotel, independently. However, there is only
%a PayHotel, and no PayFlight task. This looks asymmetric. Moreover, a glance at the figure raises the 
%question of why not pay for everything in Checkout. I know that this is explained
%later, but until that time the reader is confused. Maybe we can avoid this by changing 
%some names. Checkout $\mapsto$ BookInitialTrip, PayHotel $\mapsto$ AlsoBookHotel, or something like that.}
%\yuliang{The asymmetry should be fine since we are mimicking what Expedia does. The names of tasks look good.}
The customer is also allowed to store candidate trips in an artifact relation $\mathtt{TRIPS}$ 
and retrieve previously stored trips. 
(Note that for simplicity, our example considers only outbound flights in the trip.
Return flights can be added by a simple extension to the specification.)
After the customer has made a decision, the \textbf{BookInitialTrip}
task is called to book the trip and the payment is processed.
The process also mimics a key feature of Expedia as follows.
After payment is made successfully, if the customer booked the flight with no hotel reservation,
then she has the opportunity to add a hotel reservation by calling the \textbf{AddHotel} task.
When she does so, the task \textbf{AlsoBookHotel} needs to be called to handle the payment of
the added hotel reservation. Note that the \textbf{AlsoBookHotel} task can only be called after the flight is booked for
but a hotel reservation is missing in the trip.
Once the payment is made, the customer can cancel the order by calling the \textbf{Cancel} task.
Using \textbf{Cancel}, the customer is able to cancel the flight and/or the hotel with a full refund.
It is important to note that if the customer cancels the purchase of the flight,
then she cannot receive the discount on the hotel reservation.

The tasks are specified below.
For convenience, we use existential quantifications in conditions, which can be simulated by adding extra variables.
String values are used as syntactic sugar for numeric variables. We assume that the set of strings we used 
(``$\mathtt{Unpaid}$'', ``$\mathtt{Paid}$'', ``$\mathtt{FlightCanceled}$'', etc.)
correspond to distinct numeric constants. In particular, the string ``$\mathtt{Unpaid}$''
corresponds to the constant 0.
Also for convenience, we use artifact variables with the same names in parent and child tasks.
By default, each input/return variable is mapped to the variable in the parent/child task having the same name.

%\vspace{2mm}
%\noindent
%\textbf{ChooseDate}: This task allows the user to choose the dates of travel. 
%It has two ID artifact variables ($\adatefrom$ and \\ $\adateto$) and no artifact relation.
%The task has only one internal service \emph{Choose}, which picks 
%$\adatefrom$, and optionally $\adateto$. 
%As the root task, \textbf{ChooseDate} has the default root opening and closing services (see Section \ref{sec:framework}).
%The service \emph{Choose} is specified as follows.
%
%\vspace{2mm}
%\noindent
%\emph{Choose}: \\
%Pre-condition: \emph{True} \\
%Post-condition:
%\vspace{-2mm}
%\begin{align*}
%& \exists d_1 \exists d_2 \ \dbdates(\adatefrom, d_1) \land (\adateto = \anull \\
%& \lor (\dbdates(\adateto, d_2) \land (d_2 > d_1)))
%\end{align*}
%Note that this task can be applied any number of times (as long as its children tasks have returned),
%each time initiating another trip.
%

\vspace{2mm}
\noindent
\textbf{ManageTrips}: 
This is the root task, modeling the process whereby the customer creates, stores, and retrieves candidate trips.  
A trip consists of a flight and/or hotel reservation. 
Eventually, one of the candidate trips may be chosen for booking.
%\victor{Does Checkout implement the checkout? 
%If so please call it Checkout. At this point you should also explain the role of PayHotel under AddHotel.} 
As the root task, its opening condition is $\mathtt{true}$ and closing condition is $\mathtt{false}$.
The task has the following artifact variables:
\vspace*{-1mm}
\begin{itemize}\itemsep=0pt\parskip=0pt
\item ID variables:  $\aflightid$, $\ahotelid$, 
\item numeric variables: $\astatus$ and $\aamountpaid$ 
%\item input variables: $\adatefrom$, $\adateto$.
%\item return variables: none \victor{Nothing is said about the return variables and closing service of this task.} \yuliang{it is empty.}
\end{itemize}
\vspace*{-1mm}
It also has an artifact relation $\mathtt{TRIPS}$ storing candidate trips $(\aflightid, \ahotelid)$.
The customer can use the subtasks \textbf{AddFlight} and \textbf{AddHotel} (specified below) to fill in variables
$\aflightid$ and $\ahotelid$. In addition, the task has two internal services: \emph{StoreTrip} and \emph{RetrieveTrip}.
Intuitively, when \emph{StoreTrip} is called, the current candidate trip $(\aflightid, \ahotelid)$ is inserted into
$\mathtt{TRIPS}$. When \emph{RetrieveTrip} is called, one tuple is non-deterministically chosen and removed
from $\mathtt{TRIPS}$, and $(\aflightid, \ahotelid)$ is set to be the chosen tuple.
The two tasks are specified as follows:

\vspace{2mm}
\noindent
\emph{StoreTrip}: \\
Pre-condition: $\astatus = \text{``Unpaid''} \land (\aflightid \neq \anull \lor \ahotelid \neq \anull)$ \\
Post-condition: $\aflightid = \anull \land \ahotelid = \anull \land \astatus = \text{``Unpaid''} \land 
\aamountpaid = 0$ \\
Set update: $\{+\mathtt{TRIPS}(\aflightid, \ahotelid)\}$

\vspace{2mm}
\noindent
\emph{RetrieveTrip}: \\
Pre-condition: $\astatus = \text{``Unpaid''}$ \\
Post-condition: $\astatus = \text{``Unpaid''} \land \aamountpaid = 0$ \\
Set update: $\{-\mathtt{TRIPS}(\aflightid, \ahotelid)\}$

\vspace{2mm}
\noindent
\textbf{AddFlight}: This task adds a flight to the trip.
It can be opened if $\aflightid = \anull$ and $\astatus = \text{``Unpaid''}$ in the parent task. 
It has no input variable and the return variable is $\aflightid$.
The task has a single internal service \emph{ChooseFlight} 
that chooses a flight from the $\dbflights$ database and stores it in $\aflightid$, 
which is returned to \textbf{ManageTrips}. 

\vspace{2mm}
\noindent
\textbf{AddHotel}: 
This task adds a hotel reservation to the trip. It can be opened when $\ahotelid = \anull$ %, $\adateto \neq \anull$
and $\astatus$ is either ``Paid'' or ``Unpaid''. 
%\victor{Shouldn't $\adateto$ be non-null?} \yuliang{Right. Fixed.}

%If $\astatus = \text{``Paid''}$, then the 
%the case that the customer chooses to add a hotel 
%reservation after selecting the air ticket without reserving a hotel room. It can be opened if 
%$\ahotelid$ in the parent task is null.
This task has the following artifact variables:
\vspace*{-1mm}
\begin{itemize}\itemsep=0pt\parskip=0pt
\item ID variables: \uline{$\aflightid$}\footnote{the underlined variables are input variables}, \uwave{$\ahotelid$}\footnote{the wavy underlined variables are return variables}
\item numeric variables: \uline{$\astatus$}, \uline{$\aamountpaid$}, \uwave{$\anewamountpaid$} (overwriting $\aamountpaid$ in the parent task when the task returns), $\alowerprice$, $\ahigherprice$ and $\ahotelprice$
%\item input variables: $\astatus$, $\aflightid$, $\aamountpaid$, \\ $\adatefrom$, $\adateto$
%\item return variables: $\ahotelid$, $\anewamountpaid$
\end{itemize}
\vspace*{-1mm}
The task has a single internal service \emph{ChooseHotel} 
which picks a hotel from $\dbhotels$ and determines the price by checking 
whether the hotel is compatible with the chosen flight.
If they are compatible, then $\ahotelprice$ is set to 
the discount price, otherwise it is set to the full price. 
%\alin{by checking compatiblity of flight with hotel}
%\yuliang{fixed.}

A hotel can be added to the trip in two scenarios.
First, if $\astatus$ is ``Unpaid'', which means that the trip has not been booked,
then this task chooses a hotel and the id of the hotel is returned to \textbf{ManageTrips}.
Second, if $\astatus$ is ``Paid'',  which means that a flight has already been purchased without a hotel reservation,
then this task chooses a hotel and then the child task \textbf{AlsoBookHotel} needs to be called to handle the payment 
of the newly added hotel. In \textbf{AlsoBookHotel}, a payment is received and the new total amount of payment received is written into
variable $\anewamountpaid$ when \textbf{AlsoBookHotel} returns.
%\victor{Maybe this should be called "FlightPaid"}
%\victor{Which task implements purchases of flights?} 
%\victor{Why only when status is Paid ? What if it is Unpaid? } 
%\victor{If this handles payment then what does the {\bf Checkout} task do ?} \yuliang{Added explanation. Please check.}
%\victor{I still don't understand which task pays for what and when. This must be explained in the beginning. Also it would be helpful to pick suggestive names. For example, if {\bf Checkout} only deals with flight payments, it should be called {\bf FlightCheckout}, etc. }
%\yuliang{added explanation. See above.}
 % and $\anewamountpaid$ is calculated. \victor{What are the input variables?}
% It also checks whether the air ticket is paid or unpaid by checking the input variable $\astatus$. 
%\victor{What's the intuition? Why can't {\bf Checkout} just handle all payments?}

The closing service of \textbf{AddHotel} has condition $\astatus = \text{``Unpaid''} \lor 
(\astatus = \text{``Paid''} \land \ahotelprice = \\ \anewamountpaid - \aamountpaid)$, 
which means that either there is no need to call \textbf{AlsoBookHotel} or a correct payment has been received in \textbf{AlsoBookHotel}.
%\victor{What is $\anewamountpaid$ ? It appears here for the first time.}
%and $\ahotelid$ and \\ $\anewamountpaid$ are returned to the parent task ($\aamountpaid$ in the parent task is written
%by the value of $\anewamountpaid$). 
The \emph{ChooseHotel} service is specified as follows:

\vspace{2mm}
\noindent
\emph{ChooseHotel}: \\
Pre-condition: $\mathtt{True}$ \\
Post-condition: \vspace*{-1mm}
\begin{align*}
& \exists cid \exists p_f \ (\aflightid = \anull \rightarrow cid = \anull) \land \\
& (\aflightid \neq \anull \rightarrow \dbflights(\aflightid, p_f, cid)) \land \\
& \dbhotels(\ahotelid, \ahigherprice, \alowerprice) \land \\
& (cid = \ahotelid \rightarrow \ahotelprice = \alowerprice ) \land \\
& (cid \neq \ahotelid \rightarrow \ahotelprice = \ahigherprice ) \land \\
& (\anewamountpaid = 0)
\end{align*}

\vspace*{2mm}
\noindent
\textbf{AlsoBookHotel}: This task handles payment of hotel reservation made after the flight is purchased.
It can be opened if $\ahotelid \neq \anull$ and $\astatus = 
\text{``Paid''} $ in \textbf{AddHotel}.
It receives input variables $\ahotelprice$ and $\aamountpaid$ from the parent and has local numeric variables 
$\anewamountpaid$ and $\ahotelamountpaid$. 
It has a single service \emph{Pay} which processes the payment. This service simply receives
a hotel payment in variable $\ahotelamountpaid$ and the new total amount of payment received is calculated 
($\anewamountpaid \\ = \aamountpaid + \ahotelamountpaid$).
The service can fail and the user can retry for unlimited number of times.
This task can return only when the payment is successful, which means that the closing condition is 
$\ahotelamountpaid = \ahotelprice$.
When \textbf{AlsoBookHotel} returns, the numeric variable $\anewamountpaid$ is returned to \textbf{ManageTrips}.

\vspace*{2mm}
\noindent
\textbf{BookInitialTrip}: This task allows the customer to reserve and pay for the chosen trip. 
Its opening condition is $\astatus = \text{``Unpaid''}$. This task has the following variables:

\vspace*{-1mm}
\begin{itemize}\itemsep=0pt\parskip=0pt
\item ID variables:  \uline{$\aflightid$}, \uline{$\ahotelid$}
\item numeric variables: \uwave{$\astatus$}, \uwave{$\aamountpaid$}, $\aticketprice$, $\ahotelprice$
\end{itemize}
\vspace*{-1mm}
The task contains a single service \emph{Pay} to process the payment, which can fail and be retried for
an unlimited number of times. Note that if the trip contains both the flight and hotel,
when \emph{Pay} is called, the payments for both of them are received.

If the payment is successful (i.e. $\aamountpaid$ equals to the flight price plus the hotel price), 
$\astatus$ is set to ``Paid''. Otherwise it is set to ``Failed''.
The closing condition of this task is $\astatus = \text{``Paid''}$ or $\astatus = \text{``Failed''}$.
When \textbf{BookInitialTrip} returns, $\astatus$ and $\aamountpaid$ in the parent task are updated by the
new $\astatus$ and $\aamountpaid$ returned by \textbf{BookInitialTrip}.
The \emph{Pay} service is specified as follows:

\vspace*{2mm}
\noindent
\emph{Pay}: \\
Pre-condition: $\ahotelid \neq \anull \lor \aflightid \neq \anull$ \\
Post-condition: 
\vspace*{-2mm}
\begin{align*}
& \exists cid \exists p_1 \exists p_2 \\
& (\aflightid = \anull \rightarrow \aticketprice = 0 \land cid = \anull) \land \\
& (\aflightid \neq \anull \rightarrow  \dbflights( \aflightid, \aticketprice, \\
& cid)) \land (\ahotelid = \anull \rightarrow \ahotelprice = 0) \land \\ 
& (\ahotelid \neq \anull \rightarrow (\dbhotels(\ahotelid, p_1, p_2) \land \\
& (\ahotelid = cid \rightarrow \ahotelprice = p_2) \land \\
& (\ahotelid \neq cid \rightarrow \ahotelprice = p_1)) \land \\
& (\aamountpaid = \aticketprice + \ahotelprice \rightarrow \\
& \astatus = \text{``Paid''}) \land (\aamountpaid \neq \aticketprice + \\
& \ahotelprice \rightarrow  \astatus = \text{``Failed''})
\end{align*}
%
%\begin{align*}
%& \exists a_1 \exists a_2 \exists \aticketprice \exists aid \exists d_1 \exists d_2 \exists cid
%\exists hn \exists p_1 \exists p_2 \exists hp \\
%& \dbflights(\aflightid, a_1, a_2, \adatefrom, \aticketprice, aid) \land \\
%& \dbhotels(\ahotelid, hn, p_1, p_2) \land \dbdates(\adatefrom, d_1) \land \\
%& \dbdates(\adateto, d_2) \land (d_1 < d_2) \land \\ 
%& (\dbcompatible(cid, aid, \ahotelid) \rightarrow hp = p_2 * (d_2 - d_1)) \land \\
%& (\neg \dbcompatible(cid, aid, \ahotelid) \rightarrow hp = p_1 * (d_2 - d_1)) \land \\
%& (\aamountpaid = \aticketprice + hp \rightarrow \astatus = \text{``Paid''}) \land \\
%& (\aamountpaid \neq \aticketprice + hp \rightarrow \astatus = \text{``Failed''})
%\end{align*}
%$(\astatus = \text{``Paid''} \land \aamountpaid = \aticketprice + \ahotelprice) \lor \astatus = \text{``Failed''}$

\vspace*{2mm}
\noindent
\textbf{Cancel}: In this task, the customer can cancel the flight and/or hotel 
after the trip has been paid for. Its opening condition is $\astatus = \text{``Paid''}$. 
This task has the following variables:
\vspace*{-1mm}
\begin{itemize}\itemsep=0pt\parskip=0pt
\item ID variables:  \uline{$\ahotelid$} and \uline{$\aflightid$}
\item numeric variables: \uline{$\aamountpaid$}, $\aticketprice$, \\
$\alowerprice$, $\ahigherprice$, $\ahotelprice$, \\
$\aamountrefunded$ and \uwave{$\astatus$}
\end{itemize}
\vspace*{-1mm}

The task has 3 services, \emph{CancelFlight}, \emph{CancelHotel} and \emph{CancelBoth} which 
cancel the flight, the hotel reservation, or both of them, respectively.
When any of these services is called, $\aamountrefunded$ is calculated to be the correct amount
needs to be refunded to the customer and $\astatus$ is set to ``FlightCanceled'', ``HotelCanceled'' and ``AllCanceled'' respectively. 
In particular, if the customer would like to cancel the flight 
while keeping the hotel reservation, and if a discount has been applied on the hotel reservation,
then the correct $\aamountrefunded$ equals to $\aticketprice$ minus the difference between 
the normal cost and the discounted cost of the hotel since she is no longer eligible for the discount.

The closing condition of this task is True. 
We show the specification of \emph{CancelFlight} as an example. 
Let $\mathtt{Discounted}$ be the subformula 
$$ (\ahotelid \neq \anull) \land (\ahotelprice = \alowerprice)$$
And let $\mathtt{Penalized}$ be the subformula
\begin{align*}
\aamountrefunded = & \ \aticketprice \ - \\
& (\ahigherprice - \alowerprice) 
\end{align*}
%\begin{align*}
%&\exists hn \exists p_1 \exists p_2 \exists d_1 \exists d_2 \ \dbhotels(\ahotelid, hn, p_1, p_2) \land \\
%&\dbdates(\adatefrom, d_1) \land \dbdates(\adateto, d_2) \land \\
%& (\aamountrefunded = \aticketprice - (d_2 - d_1) * \\
%& (p_1 - p_2))
%\end{align*}
%\vspace*{2mm}
\noindent
\emph{CancelFlight}: \\
Pre-condition: 
\begin{align*}
& \aflightid \neq \anull \land \astatus \neq \text{``FlightCanceled''} \land \\
& \astatus \neq \text{``HotelCanceled''} \land \astatus \neq \text{``AllCanceled''}
\end{align*}
Post-condition:
\begin{align*}
& \exists cid \ \dbflights(\aflightid, \aticketprice, cid) \land \\
& (\ahotelprice = \aamountpaid - \aticketprice) \land \\
& (\ahotelid \neq \anull \rightarrow \\
& (\dbhotels(\ahotelid, \ahigherprice, \alowerprice) \land \\
& (\neg \mathtt{Discounted} \rightarrow \aamountrefunded = \aticketprice) \land \\
& (\mathtt{Discounted} \rightarrow \mathtt{Penalized}) \land \astatus = \text{``FlightCanceled''}
\end{align*}

%\victor{My comment here seems to have disappeared. There is a problem with the formula
%$\mathtt{Discounted} \rightarrow \mathtt{Penalized}$ because it needs universal quantification. See also related
%comment on property. Maybe a solution would be to add a flag variable to indicate $\mathtt{Discounted}$ or $\mathtt{Penalized}$ ? }
%\yuliang{I fixed it in a different way by adding extra variables. The existential quantifiers are removed from the subformulas now.}

\subsection{Example HLTL-FO Property} \label{app:example-hltl}
\eat{
Let $T_1$ be a root task with child tasks $T_2$ and $T_3$.
The $\hltlfo$ formula (with no global variables)
$$
\varphi = [ 
~\mathbf{F}[\psi_2]_{T_2} 
\rightarrow 
\mathbf{G} (\sigma^o_{T_3} \rightarrow [\psi_3]_{T_3})
]_{T_1}
$$
states that whenever $T_1$ calls child task $T_2$ 
and $T_2$'s local run satisfies property $\psi_2$, then if $T_3$ is also called (via the opening service $\sigma^o_{T_3}$), its local run
must satisfy property $\psi_3$. }

Suppose we wish to enforce the following policy: \emph{if a discount is applied
to the hotel reservation, then a compatible flight must be purchased without cancellation}. 
One typical way to defeat the policy would be for a user to first pay for the
flight, then reserve the hotel with the discount price, but next cancel the flight
without penalty. Detecting such bugs can be subtle, especially in a system allowing concurrency.
The following $\hltlfo$ property of task \textbf{ManageTrips} says 
``If \emph{AddHotel} is called and a hotel reservation is added with a discounted price,
then at the task \emph{Cancel}, if the customer would like to cancel the flight, a penalty must be paid''.

The property is specified as
$[\varphi]_{\mathtt{T1}}$ where $\varphi$ is the formula:
\begin{align*}
& \varphi = \mathbf{F} [\mathbf{F} \ (\mathtt{Discounted} \land \mathbf{X} \ \sigma_{\mathtt{T6:AlsoBookHotel}}^o) ]_{\mathtt{T2:AddHotel}} \rightarrow \\
& \mathbf{G}(\sigma^o_{\mathtt{T5:Cancel}} \rightarrow [ \mathbf{G} ( \textit{CancelFlight} \ \rightarrow \mathtt{Penalized})]_{\mathtt{T5:Cancel}})
\end{align*}

%\begin{align*}
%\varphi = & \ \mathbf{G} ( [\mathbf{F} \ (\mathtt{Discounted} \land \mathbf{X} \ \sigma_{\mathtt{T7:AlsoBookHotel}}^o) ]_{\mathtt{T3:AddHotel}} \rightarrow \\
%& ((\astatus = \text{``Paid''})  \land \neg \mathbf{F} [ \mathbf{F} ( \textit{CancelFlight} \ \land \\
%& \astatus = \text{``FlightCanceled''} ]_{\mathtt{T6:Cancel}}   )))
%\end{align*}
\noindent
with the subformulas $\mathtt{Discounted}$ and $\mathtt{Penalized}$ defined above. 
%Recall that $\hltlfo$ formulas are of the form $[\xi]_{\mathtt{T1}}$. 

%\victor{Why do we need the condition $\astatus = \text{``Paid''}$ ? Doesn't this necessarily hold in order for $T_3$ to give the discount?
%Also, why does the violating cancelation have to be initiated after the call to $T_3$? Can't $T_6$ have been called before $T_3$ but not yet returned
%at the time $T_3$ is called?}
%\yuliang{I changed the formula. Can you take a look?}
%\alin{Looks good.}
%\victor{It seems that the correction I made disappeared (was there a lock problem??). The righthand side of the implication should be
%$\mathbf{G} (\sigma^o_{\mathtt{T6:Cancel}} \rightarrow [ \mathbf{G} ( \textit{CancelFlight} \rightarrow \mathtt{Penalized})]_{\mathtt{T6:Cancel}}$
%} \yuliang{fixed.}
%\victor{But there is a more serious problem: you cannot use existentially quantified conditions in the property (such as Discounted or Penalized).
%This is because properties need to be closed under negation and existential conditions become universal. 
%Unlike existential quantification in services, this cannot be simulated by additional variables. 
%So the property needs to stick to the definition which disallows any sort of quantification except for the global variables.}
%\yuliang{See above.}

Notice that in the specification there is no guard preventing \emph{AddHotel} and 
\emph{Cancel} to run concurrently after a successful payment is made, which can lead to a violation of this 
property. The problem can be fixed by adding a new variable in \textbf{ManageTrips} to
indicate whether \textbf{AddHotel} or \textbf{Cancel} are currently running and
modifying their opening conditions to make sure that these two tasks are mutual exclusive.

\section{Framework and HLTL-FO}
\label{app:hltl}

\subsection{Definition of global run}

The global runs of a HAS $\Gamma$ are obtained from interleavings of the transitions in a tree of local runs, lifted to 
transitions over instances of $\cala$. We make this more precise. 
Let $D$ be a database and \Tree a full tree of local runs over $D$. 
For a local run $\rho = (\nu_{in}, \nu_{out}, \{(I_m,\sigma_m)\}_{m < \gamma})$  (where $I_m = (\nu_m,S_m)$) and $i < \gamma$,
we denote by $\sigma(\rho,i) = \sigma_i$, $\nu(\rho,i) = \nu_i$, and $S(\rho,i) = S_i$.
Let $\preceq$ be the pre-order on the set
$\{(\rho, i) \mid \rho \in \Tree, 0 \leq i < \gamma(\rho)\}$ 
defined as the smallest reflexive-transitive relation containing the following:
\begin{enumerate}\itemsep=0pt\parskip=0pt
\item for each node $\rho$ and $0 \leq i \leq j  < \gamma(\rho)$, 
$(\rho, i) \preceq (\rho, j)$ 
\item for each edge in \Tree from $\rho_T$ to $\rho_{T_c}$ labeled $i$, $(\rho_T, i) \preceq (\rho_{T_c}, 0)$
and $(\rho_{T_c}, 0) \preceq (\rho_T, i)$. Additionally, if $\rho_{T_c}$ is returning and $m$ is the smallest
$j > i$ for which $\sigma(\rho_T, j) = \sigma_{T_c}^c$, then $(\rho_{T_c}, \gamma(\rho_{T_c})) \preceq (\rho_T, m)$
and $(\rho_T, m) \preceq (\rho_{T_c}, \\ \gamma(\rho_{T_c}))$.
\end{enumerate} 

Let $\sim$ be the equivalence relation induced by $\preceq$ (i.e., $a \sim b$ iff $a \preceq b$ and $b \preceq a$).
Note that all classes of $\sim$ are singletons except for the ones induced by (2), which are of the form 
$\{(\rho_1, i),(\rho_2, j)\}$ where $\sigma(\rho_1,i) = \sigma(\rho_2, j) \in \{\sigma_T^o, \sigma_T^c\}$ for some task $T$. 
For an equivalence class $\varepsilon$ 
of $\sim$ we denote by $\sigma(\varepsilon)$ the unique service of elements in $\varepsilon$. 
A {\em linearization} of $\preceq$ is an enumeration of the equivalence classes of $\sim$ consistent with $\preceq$. 
Consider a linearization $\{\varepsilon_i\}_{i \geq 0}$ of $\preceq$. Note that  
$\varepsilon_0 = (\rho_{T_1}, 0)$ and let $\nu(\rho_{T_1},0) = \nu_0$.
A global run induced by $\{\varepsilon_i\}_{i \geq 0}$ is
a sequence $\rho = \{ (\bar I_i, \sigma_i) \}_{i \geq 0}$ such that $\sigma_i = \sigma(\varepsilon_i)$ and
each $\bar I_i$ is an instance $(\bar \nu_i, stg_i, D, \bar S_i)$ of $\cala$, 
defined inductively as follows. For $i = 0$,
\begin{itemize}
\item $\bar \nu_0(\bar x^{T_1}) = \nu_0(\bar x^{T_1})$ (and arbitrary on other variables)
\item $stg_0 = \{T_1 \mapsto \aactive, T_i \mapsto \ainit \mid 2 \leq i \leq k\}$ 
\item $\bar S_0 = \{S^{T_i} \mapsto\emptyset \mid 1 \leq i \leq k\}$.
\end{itemize}
For $i > 0$, $\bar I_i$ is defined as follows. 
Suppose first that $\varepsilon_i = \{(\rho, j)\}$ where 
$\rho$ 
%= (\nu_{in}, \nu_{out}, \{(I_m,\sigma_m)\}_{m < \gamma})$
is a local run of task $T$ %$\sigma_j$  is an internal service of $T$, and $I_j = (\nu_j, S_j)$.  
and $\sigma(\rho,j)$ is an internal service of $T$.
Then $\bar\nu_i = \bar \nu_{i-1}[\bar x^T \mapsto \nu(\rho,j)(\bar x^T)]$,
$\bar S_i = \bar S_{i-1}[S^T \mapsto S(\rho,j)]$, and $stg_i = stg_{i-1}[\bar T \mapsto \ainit \mid \bar T \in \emph{desc}(T)]$.
Now suppose $\varepsilon = \{(\rho_T, j), (\rho_{T_c}, 0)\}$, where $T_c$ is a child of $T$, $\rho_T$ and $\rho_{T_c}$ are local runs of $T$ and $T_c$, 
and  $\sigma(\varepsilon) =  \sigma^o_{T_c}$. Then $\bar \nu_i = \bar \nu_{i-1}[\bar x^{T_c} \mapsto \nu(\rho_{T_c},0)(\bar x^{T_c})]$, $\bar S_i = \bar S_{i-1}[S^{T_c} \mapsto \emptyset]$,
and $stg_i = stg_{i-1}[T_c \mapsto \aactive]$. Finally, suppose $\varepsilon = \{(\rho_T, j), (\rho_{T_c}, \gamma-1)\}$ 
where $\sigma(\varepsilon) =  \sigma^c_{T_c}$.
Then $\bar \nu_i = \bar \nu_{i-1}[\bar x^T \mapsto \nu(\rho_T,j)(\bar x^T)]$, 
$stg_i = stg_{i-1}[T_c \mapsto \aclosed]$, and $\bar S_i = S_{i-1}[S^{T_c} \mapsto \emptyset]$. 

We denote by  $\call(\Tree)$ the set of global runs induced by linearizations of $\preceq$. 
The set of global runs of $\Gamma$ on a database $D$ is 
$\emph{Runs}_D(\Gamma) = \bigcup \{\call(\Tree) \mid $
$\Tree$ is a full tree of local runs of $\Gamma \mbox{ on } D\}$ and
the set of global runs of $\Gamma$ is $\emph{Runs}(\Gamma) = \bigcup_D \emph{Runs}_D(\Gamma)$.

\subsection{Review of LTL}
\label{app:ltl}

We review the classical definition of linear-time temporal logic (LTL) over a set $P$ of propositions.
LTL specifies properties of infinite words ($\omega$-words)
$\{\tau_i\}_{i \geq 0}$ over the alphabet consisting of truth assignments to $P$.
Let $\tau_{\geq j}$ denote $\{\tau_i\}_{i \geq j}$, for $j \geq 0$.

The meaning of the temporal operators {\bf X}, {\bf U} is the following (where $\models$ denotes satisfaction
and $j \geq 0$):
\begin{itemize}
\item $\tau_{\geq j} \models {\bf X} \varphi$ iff $\tau_{\geq j+1} \models \varphi$,
\item $\tau_{\geq j} \models \varphi {\bf U} \psi$ iff $\exists k \geq j$ such that $\tau_{\geq k} \models \psi$
and $\tau_{\geq l} \models \varphi$ for $j \leq l < k$.
\end{itemize}
Observe that the above temporal operators can simulate all commonly used
operators, including  {\bf G} (always) and {\bf F} (eventually).
Indeed, ${\bf F} \varphi \equiv \mbox{\em true}~{\bf U}~\varphi$ and
${\bf G} \varphi ~\equiv~ \neg ({\bf F} \neg \varphi)$.

The standard construction of a B\"{u}chi automaton $B_\varphi$
corresponding to an LTL formula $\varphi$ is given in \cite{VW:LICS:86,SVW87}.
The automaton $B_\varphi$ has exponentially many states and accepts
precisely the set of $\omega$-words that satisfy $\varphi$.

It is sometimes useful to apply LTL on {\em finite} words rather than $\omega$-words.
The finite semantics we use for temporal operators is the following \cite{ltl-finite}.
Let $\{\tau_i\}_{0 \leq i \leq n}$ a finite sequence of
truth values of $P$. Similarly to the above, let $\tau_{\geq j}$ denote $\{\tau_i\}_{j \leq i \leq n}$, for $0 \leq j \leq n$.
The semantics of {\bf X} and {\bf U} are defined as follows:
\begin{itemize}\itemsep=0pt\parskip=0pt
\item $\tau_{\geq j} \models {\bf X} \varphi$ iff $n > j$ and $\tau_{\geq j+1} \models \varphi$,
\item $\tau_{\geq j} \models \varphi {\bf U} \psi$ iff $\exists k,  j \leq k \leq n$ such that $\tau_{\geq k} \models \psi$
and $\tau_{\geq l} \models \varphi$ for $j \leq l < k$.
\end{itemize}

It is easy to verify that for the $B_\varphi$ obtained by the standard construction \cite{VW:LICS:86,SVW87}
there is a subset $Q^{\emph{fin}}$ of its states such that $B_\varphi$
viewed as a finite-state automaton with final states $Q^{\emph{fin}}$
accepts precisely the finite words that satisfy $\varphi$.

\subsection{Proof of Theorem \ref{thm:ltlfo} } \label{app:ltlfo}
We show that it is undecidable whether a HAS $\Gamma = \langle \cala, \Sigma, \Pi \rangle$ 
satisfies an LTL formula over $\Sigma$. 
%%\yuliang{need definition of $\ltlfo$}
%We consider the following simple variant of $\ltlfo$. 
%Given a Hierarchical artifact system $\Gamma = \langle \cala, \Sigma, \Pi \rangle$, 
%we define that an $\ltlfo$ formula $\varphi$ over $\Gamma$ is an LTL formula with propositions $\Sigma$.
%A global run $\rho = \{(\bar{I}_i, \sigma_i)\}_{i \geq 0}$ of $\Gamma$ satisfies $\varphi$
%iff the sequence $\{\sigma_i\}_{i \geq 0} \models \varphi$.
%Notice that the above definition of $\ltlfo$ is the same as having the set of propositions 
%$P$ being interpreted as services in $\Sigma$ only. $\ltlfo$ formulas in this definition can be simulated by
%$\ltlfo$ formulas under the normal definition with $P$ interpreted as 
%condition of the forms $(x = c)$ where $x$ is a variable and $c$ is a constant. 
The proof is by reduction from the repeated state reachability problem of VASS with reset arcs and bounded lossiness (RB-VASS) \cite{mayr2003undecidable}. 
An RB-VASS extends the VASS reviewed in Section \ref{sec:verification} as follows.
In addition to increment and decrement of the counters, an action of RB-VASS also allows resetting the values of some counters to 0.
And after each transition, the value of each counter can decrease non-deterministically by an integer value bounded by some constant $c$.
The results in \cite{mayr2003undecidable} (Definition 2 and Theorem 18) indicate that the repeated state reachability problem for RB-VASS is undecidable for every fixed $c \geq 0$, since the structural termination problem for Reset Petri-net with bounded lossiness can be reduced to 
the repeated state reachability problem for RB-VASS's. In our proof, we use RB-VASS's with $c = 1$. \yuliang{added connection to the cited paper.}
%The repeated state reachability problem for RB-VASS is undecidable ,
%as indicated by the results in \cite{mayr2003undecidable}. 
%\yuliang{as implied by Definition 2 and Theorem 15 there}

Formally, a RB-VASS $\calv$ (with lossiness bound $1$ and dimension $d > 0$) is a pair $(Q, A)$ where $Q$ is a finite set of states and $A$ is a set of actions
of the form $(p, \bar{a}, q)$ where $\bar{a} \in \{-1, +1, r\}^d$, and $p, q \in Q$. 
A run of $\calv = (Q, A)$ is a sequence $(q_0, \bar{z}_0), \dots (q_n, \bar{z}_n)$ where $\bar{z}_0 = \bar{0}$
and for each $i \geq 0$, $q_i \in Q$, $\bar{z}_i \in \mathbb{N}^d$, and for some $\bar{a}$ such that 
$(q_i, \bar{a}, q_{i+1}) \in A$, and for $1 \leq j \leq d$:
\begin{itemize}
\item if $\bar{a}(j) \in \{-1, +1\}$, then $\bar{z}_{i+1}(j) = \bar{z}_{i}(j) + \bar{a}(j)$ or 
$\bar{z}_{i+1}(j) = \bar{z}_{i}(j) + \bar{a}(j) - 1$, and
\item if $\bar{a}(j) = r$, then $\bar{z}_{i+1}(j) = 0$.
\end{itemize}
For a given RB-VASS $\calv = (Q, A)$ and a pair of states $q_0, q_f \in Q$, we say that $q_f$ is repeatedly reachable from $q_0$
if there exists a run $(q_0, \bar{z}_0) \dots (q_n, \bar{z}_n) \dots (q_m, \bar{z}_m)$ of $\calv$ such that
$q_n = q_m = q_f$ and $\bar{z}_n \leq \bar{z}_m$. As discussed above, 
checking whether $q_f$ is repeatedly reachable from $q_0$ is undecidable.

We now show that for a given RB-VASS $\calv = (Q, A)$ and $(q_0, q_f)$,
one can construct a HAS $\Gamma = \langle \cala, \Sigma, \Pi \rangle$ 
and LTL property $\Phi$ over $\Sigma$ such that $q_f$ is repeatedly reachable from $q_0$ iff $\Gamma \models \Phi$.
At a high level, the construction of $\Gamma$ uses $d$ tasks to simulate the $d$-dimensional vector of counters.
Each task is equipped with an artifact relation, and the number of elements in the artifact relation
is the current value of the corresponding counter. Increment and decrement the counters are simulated by
internal services of these tasks, and reset of the counters are simulated by closing and reopening the task 
(recall that this resets the artifact relation to empty).
Then we specify in the LTL formula $\Phi$ that the updates of the counters of the same action 
are grouped in sequence. Note that this requires coordinating the actions of sibling tasks, which is not possible in $\hltlfo$.
The construction is detailed next.

The database schema of $\Gamma$ consists of a single unary relation $R(\underline{id})$.
The artifact system has a root task $T_1$ and subtasks 
$\{P_0, P_1, \dots, P_d, C_1, \dots, C_d\}$ which form the following tasks hierarchy:

\begin{figure}[!ht]
\centering
\includegraphics[scale=0.7]{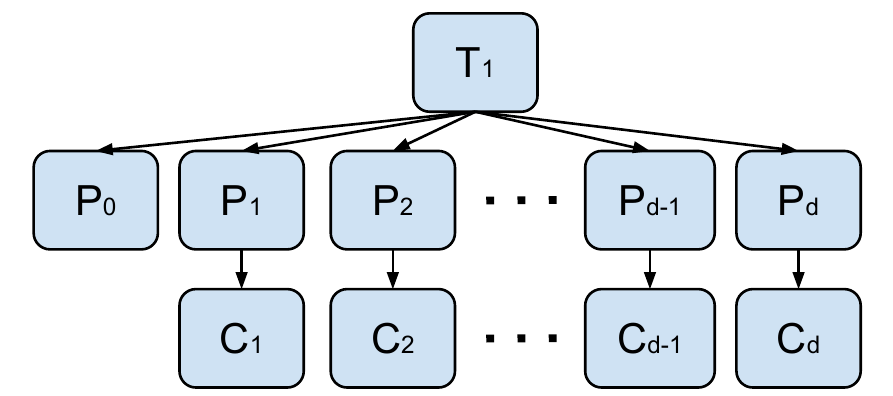}
\caption{Tasks Hierarchy} 
\label{fig:ltl}
\end{figure}

The tasks are defined as follows. 
The root task $T_1$ has no variables nor internal services. 
The task $P_0$ contains a numeric variable $s$, indicating the current state of the RB-VASS.
For each $q \in Q$, $P_0$ has a service $\sigma^q$, 
whose pre-condition is true and post-condition sets $s$ to $q$.

For $i \geq 1$, task $P_i$ has no variable. 
It has a single internal service $\sigma_i^r$ whose pre- and post-conditions are both $\mathtt{true}$.

Each $C_i$ has an ID variable $x$, an artifact relation $S_i$ and a pair of services $\sigma^+_i$ and $\sigma^-_i$, which 
simply insert $x$ into $S_i$ and removes an element from $S_i$, respectively.
Intuitively, the size of $S_i$ is the current value of the $i$-th counter.
Application of service $\sigma_i^r$ corresponds to resetting the $i$-th counter.
And application of services $\sigma_i^+$ and $\sigma_i^-$ correspond to 
increment and decrement of the $i$-th counter, respectively.

Except for the closing condition of $T_1$, all opening and closing conditions of tasks are $\mathtt{true}$.

We encode the set of actions $A$ into an LTL formula as follows.
For each state $p \in Q$, we denote by $\alpha(p)$ the set of actions starting from $p$.
For each action $\alpha = (p, \bar{a}, q) \in A$, we construct an LTL formula $\varphi(\alpha)$ as follows. 
First, let $\phi_1, \dots \phi_d, \phi_{d+1}$ be LTL formulas where:
\begin{itemize}
\item $\phi_{d+1} = \mathbf{X} \sigma^q$,
\item for $i = d, d - 1, \dots, 1$: 
\begin{itemize}
\item if $\bar{a}(i) = +1$, then $\phi_i =  \sigma_i^+ \land \mathbf{X} \phi_{i+1}$,
\item if $\bar{a}(i) = -1$, then $\phi_i = (\sigma_i^- \land \mathbf{X} \phi_{i+1})\lor (\sigma_i^- \land \mathbf{X} (\sigma_i^- \land \mathbf{X} \phi_{i+1}))$, and
\item if $\bar{a}(i) = r$, then $\phi_i = \sigma_i^c \land \mathbf{X} (\sigma^r_i \land \mathbf{X} (\sigma^o_i \land \mathbf{X} \phi_{i+1}))$
where $\sigma_i^o$ and $\sigma_i^c$ are the opening and closing services of task $C_i$.
\end{itemize}
\end{itemize}
Let $\varphi(\alpha) = \mathbf{X} \phi_1$. Intuitively, $\varphi(\alpha)$ specifies a sequence of service calls
that update the content of the artifact relations $S_1, \dots S_d$ according to the vector $\bar{a}$. 
In particular, for $\bar{a}(i) = r$, the subsequence of services $\sigma_i^c \sigma_i^r \sigma_i^o$ 
first closes task $C_i$ then reopens it. This empties $S_i$.
For $\bar{a}(i) = +1$, by executing $\sigma_i^+$, the size of $S_i$ might be increased by 1 or 0, 
depending on whether the element to be inserted is already in $S_i$.
And for $\bar{a}(i) = -1$, we let $\sigma_i^-$ to be executed either once or twice, 
so the size of $S_i$ can decrease by 1 or 2 nondeterministically.
Then we let
$$ \Phi = \Phi_{\mathtt{init}} \land \bigwedge_{p \in Q} \mathbf{G} \left( \sigma^p \rightarrow  \bigvee_{\alpha \in \alpha(p)} \varphi(\alpha) \right) \land \mathbf{G F} \sigma^{q_f}$$
where $\Phi_{\mathtt{init}}$ is a formula specifying that the run is correctly initialized, which simply means 
that the opening services $\sigma_T^o$ of all tasks are executed once at the beginning of the run, and then a $\sigma^{q_0}$ is executed.

The second clause says that for every state $p \in Q$, whenever the run enters a state $p$ (by calling $\sigma^p$),
a sequence of services as specified in $\varphi(\alpha)$ is called to update $S_1, \dots, S_k$, simulating 
the action $\alpha$ that starts from $p$.

Finally, the last clause $\mathbf{G F} \sigma^{q_f}$ guarantees that the service $\sigma^{q_f}$ is applied infinitely often, 
which means that $q_f$ is reached infinitely often in the run.

We can prove the following lemma, which implies Theorem \ref{thm:ltlfo}:

\begin{lemma}
For RB-VASS $(Q, A)$ and states $q_0, q_f \in Q$, there exists a run $(q_0, \bar{z}_0), \dots, (q_m, \bar{z}_m), \dots, (q_n, \bar{z}_n)$ 
of $(Q, A)$ where $q_m = q_n = q_f$ and $\bar{z}_m \leq  \bar{z}_n$ 
\victor{Added the inequality.} 
iff there exists a global run $\rho$ of $\Gamma$ such that $\rho \models \Phi$.
\end{lemma}

\subsection{Expressiveness of HLTL-FO}
\label{sec:complete}

We next show that $\hltlfo$ expresses, in a reasonable sense, all interleaving-invariant $\ltlfo$ properties.
We consider a notion of interleaving-invariance of $\ltlfo$ formulas based on their
propositional structure, rather than the specifics of the propositions' interpretation (which may lead to ``accidental'' invariance).
In view of Lemma \ref{lem:noglobal}, we consider only formulas with no global variables or set atoms.
We first recall the logic $\ltlfo$, slightly adapted to our context.
Let $\Gamma = \langle \cala, \Sigma, \Pi \rangle$ be a HAS where
$\cala = \langle \calh, \mathcal{DB} \rangle$.
An $\ltlfo$ formula $\varphi_f$ over $\Gamma$ consists of an LTL formula $\varphi$ with propositions $P \cup \Sigma$
together with a mapping $f$ associating to each $p \in P$ a condition over $\bar x^T$ for some $T\ \in \calt$
(and we say that $f(p)$ is over $T$) .
Satisfaction of $\varphi_f$ on
a global run $\rho = \{(I_i,\sigma_i)\}_{i \geq 0}$ of $\Gamma$ on database $D$,
where $I_i = (\nu_i, stg_i, D, S_i)$,
is defined as usual, modulo the following:
\begin{itemize}\itemsep=0pt\parskip=0pt
%\item $active_T$ holds in $(\rho_i,\sigma_i)$ if $stg_i(T) = \emph{active}$
\item $f(p)$ over $T$ holds in $(I_i,\sigma_i)$ iff $stg_i(T) = \aactive$ and the condition
$f(p)$ on $\nu_i(\bar x^T)$ holds;
%\item proposition $S_T$ holds in $(\rho_i,\sigma_i)$ iff $stg_i(T) = \emph{active}$ and $S_T(\nu_i(\bar x^T))$ holds;
\item proposition $\sigma$ in $\Sigma$ holds in $(I_i, \sigma_i)$ if $\sigma = \sigma_i$.
\end{itemize}
Thus, the information about $(I_i,\sigma_i)$ relevant to satisfaction of $\varphi_f$ consists of
$\sigma_i$, the stage of each task (active or not), and the truth values in $I_i$ of $f(p)$ for $p \in P$.

We now make more precise the notion of (propositional) invariance under interleavings.
Consider an $\ltlfo$ formula $\varphi_f$ over $\Gamma$. Invariance under interleavings is a property of
the propositional formula $\varphi$ (so independent on the interpretation of propositions provided by $f$). Let $P \cup \Sigma$
be the set of propositions of $\varphi$ and let $P_T$
denote the subset of $P$ for which $f(p)$ is a condition over $\bar x^T$. Thus, $\{P_T \mid T \in \calt\}$
is a partition of $P$.
We define the set $\call(\Gamma)$ of $\omega$-words associated to $\Gamma$, on which $\varphi$ operates.
The alphabet, denoted ${\bf A}(\Gamma)$, consists of all triples $(\kappa, stg, \sigma)$ where $\sigma \in \Sigma$,
$\kappa$ is a truth assignment to the propositions in $P$, and $stg$ is a mapping associating to each $T \in \calt$
its stage ($\aactive, \ainit,$ or $\aclosed$).
An $\omega$-word $\{(\kappa_i, stg_i,  \sigma_i)\}_{i \geq 0}$ over ${\bf A}(\Gamma)$ is in $\call(\Gamma)$
if the following hold:
\begin{enumerate}\itemsep=0pt\parskip=0pt
\item for each $i > 0$, if $\sigma_i \in \Sigma^\delta_T$, then
$\kappa_i$ and $\kappa_{i-1}$ agree on all $P_{\bar T}$ where $\bar T \neq T$;
\item the sequence of calls, returns, and internal services obeys the conditions on service sequences in global runs of $\Gamma$;
\item for each $i > 0$ and $T \in \calt$, $stg_i(T)$ is the stage of $T$ as determined by the sequence
of calls and returns in $\{\sigma_j\}_{j < i}$.
\end{enumerate}
The formal definition of (2) and (3) mimic closely the analogous definition of global runs of
HAS's  (omitted).
Consider an $\omega$-word $u = \{(\kappa_i, stg_i, \sigma_i)\}_{i \geq 0}$ in $\call(\Gamma)$.
We define the partial order $\preceq_u$ on $\{ i \mid i \geq 0\}$ as the reflexive-transitive closure of the
relation consisting of all pairs $(i,j)$ such that $i < j$ and for some $T$,
$\sigma_i, \sigma_j \in \Sigma^{obs}_T$.
Observe that $0$ is always the minimum element in $\preceq_u$.
A linearization of $\preceq_u$ is a total order on $\{ i \mid i \geq 0\}$ containing $\preceq_u$.
One can represent a linearization of $\preceq_u$ as a sequence
%$\{i_j \mid j \geq 0 i_{j+1} \not\preceq_u i_j\}$.  
$\{i_j \mid j \geq 0\}$ such that $i_n \preceq_u i_m$ implies that $n \leq m$.
For each such linearization $\alpha$ we define the $\omega$-word
$u_\alpha  = \{(\bar \kappa_j, \overline{\emph{stg}}_j, \sigma_{i_j})\}_{ j \geq 0}$
in $\call(\Gamma)$ as follows. The stage function is the one determined by the sequence of services.
The functions $\bar \kappa_j$ are defined by induction as follows:
\begin{itemize}\itemsep=0pt\parskip=0pt
\item $\bar \kappa_{0} = \kappa_0$;
\item if $j >0$ and $\sigma_{i_j} \in \Sigma^\delta_T$
then $\bar \kappa_{j} = \bar \kappa_{j-1}[P_T \mapsto \kappa_{i_j}(P_T)]$
%on all $P_{\bar T}$ where $\bar T \neq T$
%and $\bar \kappa_{j}|_{P_T} = \kappa_{i_j}|_{P_T}$.  
\end{itemize}
Intuitively, $u_\alpha$ is obtained from $u$ by commuting actions that are incomparable with respect to $\preceq_u$,
yielding the linearization $\alpha$.
We note that the relation $\preceq_u$ is the analog to our setting of Mazurkiewicz
traces, used in concurrent systems to capture dependencies among process actions \cite{mazur,mazurbook,gastin06}.

\vspace{-2mm}
\begin{definition}
An $\ltlfo$ formula $\varphi_f$ over $\Gamma$ is propositionally invariant with respect to interleavings if for every
$u \in \call(\Gamma)$ and linearization $\alpha$ of $\preceq_u$ ,
$~~u \models \varphi~~$  iff  $~~u_\alpha \models \varphi$.
\end{definition}

We can show the following. 

\begin{theorem} \label{thm:complete}
$\hltlfo$ expresses precisely the $\ltlfo$ \\ properties of HAS's that are propositionally invariant with respect to interleavings.
\end{theorem}

We next sketch the proof.
For conciseness, we refer throughout the proof to {\em propositionally} interleaving-invariant $\ltlfo$ simply as interleaving-invariant $\ltlfo$.

Showing that $\hltlfo$ expresses only interleaving-invariant $\ltlfo$ properties is straightforward.
The converse however is non-trivial.
We begin by showing a normal form for LTL formulas, which facilitates the application to our context of results
from \cite{gastin04,gastin06} on temporal logics for concurrent processes.
Consider the alphabet ${\bf H}(\Gamma) = \{(\kappa, \sigma) \mid (\kappa, stg, \sigma) \in {\bf A}(\Gamma)\}$.
Thus, ${\bf H}(\Gamma)$ is ${\bf A}(\Gamma)$ with the stage information omitted.
Let $\calh(\Gamma) = h(\call(\Gamma))$
where $h((\kappa,stg,\sigma)) = (\kappa, \sigma)$. 
We define local-LTL to be LTL 
using the set of propositions P$\Sigma = \{(p,\sigma) \mid p \in P_T, \sigma \in \Sigma_T^{obs}\}$.
A proposition
$(p,\sigma)$ holds in $(\bar \kappa,\bar \sigma)$ iff $\bar \sigma = \sigma$ and $\bar \kappa(p)$ is true.    
The definition of interleaving-invariant local-LTL formula is the same as for LTL.

\begin{lemma} \label{lem:local-ltl}
For each interleaving-invariant LTL formula $\varphi$ over $\call(\Gamma)$ one can construct an interleaving-invariant local-LTL formula $\bar \varphi$ 
over $\calh(\Gamma)$
such that for every $u \in \call(\Gamma)$, $u \models \varphi$ iff $h(u) \models \bar \varphi$
where $h((\kappa,stg,\sigma)) = (\kappa, \sigma)$. 
\end{lemma}

\begin{proof}
We use the equivalence of FO and LTL over $\omega$-words \cite{kamp}. It is easy to see that each LTL formula $\varphi$ 
over $\call(\Gamma)$ can be translated into an FO formula $\psi(\varphi)$ over 
$\calh(\Gamma)$ using only propositions in P$\Sigma$, such that for every $u \in \call(\Gamma)$,
$u \models \varphi$ iff $h(u) \models \psi(\varphi)$. Indeed, it is straightforward to define by FO means 
the stage of each transaction in a given configuration, as well as each proposition in $P \cup \Sigma$ 
in terms of propositions in P$\Sigma$, on words in $\calh(\Gamma)$.  
One can then construct from the FO sentence $\psi(\varphi)$ an LTL formula $\bar \varphi$ 
equivalent to it over words in $\calh(\Gamma)$, using the same set of propositions P$\Sigma$.
The resulting LTL formula is thus in local-LTL, and it is easily seen that it is interleaving-invariant.
\end{proof}

We use a propositional variant HLTL of $\hltlfo$, 
defined over $\omega$-words in $\calh(\Gamma)$
similarly to $\hltlfo$.
More precisely, LTL formulas applying to transaction $T$ use propositions
in $P_T \cup \Sigma^{obs}_T$ and expressions $[\psi]_{T_c}$ where $T_c$ is a child of $T$ and $\psi$ is an HLTL formula applying to 
$T_c$. 

We show the following key fact.

\begin{lemma} \label{lem:hltl}
For each interleaving-invariant local-LTL formula over $\calh(\Gamma)$ there exists
an equivalent HLTL formula over $\calh(\Gamma)$. 
\end{lemma}

\begin{proof}
To show completeness of HLTL, we use a logic shown in \cite{gastin04,gastin06}
to be complete for expressing LTL properties invariant with respect to valid interleavings
of actions of concurrent processes (or equivalently, well-defined on Mazur-kievicz traces).
The logic, adapted to our framework, operates on partial orders $\preceq_u$ 
of words $u \in \calh(\Gamma)$, and is denoted LTL$(\preceq)$. 
For $u = \{(\kappa_i, \sigma_i) \mid i \geq 0\}$, we define the projection of $u$ on $T$ as the 
subsequence $\pi_T(u) = \{(\kappa_{i_j}|_{P_T}, \sigma_{i_j})\}_{j \geq 0}$ where $\{\sigma_{i_j} \mid j \geq 0\}$ 
is the subsequence of $\{\sigma_i \mid i \geq 0\}$ retaining all services in $\Sigma^{obs}_T$.
LTL$(\preceq)$ uses the set of propositions P$\Sigma$ and
the following temporal operators on $\preceq_u$:
\begin{itemize}
\item {\bf X}$_T \varphi$, which holds in $(\kappa_i,\sigma_i)$ if 
$\pi_T(v) \neq \epsilon$ 
for $v = \{(\kappa_j, \sigma_j) \mid j \geq m\}$, where $m$ is the minimum index such that $i \prec_u m$,
and $\varphi$ holds on  $\pi_T(v)$;
\item $\varphi$ {\bf U}$_T$ $\psi$, which holds in $(\kappa_i,\sigma_i)$ if 
$\pi_T(v) \neq \epsilon$ 
for $v = \{(\kappa_j, \sigma_j) \mid j \geq i\}$, and $\varphi$ {\bf U} $\psi$ holds
on $\pi_T(v)$.
\end{itemize}
From Theorem 18 in \cite{gastin04} and Proposition 2 and Corollary 26 in \cite{gastin06} it follows that 
LTL$(\preceq)$ expresses all local-LTL properties over $\calh(\Gamma)$ invariant with respect to interleavings. 

We next show that HLTL can simulate LTL$(\preceq)$.
To this end, we consider an extension of HLTL in which LTL$(\preceq)$ formulas may be used in addition to
propositions in $P_T \cup \Sigma^{obs}_T$ in every formula applying to transaction $T$. 
We denote the extension by HLTL+LTL$(\preceq)$. Note that for each formula $\xi$ in LTL$(\preceq)$, $[\xi]_{T_1}$ is 
an HLTL+LTL$(\preceq)$ formula. The proof consists in showing that the  LTL$(\preceq)$ formulas can be
eliminated from HLTL+LTL$(\preceq)$ formulas. This is done by recursively reducing the depth of nesting
of $X_T$ and $U_T$ operators, and finally eliminating propositions. 
We define the {\em rank} of an LTL$(\preceq)$ formula to be the maximum number of $X_T$ and $U_T$ operators 
along a path in its syntax tree. For a formula $\xi$ in HLTL+LTL$(\preceq)$, we define $r(\xi) = (n,m)$ 
where $n$ is the maximum rank of an LTL$(\preceq)$ formula occurring in $\xi$, and $m$ is the number of
such formulas with rank $n$. The pairs $(n,m)$ are ordered lexicographically.
 
Let $[\xi]_{T_1}$ be an HLTL+LTL$(\preceq)$ formula. We associate to $[\xi]_{T_1}$ the tree $\T(\xi)$ whose nodes
are all occurrences of subformulas of the form $[\psi]_T$, with an edge 
from $[\psi_i]_{T_i}$ to 
$[\psi_j]_{T_j}$ if the latter occurs in $\psi_i$ and $T_j$ is a child of $T_i$ in $\calh$.

Consider an  HLTL+LTL$(\preceq)$ formula $[\xi]_{T_1}$ such that $r(\xi) \geq (1,1)$. 
Suppose $\xi$ has a subformula ${\bf X}_T \varphi$ in LTL$(\preceq)$ of maximum rank.
Pick one such occurrence and let $\bar T$ be the minimum task (wrt $\calh$) such that ${\bf X}_T \varphi$ occurs in $[\psi]_{\bar T}$.
We construct an HLTL+LTL$(\preceq)$ formula $\bar \xi$ such that $r(\bar \xi) < r(\xi)$, essentially by eliminating
{\bf X}$_T$.  We consider 4 cases: $T = \bar T$, $T$ is a descendant or ancestor of $\bar T$, or neither.

Suppose first that $T = \bar T$. Consider an occurrence of ${\bf X}_T \varphi$. Intuitively, there are two cases: ${\bf X}_T \varphi$ is
evaluated inside the run of $T$ corresponding to  $[\psi]_T$, or at the last configuration. 
In the first case ($\neg \sigma^c_T$ holds),  
${\bf X}_T \varphi$ is equivalent to ${\bf X} \varphi$.
In the second case ($\sigma^c_T$ holds), ${\bf X}_T \varphi$ holds iff $\varphi$ holds at the next call to $T$. Thus, 
$\xi$ is equivalent to $\xi_1 \vee \xi_2$, where: 
\begin{enumerate}
\item $\xi_1$ says that $\varphi$ does not hold at the next call to $T$ (or no such call exists) and 
${\bf X}_T \varphi$ is replaced in $\psi$
by $\neg \sigma^c_T \wedge {\bf X} \varphi$ 
\item $\xi_2$ says that $\varphi$ holds at the next call to $T$ (which exists) and 
${\bf X}_T \varphi$ is replaced in $\psi$
by $\neg \sigma^c_T \rightarrow {\bf X} \varphi$.
\end{enumerate} 
We next describe how $\xi_1$ states that $\varphi$ does not hold at the next call to $T$
($\xi_2$ is similar). We need to state that either 
there is no future call to $T$, or such a call exists and $\neg \varphi$ holds at the first such call. 
Consider the path from $T_1$ to $T$ in $\calh$. Assume for simplicity that the path is
$T_1, T_2, \ldots, T_k$ where $T_k = T$. For each $i$, $1 \leq i < k$, we define inductively (from $k-1$ to $1$)
formulas $\alpha_i, \beta_i(\neg \varphi)$
such that $\alpha_i$ says that there is no call leading to $T$ in the remainder of the current subrun of $T_i$,
and $\beta_i(\neg \varphi)$ says that such a call exists and the first call leads to a subrun of $T$ satisfying $\neg \varphi$.  
First, $\alpha_{k-1} = {\bf G} (\neg \sigma^o_{T_k})$
and $\beta_{k-1}(\neg \varphi) = \neg \sigma^o_{T_k} ~{\bf U} ~[\neg \varphi]_{T_k}$.
For $1 \leq i < k-1$, $\alpha_i = {\bf G} (\sigma^o_{T_{i+1}} \rightarrow [\alpha_{i+1}]_{T_{i+1}})$
and $\beta_i(\neg \varphi) = (\sigma^0_{T_{i+1}} \rightarrow [\alpha_{i+1}]_{T_{i+1}})~{\bf U}~[\beta_{i+1}(\neg \varphi)]_{T_{i+1}}$.
Now $\xi_1 = \xi_1^0 \vee \bigvee_{1 \leq j < k} \xi_1^j$ where $\xi_1^0$ states that there is no next call to $T$
and $\xi_1^j$ states that $T_j$ is the minimum task such that the next call to $T$ occurs during the {\em same} run of $T_j$
(and satisfies $\neg \varphi$).
More precisely, let $[\psi_1]_{T_1}, [\psi_2]_{T_2}, \ldots [\psi_k]_{T_k}$ be the path leading from $[\xi]_{T_1}$
to $[\psi]_{T}$ in $\T(\xi)$ (so $\psi_1 = \xi$ and $\psi_k = \psi$). Then $\xi_1^0$ is obtained by replacing each
$\psi_i$ by $\bar \psi_i$, $1 \leq i < k$, defined inductively as follows. First, $\bar \psi_{k-1}$ 
is obtained from $\psi_{k-1}$ by replacing 
$[\psi_k]_{T_k}$ with $[\psi_k]_{T_k} \wedge \alpha_{k-1}$. For $1 \leq i < k-1$,
$\bar \psi_i$ is obtained from $\psi_i$ by replacing $[\psi_{i+1}]_{T_{i+1}}$ with
$[\bar \psi_{i+1}]_{T_{i+1}} \wedge \alpha_i$. 
For $1 \leq j < k$, $\xi_1^j$ is obtained by replacing in $\psi_j$, $[\psi_{j+1}]_{T_{j+1}}$ with 
$[\bar \psi_{j+1}]_{T_{j+1}} \wedge \beta_j(\neg \varphi)$. 
It is clear that $\xi_1$ states the desired property. The formula $\xi_2$ is constructed similarly.
Note that $r(\xi_1 \vee \xi_2) < r(\xi)$.

Now suppose $T$ is an ancestor of $\bar T$.  We reduce this case to the previous ($T = \bar T$).
Let $T'$ be the child of $T$.  Suppose $[\psi_T]_T$ is the ancestor of 
$[\psi]_{\bar T}$ in $\T(\xi)$.
Then $\xi$ is equivalent to $\bar \xi = \xi_1 \vee \xi_2$ where:
\begin{enumerate}
\item $\xi_1$ says that $\varphi$ does not hold at the next action of $T$ wrt $\preceq$ (or no such next action exists) and
$\psi$ is replaced by $\psi({\bf X}_T \varphi \leftarrow \emph{false})$ ($\leftarrow$ denotes substitution)
\item $\xi_2$ says that $\varphi$ holds at the next action of $T$ wrt $\preceq$ and
$\psi$ is replaced by $\psi({\bf X}_T \varphi \leftarrow \emph{true})$
\end{enumerate}
To state that $\varphi$ does not hold at the next call to $T$ (or no such call exists) 
$\xi_1$ is further modified by replacing in $\psi_T$, $[\psi_{T'}]_{T'}$ with
$[\psi_{T'}]_{T'} \wedge ({\bf G}(\neg \sigma^c_{T'}) \vee (\neg \sigma^c_{T'} ~{\bf U}~ (\sigma^c_{T'} \wedge \neg {\bf X}_T \varphi))$. 
Smilarly, $\xi_2$ is further modified by replacing in $\psi_T$, $[\psi_{T'}]_{T'}$ with
$[\psi_{T'}]_{T'} \wedge (\neg \sigma^c_{T'} ~{\bf U}~ (\sigma^c_{T'} \wedge {\bf X}_T \varphi))$.
Note that there are now two occurrences of ${\bf X}_T \varphi$ in the modified $\psi_T$'s.
By applying twice the construction for the case $\bar T = T$ we obtain an equivalent $\bar \xi$ such that
$r(\bar \xi) < r(\xi)$. 

Next consider the case when $\bar T$ is an ancestor of $T$. Suppose the path from $T_1$ to $T$ in $\calh$
is $T_1, \ldots, T_i, \ldots T_k$ where $T_i = \bar T$ and $T_k = T$.
Consider the value of ${\bf X}_T \varphi$ in the run $\rho_\psi$ of $\bar T$ on which $\psi$ is evaluated.
Similarly to the case $T = \bar T$, there are two cases: 
$\varphi$ holds at the next invocation of $T$ following $\rho_\psi$, or it does not.
Thus, $\xi$ is equivalent to $\xi_1 \vee \xi_2$, where:
\begin{enumerate}
\item $\xi_1$ says that $\varphi$ does not hold at the next call to $T$ (or no such call exists) and
${\bf X}_T \varphi$ is replaced in $\psi$
by $\beta_i(\varphi)$, where $\beta_i(\varphi)$ says that there exists a future call leading to $T$ in the current run of $\bar T$, and 
the first such run of $T$ satisfies $\varphi$; $\beta_i(\varphi)$ is constructed as in the case $T = \bar T$.
\item $\xi_2$ says that $\varphi$ holds at the next call to $T$ following the current run of $\bar T$ and
${\bf X}_T \varphi$ is replaced in $\psi$
by $\alpha_i \vee \beta_i(\varphi)$ where $\alpha_i$, constructed as for the case $T = \bar T$, 
says that there is no future call leading to $T$ in the current run of $\bar T$.
\end{enumerate}
To say that $\varphi$ does not hold at the next call to $T$ following $\rho_\psi$ (or no such call exists), $\xi_1$ is modified analogously to 
the case $\bar T = T$, and similarly for $\xi_2$.

Finally suppose the least common ancestor of $\bar T$ and $T$ is $\hat T$ distinct from both.
Let $[\psi_{\hat T}]_{\hat T}$ be the ancestor of $[\psi]_{\bar T}$ in $\T(\xi)$.
Consider the value of ${\bf X}_T \varphi$ in the run of $\bar T$ on which $\psi$ is evaluated.
There are two cases: $\varphi$ holds at the next invocation of $T$ following the run of $\bar T$, or it does not.
Thus, $\xi$ is equivalent to $\xi_1 \vee \xi_2$, where:
\begin{enumerate}
\item $\xi_1$ says that $\varphi$ does not hold at the next call to $T$ (or no such call exists) and
$\psi$ is replaced by $\psi({\bf X}_T \varphi \leftarrow \emph{false})$
\item $\xi_2$ says that $\varphi$ holds at the next call to $T$ and
$\psi$ is replaced by $\psi({\bf X}_T \varphi \leftarrow \emph{true})$
\end{enumerate}
To say that $\varphi$ does not hold at the next call to $T$ (or no such call exists), $\xi_1$ is modified analogously to
the case $\bar T = T$, and similarly for $\xi_2$, taking into account the fact that the next call to $T$, if it exists,
must take place in the current run of $\hat T$ or of one of its ancestors.
This completes the simulation of {\bf X}$_T \varphi$.

Now  suppose $\xi$ has a subformula $(\varphi_1 ~{\bf U}_T~ \varphi_2)$ of maximum rank.
Pick one such occurrence and let $\bar T$ be the minimum task (wrt $\calh$) such that 
$(\varphi_1 ~{\bf U}_T~ \varphi_2)$  occurs in $[\psi]_{\bar T}$. There are several cases: $\bar T = T$, 
$\bar T$ is an ancestor or descendant of $T$, or neither. The simulation technique is similar to the above.
We outline the construction for the most interesting case when $\bar T = T$.

Consider the run of $T$ on which $[\psi]_T$ is evaluated. There are two cases: 
$(\dag)$ $(\varphi_1 ~{\bf U}_T~ \varphi_2)$ holds on the concatenation of the future runs of $T$, or $(\dag)$ does not hold.
Thus, $\xi$ is equivalent to $\xi_1 \vee \xi_2$ where:

\begin{enumerate}
\item $\xi_1$ says that $(\dag)$ holds and $\psi$ is modified by replacing the occurrence of  $(\varphi_1 ~{\bf U}_T~ \varphi_2)$
with ${\bf G} \varphi_1  \vee ~(\varphi_1 ~{\bf U}~ \varphi_2)$, and

\item $\xi_2$ says that $(\dag)$ does not hold and $\psi$ is modified by replacing the occurrence of  $(\varphi_1 ~{\bf U}_T~ \varphi_2)$
with $(\varphi_1 ~{\bf U}~ \varphi_2)$.
\end{enumerate}

We show how $\xi_1$ ensures $(\dag)$. Let $T_1, \ldots, T_k$ be the path from root to $T$ in $\calh$. 
For each $i$, $1 \leq i < k$, we define inductively (from $k-1$ to $1$)
formulas $\alpha_i, \beta_i$ as follows. Intuitively, $\alpha_i$ says that all future calls leading to $T$ 
from the current run of $T_i$ must result in runs satisfying {\bf G} $\varphi_1$:
\begin{itemize}
\item $\alpha_{k-1} = {\bf G} (\sigma^o_{T_k} \rightarrow [{\bf G}~\varphi_1]_{T_k})$,
\item for $1 \leq i < k-1$, $\alpha_i = {\bf G} (\sigma^o_{T_{i+1}} \rightarrow [\alpha_{i+1}]_{T_{i+1}})$
\end{itemize}
The formula $\beta_i$ says that there must be a future call to $T$ in the current run of $T_i$ satisfying $\varphi_1 {\bf U} \varphi_2$
and all prior calls result in runs satisfying ${\bf G} \varphi_1$:
\begin{itemize}
\item $\beta_{k-1} = (\sigma^o_{T_k} \rightarrow [{\bf G} \varphi_1]_{T_k}) ~{\bf U}~ [\varphi_1 {\bf U} \varphi_2]_{T_k}$,
\item for $1 \leq i < k-1$, $\beta_i = (\sigma^o_{T_{i+1}} \rightarrow [\alpha_{i+1}]_{T_{i+1}}) ~{\bf U}~ [\beta_{i+1}]_{T_{i+1}}$.
\end{itemize}
  
Now $\xi_1$ is $\bigvee_{1 \leq j < k} \xi_j$ where $\xi_j$ 
states that the concatenation of runs resulting from calls to $T$ within the run of $T_j$ 
on which $[\psi_j]_{T_j}$ is evaluated, satisfies $(\varphi_1 ~{\bf U}~ \varphi_2)$.   
More precisely, let $[\psi_1]_{T_1}, \ldots, [\psi_k]_{T_k}$ be the path from $[\xi]_{T_1}$ to $[\psi]_T$ in $\T(\xi)$
(so $\psi_1 = \xi$ and $\psi_k = \psi$). For each $j$ we define $\psi_i^j$, $1 \leq i < k$ as follows:
\begin{itemize}
\item if $j < k-1$, $\psi^j_{k-1}$ is obtained from $\psi_{k-1}$ by replacing $[\psi_k]_{T_k}$ with
$[\psi_k]_{T_k} \wedge \alpha_{k-1}$
\item if $j = k-1$, $\psi^j_{k-1}$ is obtained from $\psi_{k-1}$ by replacing $[\psi_k]_{T_k}$ with
$[\psi_k]_{T_k} \wedge \beta_{k-1}$
\item for $j < i < k-1$, $\psi^j_i$ is obtained from $\psi_i$ by replacing $[\psi^j_{i+1}]_{T_{i+1}}$
with  $[\psi^j_{i+1}]_{T_{i+1}} \wedge \alpha_i$
\item $\psi^j_j$ is obtained from $\psi_j$ by replacing  $[\psi^j_{j+1}]_{T_{j+1}}$ with \\
$[\psi^j_{j+1}]_{T_{j+1}} \wedge \beta_j$
\item for $1 \leq i <j$, $\psi^j_i$ is obtained from $\psi_i$ by replacing $[\psi_{i+1}]_{T_{i+1}}$
with $[\psi^j_{i+1}]_{T_{i+1}}$.
\end{itemize}
Finally, $\xi_j = [\psi^j_1]_{T_1}$. The formula $\xi_2$ is constructed along similar lines.
This completes the case $(\varphi_1~{\bf U}_T~\varphi_2)$.

Consider now the case when the formula of maximum rank is a proposition $(p,\sigma) \in $ P$\Sigma$,
where $p \in P_T$ and $\sigma \in \Sigma_T^{obs}$.
%Recall that $p \in P_T$ if $\sigma \in \Sigma_T \cup \{\sigma^o_{T'} \mid T' \in \emph{child}(T)\}$ and
%$p \in P_T \cup P_{T'}$ if $\sigma = \sigma^c_{T'}$ where $T' \in \emph{child}(T)$.
%Suppose that $p \in P_T$ and $\sigma \in  \Sigma_T \cup \{\sigma^o_{T'} \mid T' \in \emph{child}(T)\}$. 
There are several cases:
\begin{itemize}
\item $(p,\sigma)$ occurs in $[\psi]_{T}$. Then $(p,\sigma)$ is replaced with $p \wedge \sigma$.
\item $(p,\sigma)$ occurs in $[\psi]_{\bar T}$ where ${\bar T} \neq T$ and $\bar T$ is not a child or parent of $T$. 
Then $(p,\sigma)$ is replaced with {\em false}.
\item $(p, \sigma)$ occurs in $[\psi_{T'}]_{T'}$ for some parent $T'$ of $T$.
If $\sigma \in \Sigma_T$ then $(p,\sigma)$ is replaced with {\em false} in $\psi_{T'}$.
If $\sigma = \sigma^o_{T}$ then $(p,\sigma)$ is replaced by $[p]_T$.
If $\sigma = \sigma^c_{T}$, we use the past temporal operator {\bf S} whose semantics is symmetric to {\bf U}. 
This can be simulated in LTL, again as a consequence of Kamp's Theorem \cite{kamp}.
The proposition $(p,\sigma)$ is replaced in $\psi_{T'}$ by 
$\sigma_{T}^c \wedge ((\neg \sigma^o_{T})~{\bf S}~ [{\bf F}(\sigma^c_{T} \wedge p)]_{T})$
\item $(p, \sigma)$ occurs in $[\psi_{T'}]_{T'}$ for some child $T'$ of $T$. Let $[\psi_T]_T$ be the parent of $[\psi_{T'}]_{T'}$
in $\emph{Tree}(\xi)$.
As above, if $\sigma \in \Sigma_T$ then $(p,\sigma)$ is replaced with {\em false} in $\psi_{T'}$.
If $\sigma = \sigma^o_{T'}$, there are two cases: $(1)$ $p$ holds in $T$ when the call to $T'$ generating the run on which $\psi_{T'}$ is evaluated is made, and $(2)$ the above is false. Thus, $\psi_T$ is replaced by $\psi^1_T \vee \psi^2_T$ where $\psi^1_T$ corresponds to $(1)$ and  
$\psi^2_T$ to $(2)$.  Specifically:
\begin{itemize}
\item $\psi^1_T$ is obtained from $\psi_T$ by replacing $[\psi_{T'}]_{T'}$ with $p \wedge [\psi^1_{T'}]_{T'}$,
where $\psi^1_{T'}$ is obtained from $\psi_{T'}$ by replacing $(p, \sigma^o_{T'})$ with $\sigma^o_{T'}$
\item $\psi^2_T$ is obtained from $\psi_T$ by replacing $[\psi_{T'}]_{T'}$ with $\neg p \wedge [\psi^2_{T'}]_{T'}$ 
where $\psi^2_{T'}$ is obtained from $\psi_{T'}$ by replacing $(p, \sigma^o_{T'})$ with \emph{false}.
\end{itemize}
Now suppose $\sigma = \sigma^c_{T'}$. Again, there are two cases:
$(1)$ if $T'$ returns then $p$ holds in the run of $T$ when $T'$ returns, and $(2)$ this is false.
The two cases are treated similarly to the above.
\end{itemize}
This concludes the proof of the lemma.
\end{proof}

Theorem \ref{thm:complete} now follows. Let $\varphi_f$ be an interleaving-invariant $\ltlfo$ formula over $\Gamma$. 
By Lemma \ref{lem:local-ltl}, we can assume that $\varphi$ is in local-LTL
and in particular uses the set of propositions P$\Sigma$. 
By Lemma \ref{lem:hltl}, there exists an HLTL formula $[\xi]_{T_1}$ equivalent to $\varphi$ over $\omega$-words in 
$\calh(\Gamma)$, using propositions
in $P \cup \Sigma$. 
Moreover, by construction, each sub-formula $[\psi]_T$ of $[\xi]_{T_1}$ uses only propositions
in $P_T \cup \Sigma^{obs}_T$. 
It is easily seen that formula obtained by replacing each $p$ with $f(p)$ is a well-formed $\hltlfo$ formula
equivalent to $\varphi_f$ on all runs of $\Gamma$.

\subsection{Simplifications} \label{app:simplification}

We first show that the global variables, as well as set atoms,
can be eliminated from $\hltlfo$ formulas.

\begin{lemma} \label{lem:noglobal} Let $\Gamma$ be a HAS and $\forall \bar y [\varphi_f(\bar y)]_{T_1}$ an $\hltlfo$
formula over $\Gamma$. One can construct in linear time a HAS $\bar \Gamma$ and an $\hltlfo$
formula $[\bar \varphi_f]_{\bar T_1}$, where $\bar \varphi_f$ contains no atoms $S^T(\bar z)$,
such that $\Gamma \models \forall \bar y [\varphi_f(\bar y)]_{T_1}$ iff
$\bar \Gamma \models [\bar \varphi_f]_{\bar T_1}$.
\end{lemma}

\begin{proof}
Consider first the elimination of global variables. Suppose $\Gamma$ has tasks $T_1, \ldots, T_k$.
The Hierarchical artifact system $\bar \Gamma$ is constructed from $\Gamma$
by adding $\bar y$ to the variables of $T_1$ and augmenting the
input variables of all other tasks with $\bar y$ (appropriately renamed).
Note that $\bar y$ is unconstrained,
so it can be initialized to an arbitrary valuation and then passed as input to all other tasks.
Let $\Gamma$ consist of the resulting tasks, $\bar T_1, \ldots, \bar T_k$.
It is clear that  $\Gamma \models \forall \bar y [\varphi_f(\bar y)]_{T_1}$ iff
$\bar \Gamma \models [\bar \varphi_f]_{\bar T_1}$.

Consider now how to eliminate atoms of the form $S^T(\bar z)$ from $\bar \varphi_f$.
Recall that for all such atoms, $\bar z \subseteq \bar y$,
so $\bar z$ is fixed throughout each run. The idea is keep track of the membership of $\bar z$ in $S^T$
using two additional numeric artifact variables $x_{\bar z}$ and $y_{\bar z}$, such that $x_{\bar z} = y_{\bar z}$
indicates that $S^T(\bar z)$ holds\footnote{This is done to avoid introducing constants, that could also be used as flags.}.
Specifically,  a pre-condition ensures that $x_{\bar z} \neq y_{\bar z}$ initially holds, then $x_{\bar z} \neq y_{\bar z}$
is enforced as soon as there is an insertion $+S^T(\bar s^T)$ for which $\bar s^T = \bar z$, and
$x_{\bar z} \neq y_{\bar z}$ is enforced again whenever there is a retrieval of a tuple equal to $\bar z$.
This can be achieved using pre-and-post conditions of services carrying out the insertion or retrieval.
Then the atom $S^T(\bar z)$ can be replaced in
$\bar \varphi_f$ with $(x_{\bar z} =  y_{\bar z})$.
\end{proof}

We next consider two simplifications of artifact systems regarding the interaction of tasks with their subtasks.

\begin{lemma}  \label{lem:simplification}
%\label{lem:nonoverlapping}
Let $\Gamma$ be a HAS and $\varphi$ an $\hltlfo$ property over $\Gamma$. One can construct a HAS
$\tilde{\Gamma}$ and an $\hltlfo$ formula $\tilde{\varphi}$ such that
$\Gamma \models \varphi$ iff $\tilde{\Gamma} \models \tilde{\varphi}$ and:
$(i)$ $\bigcup_{T_c \in child(T)} \xttcup$ and
$\bigcup_{T_c \in child(T)} \xttcdown$ are disjoint for each task $T$ in $\tilde{\Gamma}$,
$(ii)$ for each child task $T_c \in child(T)$, $\xttcup \cap \varnum = \emptyset$.
\end{lemma}

\begin{proof}
Consider (i). \newcommand{\yttcdown}{\hat{x}^T_{T_c^\downarrow}}
\newcommand{\xinit}{x_{init}}
\newcommand{\sigmainit}{\sigma_T^{init}}
We describe here informally the construction of $\tilde{\Gamma}$ that eliminates overlapping between \\
$\bigcup_{T_c \in child(T)} \xttcup$ and $\bigcup_{T_c \in child(T)} \xttcdown$.
For each task $T$ and for each subtask $T_c$ of $T$, for each variable $x \in \xttcdown$,
we introduce to $T$ a new variable $\hat{x}$ whose type is the same as the type (id or numeric) of $x$.
We denote by $\yttcdown$ the set of variables added to $T$ for subtask $T_c$.
Then instead of passing $\xttcdown$ to $T_c$, $T$ passes $\yttcdown$ to $T_c$ when $T_c$ opens.
And for the opening service $\sigma_{T_c}^o$ with opening condition $\pi$, we check $\pi$ in conjunction with $\bigwedge_{x \in \xttcdown } (x = \hat{x})$.
Note that $\bigcup_{T_c \in child(T)} \yttcdown$ and $\bigcup_{T_c \in child(T)} \xttcup$ are disjoint.
By this construction, in each run of $\tilde{\Gamma}$, after each application of an internal service $\sigma$ of task $T$,
the variables in $\yttcdown$ for each subtask $T_c$ receives a set of non-deterministically chosen values. Then each subtask $T_c$ can be opened
only when $\yttcdown$ and $\xttcdown$ have the same values. So passing $\yttcdown$ to $T_c$ is equivalent to passing $\xttcdown$ to $T_c$.

To guarantee that there is a bijection from the runs of $\Gamma$ to the runs of $\tilde{\Gamma}$,
we also need to make sure that the values of $\yttcdown$ are non-deterministically chosen before the first application of internal service.
(Recall that they either contain 0 or \anull at the point when $T$ is opened.)
So we argument $T$ with an extra binary variable $\xinit$ and an extra internal service $\sigmainit$.
Variable $\xinit$ indicates whether task $T$ has been ``initialized''.
The service $\sigmainit$ has precondition that checks whether $\xinit = 0$ and post-condition sets $\xinit = 1$.
It sets all id variables to $\anull$ and numeric variables 0 except for variables in $\yttcdown$ for any $T_c$.
So application of $\sigmainit$ assigns values to $\yttcdown$ for every subtask $T_c$ non-deterministically and all other variables
are initialized to the initial state when $T$ is opened.
All other services are modified such that they can be applied only when $\xinit = 1$.
So in a projected run $\rho_T$ of $\tilde{\Gamma}$, the suffix with $\xinit = 1$ corresponds to the original projected run of $\Gamma$.
Thus we only need to rewrite the $\hltlfo$ property $\varphi$ to $\tilde{\varphi}$ such that each formula in $\Phi_T$
only looks at the suffix of projected run $\rho_T$ after $\xinit$ is set to be 1.
(Namely, each $\psi \in \Phi_T$ is replaced with $\mathbf{F}( (\xinit=1) \land \psi)$.)
%ild}(T)

Now consider (ii).
We outline the construction of $\tilde{\Gamma}$ and $\tilde{\varphi}$ informally.
For each task $T$, we introduce a set of new numeric variables $\{x_{T_c} | T_c \in child(T), x \in \xttcup \cap \varnum \}$ to $\bar{x}^T$.
Intuitively, these variables contain non-deterministically guessed returning values from each child task $T_c$. These 
are passed to each child task $T_c$ as additional input variables. Before $T_c$ returns, these are compared to the values 
of the returning numeric variables of $T_c$, and $T_c$ returns only if they are identical.
More formally, for each child task $T_c$ of $T$, variables $\{x_{T_c} | x \in \xttcup \cap \varnum\}$ are passed from $T$ to $T_c$
as part of the input variables of $T_c$.
For each variable $x_{T_c}$ in $T$, we let $x_{T_c \rightarrow T} \in \bar{x}^{T_c}$ be the corresponding input variable of $x_{T_c}$.
And for each $x_{T_c}$, we denote by $x_{ret}$ the variable in $\bar{x}^{T_c}$ satisfying that
$f_{out}(x) = x_{ret}$ for $f_{out}$ in the original $\Gamma$.
Then at $T_c$, we remove all numeric variables from $\bar{x}^{T_c}_{ret}$
and add condition $\bigwedge_{x \in {\xttcup} \cap \varnum} x_{ret} = x_{T_c \rightarrow T} $ to the closing condition of $T_c$.
Note that we need to guarantee that the variables in $\{x_{T_c} | T_c \in child(T), x \in \xttcup \cap \varnum \}$ obtain
non-deterministically guessed values. This can be done as in the simulation for (i).

Conditions on $\bar x^T$ after a subset $T$'s children has returned are evaluated using the guessed values 
for the variables returned so far. Specifically, the correct value to be used is the latest returned by a child transaction, if any
(recall that children tasks can overwrite each other's numeric return variables in the parent). 
Keeping track of the sequence of returned transactions and evaluating conditions with the correct value
can be easily done directly in the verification algorithm, at negligible extra cost. 
This means that we can assume that tasks have the form in (ii) without the exponential blowup in the conditions, 
but with the quadratic blowup in the number of variables.
%\yuliang{I think the quadratic blow-up in the number of variables is not avoidable by adapting the verification algorithm. This needs to be taken into account in the complexity analysis.}
%\victor{agreed. I modified slightly the text.}

To achieve the simulation fully via the specification is costlier because some of the conditions needed have exponential size. 
We next show how this can be done.  Intuitively, we guess initially an order of the return of the children
transactions and enforce that it be respected. We also keep track of the children that have already returned. 
Let $\emph{child}(T) = \{T_1,\ldots,T_n\}$. To guess an order of return, we use new ID variables $\bar o = \{o_{ij} \mid 1 \leq i,j \leq n\}$.
Intuitively, $o_{ij} \neq \anull$ says that $T_i$ returns before $T_j$. We also use new ID variables $\{t_i \mid 1 \leq i \leq n\}$,
where $t_i \neq \anull$ means that $T_i$ has returned. The variables $\bar o$ are subject to a condition  
specifying the axioms for a total order:
$$\begin{array}{l}
\wedge_{1 \leq i , j \leq n} (o_{ij} \neq \anull \vee o_{ji} \neq \anull) \\
\wedge_{1 \leq i < j \leq n} \neg (o_{ij} \neq \anull \wedge o_{ji} \neq \anull) \\
\wedge_{1 \leq i,j,m \leq n} ((o_{ij} \neq \anull \wedge o_{jm} \neq \anull) \rightarrow o_{im} \neq \anull) \\
\end{array}
$$
These are enforced using pre-conditions of services as well as one additional initial internal service 
(which in turn requires a minor modification to $\varphi$, similarly to (i)).
When $T_i$ returns, $t_i$ is set to a non-null value, and the condition 
$$\bigwedge_{1 \leq i,j \leq n} (t_i \neq \anull \wedge t_j = \anull) \rightarrow o_{ij} \neq \anull$$ 
enforcing that transactions return in the order specified by $\bar o$ is maintained using pre-conditions.
Observe that, at any given time, the latest transaction that has returned is the $T_i$ such that 
$$t_i \neq \anull \wedge \bigwedge_{1 \leq j \leq n} ((o_{ij} \neq \anull) \rightarrow t_j = \anull)$$

For each formula $\pi$ over $\bar{x}^T$, we construct a formula $o(\pi)$ by replacing each variable $x \in \bar{x}^T_{\mathbb{R}}$
with $x_{T_c}$ for the latest $T_c$ where $x \in \xttcup$ if there is such $T_c$). 
The size of the resulting $o(\pi)$ is exponential in the maximum arity of database relations.
Finally we obtain $\tilde{\Gamma}$ and $\tilde{\varphi}$ by for every $T \in \calh$, replacing each condition $\pi$ over $\bar{x}^T$
with $o(\pi)$.  
One can easily verify that $\tilde{\Gamma} \models \tilde{\varphi}$ iff $\Gamma \models \varphi$ and
for every task $T$ of $\Gamma$, $\xttcup$ does not contain numeric variables.
This completes the proof of (ii). 
\end{proof}

The construction in (i) takes linear time in the original
specification and property. For (ii), the construction introduces a quadratic number of new variables and
the size of conditions becomes exponential in the maximum arity of data-base relations.
However, as discussed in Appendix \ref{app:hltl}, the verification algorithm can be slightly adapted to circumvent
the blowup in the specification without penalty to the complexity. Intuitively, this makes efficient use
of non-determinism, avoiding the explicit enumeration of choices required in the specification, which leads to the exponential blowup.

\section{Verification without Arithmetic} \label{sec:verification-app}

\subsection{Proof of Theorem \ref{thm:actual-symbolic}}

\subsubsection{Only-if: from actual runs to symbolic runs}  

Let $\Tree$ be a tree of local runs accepted by $\calb_\varphi$ (with database $D$).
The construction of $\Sym$ from $\Tree$ is simple. This can be done by replacing
each local run $\rho_T \in \Tree$ with a local symbolic run $\trt$.
More precisely, let 
$$\rho_T = (\nu_{in}, \nu_{out}, \{(J_i,\sigma_i)\}_{0 \leq i < \gamma})$$ be a local run in $\Tree$, where $J_i = (\nu_i, S_i)$,
We construct a corresponding local symbolic run 
$$\trt = (\tau_{in}, \tau_{out}, \{(I_i,\sigma_i)\}_{0 \leq i < \gamma})$$
For $0 \leq i < \gamma$, $I_i = (\tau_i, \bar{c}_i)$ is constructed from $(\nu_i, S_i)$ as follows.
The navigation set $\cale_T$ of $\tau_i$ contains every $x_R$ for every $x \in \bar{x}^T$ and
$R$ such that $\nu(x)$ is an ID of relation $R$ in $D$.
Then we define $\nu^*_i$ to be a mapping from
$\cale^+_T = \cale_T \cup \{0, \anull\} \cup \bar{x}^T$ to actual values, where:
\begin{itemize}\itemsep=0pt\parskip=0pt
\item $\nu^*_i(e) = e$ if $e \in \{0, \anull\}$,
\item $\nu^*_i(e) = \nu_i(x)$ for $e = x$ or $e = x_R$, and
\item $\nu^*_i(e.\xi) = t.\xi$ if $\nu^*_i(e)$ is an ID of a tuple $t \in D$.
\end{itemize}
We construct the equality type $\sim_{\tau_i}$ such that for every $e$ and $e'$ in $\cale^+$,
$e \sim_{\tau_i} e'$ iff $\nu^*_i(e) = \nu^*_i(e')$.
Also we let $\tauin = \tau_0 | \bar{x}^T_{in}$ and $\tauout = \tau_{\gamma-1} | \bar{x}^T_{in} \cup \bar{x}^T_{ret}$ if $\nuout \neq \bot$
and $\tau_{out} = \bot$ otherwise. Since $D$ satisfies the functional dependencies,
for every $\tau_i$ and expressions $e$ and $e'$, $e \sim_{\tau_i}, e'$
implies that $\nu^*_i(e) = \nu^*_i(e')$, so for every attribute $a$, if $e.a$ and $e'.a$ are in the navigation set of $\tau_i$,
then $e.a \sim_{\tau_i}, e'.a$ because $\nu^*_i(e.a) = \nu^*_i(e'.a)$.

We also note the following facts.
\begin{fact} \label{fact:cond-actual-symbolic}
For every condition $\psi$ over $\bar{x}^T$, $D \models \psi(\nu_i)$ iff $\tau_i \models \psi$.
\end{fact}
\begin{fact} \label{fact:consistent}
For all $i,i'$ and $\bar{x} \subseteq \bar{x}^T$,  if $\nu_i(\bar{x}) = \nu_{i'}(\bar{x})$
then $\tau_i | \bar{x} = \tau_{i'} | \bar{x}$.
\end{fact}
Given the sequence $\{(\tau_i, \sigma_i)\}_{0 \leq i < \gamma}$, the sequence of vectors of $TS$-isomorphism type counters
$\{\bar{c}_i\}_{0 \leq i < \gamma}$ is uniquely defined. Let $\trt = (\tau_{in}, \tau_{out}, \{(I_i, \sigma_i)\}_{0 \leq i < \gamma})$.
In view of Fact \ref{fact:cond-actual-symbolic}, it is easy to see that $\trt$ satisfies all items in the definition
of local symbolic run that do not involve the counters.
To show that $\trt$ is a local symbolic run, it remains to show that
$\bar{c}_i \geq \bar 0$ for $0 \leq i < \gamma$.
To see that this holds, we associate a sequence of counter vectors $\{\tilde{c}_i\}_{0 \leq i < \gamma}$ to the local run
$\rho_T$, where each $\tilde{c}_i$ provides, for each
$TS$-isomorphism type $\hat{\tau}$, the number of tuples in $S_i$ of $TS$-isomorphism type  $\hat{\tau}$
(the $TS$-isomorphism type of a tuple $t \in S_i$ is defined analogously to the $T$-isomorphism type for each local instance).
By definition, $\tilde{c}_i \geq \bar 0$ for each $i \geq 0$. Thus it is sufficient to show that $\tilde{c}_i \leq \bar{c}_i$ for each $i$.
We show this by induction. For $i = 0$,  $\tilde{c}_0 = \bar{c}_0 = 0$. Suppose $\tilde{c}_{i-1} \leq \bar{c}_{i-1}$ and consider the
transition under service $\sigma_i$ in $\rho_T$ and $\trt$. It is easily seen that $\tilde{c}_{i-1}$ and $\bar{c}_{i-1}$  are modified in
the same way {\em except} in the case when $+S^T(\bar{s}^T) \in \delta$, $\hat{\tau}_{i-1}$ is not input-bound,
and $\nu_{i-1}(\bar{s}^T) \in S_{i-1}$. In this case,
if $\hat{\tau}$ is the $TS$-isomorphism type of $\nu_{i-1}(\bar{s}^T)$,
$\tilde{c}_{i}(\hat{\tau}) = \tilde{c}_{i-1}(\hat{\tau})$ whereas $\bar{c}_{i}(\hat{\tau}) = \bar{c}_{i-1}(\hat{\tau}) +1$.
In all cases, $\tilde{c}_i \leq \bar{c}_i$. Thus, $\trt$ is a local symbolic run.
The fact that $\Sym$ is a tree of symbolic local runs follows from Fact \ref{fact:consistent}, which ensures the consistency of the isomorphism
types passed to and from subtasks. Finally, the fact that $\Sym$ is accepted by $\calb_\varphi$ follows from acceptance of $\Tree$ by $\calb_\varphi$
and Fact \ref{fact:cond-actual-symbolic}.

\subsubsection{If part: from symbolic runs to actual runs} 

We denote by \FD the set of key dependences in the database schema $\db$ and \IND the set of foreign key dependences.
We show the following.
\begin{lemma} \label{lem:sym-to-loc}
For every symbolic tree of runs $\Sym$ accepted by $\calb_\beta$, there exists a tree $\Tree$ of local runs 
accepted by $\calb_\beta$ with a finite database instance $D$ where $D \models \FD$.
\end{lemma}

Note that the above does not require that $D$ satisfy $\IND$. This is justified by the following.
\begin{lemma} \label{lem:foreignkey}
For every tree of local runs $\Tree$ with database $D \models \FD$ if $\Tree$ is accepted by $\calb_\beta$ then there exists a finite database 
$D' \models \FD \cup \IND$
such that $\Tree$ with database $D'$ is also a tree of local runs accepted by $\calb_\beta$
\end{lemma}

\begin{proof}
We can construct $D'$ by adding tuples to $D$ as follows. First, for each relation $R$ such that $R$ is empty in $D$, we add an arbitrary tuple 
$t$ to $R$. Next, for each foreign key dependency $R_i[F] \subseteq R_j[\ID]$,
for each tuple $t$ of $R_i$ such that there is no tuple in $R_j$ with id $t[F]$, we add to $R_j$ a tuple $t'$ where
\begin{itemize}\itemsep=0pt\parskip=0pt
\item $t'[\ID] = t[F]$, and
\item $t'[attr(R_j) - \{\ID\}] = t''[attr(R_j) - \{\ID\}]$ where $t''$ is an existing tuple in $R_j$.
\end{itemize}
$\Tree$ with database $D'$ is accepted by $\calb_\beta$ since $D'$ is an extension of $D$.   
Also $D'$ is finite since the number of added tuples is at most linear in
the sum of number of empty relations in $D$ and the number of tuples in $D$ that violate $\IND$.
\end{proof}

To show Lemma \ref{lem:sym-to-loc}, we begin with a construction of a local run $\rho_T$ on a finite database $D_T$ 
for each local symbolic run $\trt \in \Sym$. 
The local runs are constructed so that they can be merged consistently into a tree of local runs $\Tree$
with a single finite database $D$.
The major challenge in the construction of each $\rho_T$ and $D_T$ is that if 
$\trt$ is infinite, the size of $S^T$ can grow infinitely, and a naive construction of
$\rho_T$ would require infinitely many distinct values in $D_T$. Our construction needs to ensure that
$D_T$ is always finite. 
For ease of exposition, we first consider the case where $\trt$ is finite and then extend the result to infinite $\trt$.

\subsubsection*{Finite Local Symbolic Runs}

Recall from the previous section that $\nu^*(e)$ denotes the value of expression $e$ in database $D_T$ with valuation 
$\nu$ of $\bar x^T$. By abuse of notation, we extend $\nu^*(e)$ to 
$e \in \{x_R.w | x \in \bar{x}^T, R \in \db \} \cup \bar{x}^T \cup \{0, \anull\}$
where there is no restriction on the length of $w$. So for expression $e = x_R.w$, 
$\nu^*(e)$ is the value in $D_T$ obtained by foreign key navigation starting from the value $\nu^*(x)$ at relation $R$ 
and by the sequence of attributes $w$, if such a value exists. Note that $\nu^*$ may be only partially defined since 
$D_T$ may not satisfy all foreign key constraints.
Analogously, we define $\nu^*_{in}(e)$ to be the value of $e$ in $D_T$ at valuation $\nu_{in}$
and $\nu^*_{out}(e)$ to be the value of $e$ in $D_T$ at valuation $\nu_{out}$. 

We prove the following, showing the existence of an actual local run corresponding to a finite local symbolic run.
The lemma provides some additional information used when merging local runs into a final tree of runs.  

\begin{lemma} \label{lem:sym-to-loc-finite}
For every finite local symbolic run 
 $\trt = \\ (\tauin, \tauout, \II)$ ($\gamma \neq \omega$), there exists a local run
$\rho_T = (\nu_{in}, \nu_{out}, \RS)$ on finite database $D_T \models \FD$ such that for every $0 \leq i < \gamma$,
\begin{itemize}\itemsep=0pt\parskip=0pt
\item [(i)] for every expression $e = x_R.w$ where $\nu_i^*(e)$ is defined,
there exists expression $e' = x_R.w'$ where $|w'| \leq h(T)$ such that $\nu_i^*(e) = \nu_i^*(e')$,
\item [(ii)] for all expressions $e, e' \in \cale^+_T$ of $\tau_i$,
if $\nu_i^*(e)$ and $\nu_i^*(e')$ are defined, then $e \sim_{\tau_i} e'$ iff $\nu_i^*(e) = \nu_i^*(e')$, and
\item [(iii)] for $\delta = h(T_c)$ if $\sigma_i \in \{ \sigma_{T_c}^o, \sigma_{T_c}^c\}$ for some $T_c \in child(T)$ and $\delta = 1$ otherwise,
for every expression $e \in \cale^-_T = \cale^+_T - \{x_R.w | x \in \bar{x}^T, |w| > \delta\}$, $\nu_i^*(e)$ is defined.
\end{itemize}
\end{lemma} 
Part (i), needed for technical reasons, says that for all values $v$ in $D_T$, if $v$ is the value of expression $x_R.w$,
then $v$ is also the value of an expression $x_R.w'$ where the length of $w'$ is within $h(T)$.
%In addition, consistency of equality of values reachable by navigation from variables in $D_T$ 
%and the equality type of $\tau_i$ is required.   
Part (ii) says, intuitively, that the equality types in the symbolic local run and the constructed local run are the same. 
Part (iii) states that for every $0 \leq i < \gamma$, at valuation $\nu_i$, 
every expression $e$ within $\delta$ steps of foreign key navigation from any variable $x$ is defined in $D_T$. 
Since $\delta \geq 1$, this together with (ii) implies that for every condition $\pi$, $\tau_i \models \pi$ iff $D_T \models \pi(\nu_i)$.
So if $\trt$ is accepted by some computation of a \buchi \ automaton $B(T, \eta)$ 
then $\rho_T$ is also accepted by the same computation of $B(T, \eta)$.

We provide the proof of Lemma \ref{lem:sym-to-loc-finite} in the remainder of the section. 
We first show that from each finite local symbolic run $\trt$, we can construct a
\emph{global isomorphism type} of $\trt$, which is essentially an equality type over the entire set of expressions
in the symbolic instances of $\trt$.
Then we show that the local run $\rho_T$ and database $D_T$ whose domain values are the equivalence classes of
the global isomorphism type, satisfy the properties in Lemma \ref{lem:sym-to-loc-finite}.

\subsubsection*{Global Isomorphism Types}
We prove Lemma \ref{lem:sym-to-loc-finite} by constructing $\rho_T$ and $D_T$ from $\trt = (\tauin, \tauout, \II)$ $(\gamma \neq \omega)$.
We first introduce some additional notation.

Let $\cali^+$ be the set of symbolic instances $I_i$ of $\trt$ ($i < \gamma - 1$) such that $+S^T(\bar{s}^T) \in \delta_{i+1}$ and $\hat{\tau_i}$ is not input-bound.
Similarly let $\cali^-$ be the set of symbolic instances $I_i$ ($i < \gamma$) such that $-S^T(\bar{s}^T) \in \delta_i$ 
and $\hat{\tau_i}$ is not input-bound.
We define a one-to-one function $\Ret$ from $\cali^-$ to $\cali^+$ such that for every $I_i = \Ret(I_j)$,
$i < j$ and $\hat{\tau}_i = \hat{\tau}_j$. We say that $I_j$ retrieves from $I_i$.
As $\bar{c}_i \geq 0$ for every $i$, at least one mapping $\Ret$ always exists. 
Intuitively, $\Ret$ connects symbolic instance $I_j$ to $I_i$ such that $I_j$ retrieves a tuple from $S^T$ 
which has the same isomorphism type as a tuple inserted at $I_i$. 
For each $I_i = \Ret(I_j)$, in the local run $\rho_T$ we construct, 
valuations $\nu_i$ and $\nu_j$ have same values on variables $\bar{s}^T$.
Here we ignore input-bound isomorphism types
since these can be seen as part of the input isomorphism type: in $\rho_T$, instances
having the same input-bound $TS$-isomorphism type have the same values on $\bar{s}^T$.

Recall that a \emph{segment} $S = \{(I_i, \sigma_i)\}_{a \leq i \leq b}$ is a maximum consecutive subsequence of $\II$ such that
$\sigma_a$ is an internal service and for $a < i \leq b$, $\sigma_i$ is opening service or closing service of child tasks of $T$.
For our choice of the $\Ret$ relation, we define a \emph{life cycle} $L = \{(I_i, \sigma_i)\}_{i \in J}$
as a maximum subsequence of $\II$ for $J \subseteq [0, \gamma)$
where for each pair of consecutive $(I_a, \sigma_a)$ and $(I_b, \sigma_b)$ in $L$ where $a < b$,
$(I_a, \sigma_a)$ and $(I_b, \sigma_b)$ are either in the same segment or $I_a = \Ret(I_b)$.
Note that a life cycle $L$ is also a sequence of segments.
From the definition of local symbolic runs, we can show the following properties for segments and life cycles:
\begin{lemma} \label{lem:segment-lifecycle}
$(i)$ For every segment $S = \{(I_i, \sigma_i)\}_{a \leq i \leq b}$,
for every $i, j \in [a, b]$ where $i < j$, for $\bar{x} = \{x | x \in \bar{x}^T, x \not\sim_{\tau_i} \anull\}$,
$\tau_i | \bar{x} = \tau_j | \bar{x}$.
$(ii)$ For every life cycle $L = \{(I_i, \sigma_i)\}_{i \in J}$,
for every $i, j \in J$ where $i < j$, for $\bar{x} = \{x | x \in \bar{x}^T_{in} \cup \bar{s}^T, x \not\sim_{\tau_i} \anull\}$,
$\tau_i | \bar{x} = \tau_j | \bar{x}$.
\end{lemma}

Next, for each symbolic instance $I_i$, we define the {\em pruned} isomorphism type 
\reviewer{pruned isomorphism type is defined implicitly.} \yuliang{changed this paragraph a little bit to make it an explicit definition.}
$\lambda_i = (\cale_i, \sim_i)$ of $I_i$ as follows.
Intuitively, $\lambda_i$ is obtained from $\tau_i$ by removing expressions with ``long'' navigation from variables. 
Formally, let $\cale^+_T$ be the extended navigation set of $\tau_i$ and $\cale^-_T = \cale^+_T - \{x_R.w | x \in \bar{x}^T, |w| > \delta\}$, where 
$\delta = 1$ if $T$ is a leaf task, otherwise $\delta = \\ \max_{T_c \in child(T)} h(T_c)$. 
A \emph{local expression} of $I_i$ is a pair $(i, e)$ where $e \in \cale^-_T$, and we define that
$\cale_i = \{ (i, e)| e \in \cale^-_T\}$ is the \emph{local navigation set} of $\lambda_i$. 
We also define the \emph{local equality type} $\sim_i$ of $\lambda_i$ to be an equality type over $\cale_i$ where
$(i, e) \sim_i (i, e')$ iff $e \sim_{\tau_i} e'$, for every $e, e' \in \cale^-_T$.

Then we define the global isomorphism type as follows.
A global isomorphism type is a pair $\Lambda = (\cale, \sim)$, where 
$\cale = \bigcup_{0 \leq i < \gamma} \cale_i$ is called the \emph{global navigation set} and 
$\sim$ is an equality type over $\cale$ called \emph{global equality type}.
For each expression $e \in \cale$, let $[e]$ denote its equivalence class with respect to $\sim$.
The global equality type $\sim$ is constructed as follows: 
\begin{enumerate}
\item \textbf{Initialization:} $\sim \ \leftarrow \bigcup_{0 \leq i < \gamma} \sim_i$
\item \textbf{Chase:} 
Until convergence, merge two equivalence classes $E$ and $E'$ of $\sim$ if 
$E$ and $E'$ satisfy one of the following conditions: 
\begin{itemize}\itemsep=0pt\parskip=0pt
\item \textbf{Segment-Condition:} For some segment 
 $S = \\ \{(I_i, \sigma_i)\}_{a \leq i \leq b}$,
variable $x \in \bar{x}^T$ and $i, i' \in [a, b]$
where $x \not\sim_{\tau_i} \anull$ and $x \not\sim_{\tau_{i'}} \anull$,
$E = [(i, x)]$ and $E' = [(i', x)]$.
\item \textbf{Life-Cycle-Condition:} For some life cycle
$L = \\ \{(I_i, \sigma_i)\}_{i \in J}$,
variable $x \in \bar{x}^T_{in} \cup \bar{s}^T$ and $i, i' \in J$
where $x \not\sim_{\tau_i} \anull$ and $x \not\sim_{\tau_{i'}} \anull$,
$E = [(i, x)]$ and $E' = [(i', x)]$.
\item \textbf{Input-Condition:} For some variable $x \in \bar{x}^T_{in}$ and $i, i' \in [0, \gamma)$, $E = [(i, x)]$ and $E' = [(i', x)]$.
\item \textbf{FD-Condition:} For some local expressions $(i, e), (i', e')$ and attribute $a$
where $(i, e) \sim (i', e')$, $E = [(i, e.a)]$ and $E' = [(i', e'.a)]$.
\end{itemize}
\end{enumerate}

From the global isomorphism type $\Lambda$ defined above, we construct $\rho_T$ and $D_T$ as follows.
The domain of $D_T$ is the set of equivalence classes of $\sim$.  
Each relation $R(id, a_1, \dots, a_k)$ in $D_T$ consists of all tuples $([(i,e)], [(i,e.a_1)], \ldots [(i,e.a_k)])$ for which
$(i,e), (i,e.a_1), \ldots, (i,e.a_k) \in \cale$. 
Note that the chase step guarantees that for all local expressions $(i, e), (i', e')$,
if $(i, e.a), (i', e'.a) \in \cale$ and $(i, e) \sim (i', e')$, then $(i, e.a) \sim (i', e'.a)$.
It follows that $D_T \models \FD$.
We next define $\rho_T = (\nu_{in}, \nu_{out}, \{(\rho_i, \sigma_i)\}_{0 \leq i < \gamma})$, where $\rho_i = (\nu_i, S_i)$.
First, let $\nu_i(x) =  [(i,x)]$ for $0 \leq i < \gamma$, $\nu_{in} = \nu_0|{\bar x^T_{in}}$, and $\nu_{out} = \bot$ if $\tau_{out} = \bot$
and $\nu_{out} = \nu_{\gamma-1}|\bar x^T_{ret}$ otherwise. 
Suppose that, as will be shown below, properties (i)-(iii) of Lemma \ref{lem:sym-to-loc-finite} hold for $D_T$ 
and the sequence $\{\nu_i\}_{0 \leq i < \gamma}$ so defined. Note that (ii) and (iii) imply that the pre-and-post conditions of all services
$\sigma_i$ hold. Also, by construction, for every variable $x \in \bar{x}^T$ where $\nu_{i-1}(x) = \nu_{i}(x)$ is required by the transition under $\sigma_i$ 
we always have $(i, x) \sim (i+1, x)$. Consider the sets $\{S_i\}_{0 \leq i < \gamma}$. 
Recall the constraints imposed on sets by the definition of local run: 
$S_0 = \emptyset$, and for $0 < i < \gamma$ where $\delta_i$ is the set update of $\sigma_i$,
\begin{enumerate}
\item $S_i = S_{i - 1} \cup \nu_{i - 1}(\bar{s}^T)$ if $\delta_i = \{+S^T(\bar{s}^T)\}$,
\item $S_i = S_{i - 1} - \nu_i(\bar{s}^T)$ if $\delta_i = \{-S^T(\bar{s}^T)\}$ and
\item $S_i = (S_{i - 1} \cup \{\nu_{i-1}(\bar{s}^T)\}) - \{\nu_i(\bar{s}^T)\}$ if $\delta_i = \\ \{+S^T(\bar{s}^T), -S^T(\bar{s}^T)\}$,
\item $S_i = S_{i - 1}$ if $\delta_i = \emptyset$.
\end{enumerate}
Note that the only cases that can make $\rho_T$ invalid are those for which $\delta_i$ contains $-S^T(\bar{s}^T)$. 
\reviewer{I was not even sure what it meant for a case to be problematic.} \yuliang{fixed.}
Indeed, while a tuple can always be inserted, a tuple can be retrieved only if it belongs to $S^T$ (or is simultaneously inserted as in case (3)).
Thus, in order to show that the specified retrievals are possible, it is sufficient to prove the following.

\begin{lemma} \label{lem:sets}
Let $0 < i < \gamma$ be such that  (1)-(4) hold for $\{S_j\}_{0 \leq j < i}$.
If $\delta_i = \{-S^T(\bar{s}^T)\}$ then $\nu_i(\bar s^T) \in S_{i-1}$.
If $\delta_i = \{+S^T(\bar{s}^T), -S^T(\bar{s}^T)\}$ then either $\nu_i(\bar s^T) \in S_{i-1}$ or
$\nu_i(\bar s^T) = \nu_{i-1}(\bar s^T)$.
\end{lemma}

\begin{proof}
The key observations, which are easily checked by the construction of $\Lambda$, are the following:
\begin{itemize}\itemsep=0pt\parskip=0pt
\item [(\dag)] for every $k,k' \in [0,\gamma)$, if $\hat{\tau}_k$, $\hat{\tau}_{k'}$
are not input-bound and $I_k$ and $I_{k'}$ are not in the same life cycle, 
then $\nu_k(\bar{s}^T) \neq \nu_{k'}(\bar{s}^T)$.
\item [(\ddag)] for every $k,k' \in [0,\gamma)$, if $\hat{\tau}_k$, $\hat{\tau}_{k'}$
are input-bound,
$\nu_k(\bar{s}^T) = \nu_{k'}(\bar{s}^T)$ iff $\hat{\tau}_k = \hat{\tau}_{k'}$.
\end{itemize}

Now suppose that $0 < i < \gamma$, (1)-(4) hold for $\{S_j\}_{0 \leq j < i}$, and
$\delta_i = \{-S^T(\bar{s}^T)\}$.  
Suppose first that $\hat{\tau}_i$ is not input-bound. 
Let $L$ be the life cycle to which $I_i$ belongs, and $n< i$ be such that $I_n = \Ret(I_i)$.
By $(\dag)$, $\nu_k(\bar{s}^T) \neq \nu_{i}(\bar{s}^T)$ for every $n < k < i$.
Since (1)-(4) hold for all $j < i$, $\nu_n(\bar{s}^T) \in S_{i-1}$. By construction of $\Lambda$ (specifically the Life-Cycle chase condition),
$\nu_n(\bar{s}^T) = \nu_i(\bar{s}^T)$. Thus, $\nu_i(\bar{s}^T) \in S_{i-1}$.
The case when $\delta_i = \{+S^T(\bar{s}^T), -S^T(\bar{s}^T)\}$ is similar.

Now suppose $\hat{\tau}_i$ is input-bound and $\delta_i = \{-S^T(\bar{s}^T)\}$. 
By definition of symbolic local run, $\bar c_{i-1}(\hat{\tau}_i) = 1$.
Thus, there must exist a maximum $n < i$ such that $\hat{\tau}_n = \hat{\tau}_i$ and
for which the transition under $\sigma_n$ sets $\bar c_n(\hat{\tau}_i) = 1$. Since $\bar c_{i-1}(\hat{\tau}_i) = 1$ and $n$ is maximal,
there is no $j$, $n < j < i$ for which $\delta_j$ contains $-S^T(\bar{s}^T)$ and $\hat{\tau}_j = \hat{\tau}_i$. 
%There are two cases: (i) $\delta_n = \{+S^T(\bar{s}^T)\}$ and $\hat{\tau}_n = \hat{\tau}_i$,
%and there is no $j, n < j < i$ such that $-S^T(\bar{s}^T) \in \delta_j$ and $\hat{\tau}_j = \hat{\tau}_i$. 
From the above and $(\ddag)$ it easily follows that $\nu_n(\bar{s}^T) = \nu_i(\bar{s}^T)$ and $\nu_i(\bar{s}^T) \in S_{i-1}$.
The case when $\delta_i = \{+S^T(\bar{s}^T), -S^T(\bar{s}^T)\}$ is similar.

\end{proof}

It remains to prove properties (i)-(iii) of Lemma \ref{lem:sym-to-loc-finite}. First, 
as $\delta \geq 1$ and $\delta \geq h(T_c)$ for every $T_c \in child(T)$,
property (iii) is immediately satisfied. We next prove (i) and (ii).

\subsubsection*{Proof of property (i)} \label{sec:length}
We first introduce some additional notation. 
For each $i$ and $(i,e) \in \cale_i$, we denote by $[(i,e)]_i$ the equivalence class of $(i,e)$ wrt $\sim_i$.
And for $x \in \bar x^T$ we denote by 
$\reach_i(x,w)$ the unique equivalence class of $\sim_i$ reachable from $[(i, x_R)]_i$ by some navigation $w$ (if such class exists).
More precisely: 

\begin{definition} \label{def:reach}
For each $0 \leq i < \gamma$, we define $G(\sim_i)$ to be the labeled directed graph 
whose nodes are the equivalence classes of $\sim_i$ and where for each attribute $a$, 
there is an edge labeled $a$ from 
$E$ to $F$ if there exist $e \in E$ and $f \in F$ such that
$(i, e.a) \in \cale_i$ and $e.a \sim_{\tau_i} f$. 
Note that for each $E$ there is at most one outgoing edge labeled $a$.
For $x \in \bar x^T, x \not\sim_i \anull $ and sequence of attributes $w$, we denote by $\reach_i(x,w)$ the 
unique equivalence class $F$ of $\sim_i$ reachable from
$[(i,x)]_i$ by a path in $G(\sim_i)$ whose sequence of edge labels spells $w$, if such exists, and the empty set otherwise.
\end{definition}

By our choice of $h(T)$ and our construction of the $\lambda_i$'s, we can show that
\begin{lemma} \label{lem:length}
For every $0 \leq i < \gamma$ and expression $x_R.w$, if 
$\reach_i(x, w)$ is non-empty, then there exists an expression $x_R.\tilde{w}$ where $|\tilde{w}| < h(T)$
such that $\reach_i(x, w) = \reach_i(x, \tilde{w})$.
\end{lemma}
\begin{proof}
It is sufficient to show that for each $i$, $|G(\sim_i)| < h(T)$, where $|G(\sim_i)|$ is
the number of nodes in $G(\sim_i)$. Indeed, since there is a path from $[(i, x_R)]_i$ to 
$\reach_i(x, w)$ in $G(\sim_i)$, there must exist a simple such path, of length at most $|G(\sim_i)| < h(T)$.

To show that $|G(\sim_i)| < h(T)$, recall that $|G(\sim_i)|$ is bounded by the number of isomorphism types of $\sim_i$. 
Recall that $h(T) = 1 + |\bar{x}^T| \cdot F(\delta)$ 
where $F(n)$ is the maximum number of distinct paths of length at most $n$ 
starting from any relation in the foreign key graph FK.
By definition, for each variable $x$, the number of expressions $\{e | e = x_R.w, (i, e) \in \cale_i\}$
is bounded by $F(\delta)$. Thus the number of equivalence classes of $\sim_i$ is at most
$|\bar{x}^T| \cdot F(\delta) < h(T)$. So $|G(\sim_i)| < h(T)$.
\end{proof}

Property (i) now follows from Lemma \ref{lem:length}. Let $e = x_R.w$ be an expression for which $\nu_i^*(e)$ is defined.
By construction, $\reach_i(x,w) \subseteq \nu_i^*(e)$. By Lemma \ref{lem:length}, there exists $e' = x_R.w'$ where $|w'| < h(T)$ 
and $\reach_i(x,w') = \reach(x,w)$. It follows that 
$\nu_i^*(e')$ is defined and $\nu_i^*(e) \cap \nu_i^*(e') \neq \emptyset$. 
As $\nu_i^*(e)$ and $\nu_i^*(e')$ are equivalence classes of $\sim$, we have $\nu_i^*(e) = \nu_i^*(e')$,  proving (i).

\subsubsection*{Proof of property (ii)} \label{sec:invariant}

To show property (ii), it is sufficient to show 
an invariant which implies property (ii) and is satisfied throughout the construction of $\Lambda$.
For simplicity, we assume that the chase step in the construction of $\sim$ is divided into the following 3 phases.
\begin{itemize}\itemsep=0pt\parskip=0pt
\item The \emph{Segment Phase}. In this phase, we merge equivalence classes $E$ and $E'$
that satisfies either the Segment-Condition or the FD-condition.
\item The \emph{Life Cycle Phase}. In this phase, we merge equivalence classes $E$ and $E'$
that satisfies either the Life-Cycle-Condition or the FD-condition.
\item The \emph{Input Phase}. In this phase, we merge equivalence classes $E$ and $E'$
that satisfies either the Input-condition or the FD-condition.
\end{itemize}
It is easily seen that no chase step applies after the input phase.
Thus, the above steps compute the complete chase.

For each equivalence class $E$ of $\sim$, we let $i(E)$ be the set of indices $\{i | (i, e) \in E\}$ and
for each $i \in i(E)$, we denote by $E|_i$ the projection of $E$ on the navigation set $\cale_i$.
One can show that during the segment phase, for every $E$ of $\sim$, $i(E)$ are indices within the same segment.
During the life cycle phase, for every $E$ of $\sim$, $i(E)$ are indices within the same life cycle.
And during the input phase, $i(E)$ can be arbitrary indices.

The invariant is defined as follows.

\begin{lemma}{(Invariant of $\Lambda$)} \label{lem:invariant}
Throughout the construction of $\Lambda$, 
for every equivalence class $E$ of $\sim$, there exists variable $x \in \bar{x}^T$ and navigation $w$ 
where $|w| \leq h(T)$, such that for every $i \in i(E)$, $E|_i = \reach_i(x, w)$. 
\end{lemma}

Lemma \ref{lem:invariant} implies that for each equivalence class $E$ of $\sim$
and for each $\lambda_i$, $E$ is a superset of 
at most one equivalence class of $\lambda_i$. So $(i, e) \sim (i, e')$ implies $(i, e) \sim_i (i, e')$ 
thus $\Lambda | \cale_i = \lambda_i$ for every $0 \leq i < \gamma$, which implies
property (ii) of Lemma \ref{lem:sym-to-loc-finite}.

\begin{proof}
We consider each step of the construction of the global equality type $\sim$.
For the initialization step, the invariant holds by Lemma \ref{lem:length}.

For the Chase steps, assume that the invariant is satisfied before
merging two equivalence classes $E$ and $E'$. 
For each equivalence class $E$ of $\sim$, we denote by
$x(E)$ and $w(E)$ the variable and the navigation for $E$ as stated in Lemma \ref{lem:invariant}.
To show the invariant is satisfied after merging $E$ and $E'$,
it is sufficient to show that there exists variable $y$ and navigation $u$ where $|u| \leq h(T)$
such that for every $i \in i(E)$, $E|_i = \reach_i(y, u)$ and for every $i \in i(E')$, $E'|_i = \reach_i(y, u)$. 

Consider the segment phase. Suppose first that $E$ and $E'$ are merged due to the Segment-Condition. For simplicity, 
we let $x = x(E), x' = x(E'), w = w(E)$ and $w' = w(E')$.
If $E = [(i, y)]$ and $E' = [(i', y)]$ where $i, i'$ are indices within the same segment $S$, 
then by the assumption, we have $(i, y) \in \reach_i(x, w)$, so $y \sim_{\tau_i} x_R.w$.
As $i(E)$ are indices of a segment $S$, and by Lemma \ref{lem:segment-lifecycle},
we have that for every $j \in i(E)$, $y \sim_{\tau_j} x_R.w$,
so $E|_j = \reach_j(x, w) = \reach_j(y, \epsilon)$.
Similarly, we can show that for every $j \in i(E')$, 
$E'|_j = \reach_j(y, \epsilon)$.

Next suppose $E$ and $E'$ are merged due to the FD-condition. 
Thus, $E = [(i, e.a)]$ and $E' = [(i', e'.a)]$ where $(i, e) \sim (i', e')$.
Let $E^*$ be the equivalence class of $\sim$ that contains $(i, e)$ and $(i', e')$.
By the assumption, for $y = x(E^*)$ and $u = w(E^*)$, we have that 
$E^*|_i = \reach_i(y, u)$ so $(i, e) \in \reach_i(y, u)$.
By Lemma \ref{lem:length}, there exists navigation $\tilde{u}$ where $|\tilde{u}| < h(T)$ such that
$\reach_i(y, u) = \reach_i(y, \tilde{u})$. So  \\ $(i, e.a) \in \reach_i(y, \tilde{u}.a)$.
Then in $E$, by the hypothesis, we have $(i, e.a) \in \reach_i(x, w)$ so $\reach_i(y, \tilde{u}.a) = \reach_i \\ (x, w)$.
As $i(E)$ are indices of a segment $S$, and by Lemma \ref{lem:segment-lifecycle}, 
we have that for every $j \in i(E)$, for some relation $R_1$ and $R_2$,
$y_{R_1}.\tilde{u}.a \sim_{\tau_j} x_{R_2}.w$ so $E|_j = \reach_j(x, w) = \reach_j(y, \tilde{u}.a)$.
Similarly, we can show that for every $j \in i(E')$,  $E'|_j = \\ \reach_j(y, \tilde{u}.a)$.
Therefore, the invariant is preserved during the segment phase.

Consider the life cycle phase. 
We can show that the invariant is again preserved, together with the
following additional property: for each equivalence class $E$ of $\sim$ produced in this phase,
$x(E) \in \bar{x}^T_{in} \cup \bar{s}^T$. 
Suppose $E$ and $E'$ are merged due to the Life-Cycle Condition,
where $E = [(i, y)]$, $E' = [(i', y)]$
and $y \in \bar{x}^T_{in} \cup \bar{s}^T$. We have that
$E |_j = \reach_j(x, w) = \reach_j(y, \epsilon)$ for every $j \in i(E)$.
Indeed, by Lemma \ref{lem:segment-lifecycle} and because $i(E)$ are indices of some life cycle $L$,
$x_R.w \sim_{\tau_i} y$ implies that $x_R.w \sim_{\tau_j} y$ for every index $j$ of $L$.
Similarly, $E' |_j = \reach_j(y, \epsilon)$ for every $j \in i(E')$.
The case when $E$ and $E'$ are merged in this stage due to the FD-condition is similar to the above.
Following similar analysis, we can show that the input phase also preserves the invariant 
together with the property that for every $E$ produced at the input phase, $x(E) \in \bar{x}^T_{in}$.
This uses the fact that $\tau_i | \bar{x}^T_{in} = \tauin$ for every $0 \leq i < \gamma$. 
\end{proof}

This completes the proof of Lemma \ref{lem:sym-to-loc-finite}.

\subsubsection*{Infinite Local Symbolic Runs} \label{sec:inf-local}

In this section we show that Lemma \ref{lem:sym-to-loc-finite} can be extended to infinite periodic local symbolic runs, 
which together with finite runs are sufficient to represent accepted symbolic trees of runs 
by our VASS construction (see Lemma \ref{lem:vass-correctness}). 
Specifically, we show that we can extend the construction of the global isomorphism type
to infinite periodic $\trt$, while producing only {\em finitely} many equivalence classes.
This is sufficient to show that the corresponding database $D_T$ is finite.
We define periodic local symbolic runs next. 

\begin{definition}
A local symbolic run $\trt = \\ (\tauin, \tauout, \II)$ is periodic if $\gamma = \omega$ and 
there exists $n > 0$ and $0 < t \leq n$, such that for every $i \geq n$, 
symbolic instances $I_i = (\tau_i, \bar{c}_i, \sigma_i)$ and $I_{i-t} = (\tau_{i-t}, \bar{c}_{i-t}, \sigma_{i-t})$
satisfy that $(\tau_i, \sigma_i) = (\tau_{i - t}, \sigma_{i - t})$ and $\bar{c}_i \geq \bar{c}_{i - t}$.
The integer $t$ is called the \emph{period} of $\trt$.
\end{definition} 

From Lemma \ref{lem:vass-correctness} in Section \ref{sec:verification}, we have the following:
\begin{corollary}
It there exists a symbolic tree of runs $\Sym$ accepted by $\calb_\beta$, then there exists a symbolic tree of runs $\Sym'$ 
accepted by $\calb_\beta$ such that for every $\trt \in \Sym$, $\trt$ is finite or periodic.
\end{corollary}

The above corollary indicates that for verification, it is sufficient to consider only finite and periodic $\trt$.
So what we need to prove is:

\begin{lemma} \label{lem:sym-to-loc-infinite}
For every periodic local symbolic run  $\trt = \\ (\tauin, \tauout, \IIw)$, there exists a local run
$\rho_T = \\ (\nu_{in}, \nu_{out}, \RSw)$ on finite database $D_T \models \FD$ such that for every $i \geq 0$,
\begin{itemize}\itemsep=0pt\parskip=0pt
\item [(i)] for every expression $e = x_R.w$ where $\nu_i^*(e)$ is defined,
there exists expression $e' = x_R.w'$ where $|w'| \leq h(T)$ such that $\nu_i^*(e) = \nu_i^*(e')$,
\item [(ii)] for all expressions $e, e' \in \cale^+_T$ of $\tau_i$,
if $\nu_i^*(e)$ and $\nu_i^*(e')$ are defined, then $e \sim_{\tau_i} e'$ iff $\nu_i^*(e) = \nu_i^*(e')$, and
\item [(iii)] for $\delta = h(T_c)$ if $\sigma_i \in \{ \sigma_{T_c}^o, \sigma_{T_c}^c\}$ for some $T_c \in child(T)$ and $\delta = 1$ otherwise,
for every expression $e \in \cale^-_T = \cale^+_T - \{x_R.w | x \in \bar{x}^T, |w| > \delta\}$, $\nu_i^*(e)$ is defined.
\end{itemize}
\end{lemma}

Intuitively, if we directly apply the construction of $\rho_T$ and $D_T$ from Lemma \ref{lem:sym-to-loc-finite}
in the case of finite $\trt$, then each life cycle with non-input-bound $TS$-isomorphism types would be assigned 
with distinct sets of values, which could lead to an infinite $D_T$. 
However, for any two life cycles $L_1$ and $L_2$ which are disjoint in their timespan, 
reusing the same values in $L_1$ and $L_2$ does not cause any conflict. 
And in particular, if $L_1$ and $L_2$ are identical on the sequence of $\tau_i$'s and $\sigma_i$'s,
they can share exactly the same set of values.

Thus at a high level, our goal is to show that any periodic local symbolic run $\trt$ 
can be partitioned into finitely many subsets of identical life cycles with disjoint timespans.
Unfortunately, this is generally not true if we pick the $\Ret$ function arbitrarily
(recall that $\Ret$ defines the set of life cycles). This is because an arbitrary $\Ret$ may yield 
life cycles whose timespans have unbounded length. If the timespans overlap,
it is impossible to separate the life cycles into finitely many subsets of life cycles with disjoint timespans.  
So instead of picking an arbitrary $\Ret$ as in the finite case, 
we show that for periodic $\trt$ we can construct $\Ret$ such that 
the timespan of each life cycle has bounded length. This implies that
we can partition the life cycles into finitely many subsets of identical life cycles with disjoint timespans, as desired.
Finally we show that given the partition, we can construct the local run $\rho_T$ together with a finite $D_T$.

We first define the equivalence relation between life cycles.
\begin{definition}
Segments $S_1 = \{(I_i, \sigma_i)\}_{a_1 \leq i \leq b_1}$ and \\
 $S_2 = \{(I_i, \sigma_i)\}_{a_2 \leq i \leq b_2}$ are equivalent, denoted as
$S_1 \equiv S_2$, if $\{(\tau_i, \sigma_i)\}_{a_1 \leq i \leq b_1} = \{(\tau_i, \sigma_i)\}_{a_2 \leq i \leq b_2}$.
\end{definition}
\begin{definition}
A segment $S = \{(I_i, \sigma_i)\}_{a \leq i \leq b}$ is \emph{static} 
if $I_a \in \cali^-$, $I_b \in \cali^+$ and $\tau_a | \bar{s}^T = \tau_b | \bar{s}^T$.
A segment $S$ is called \emph{dynamic} if it is not static.
\end{definition}

When we compare two life cycles $L_1$ and $L_2$, we can ignore their static segments 
since they do not change the content of $S^T$. We define equivalence of two life cycles as follows.

\begin{definition}
For life cycle $L$, let $dym(L) = \{S_i\}_{1 \leq i \leq k}$ be the sequence of dynamic segments of $L$.
Two life cycles $L_1$ and $L_2$ are equivalent, denoted as $L_1 \equiv L_2$, if 
$|dym(L_1)| = |dym(L_2)|$ and for $dym(L_1) = \{S_i^1\}_{1 \leq i \leq k}$ and 
$dym(L_2) = \{S_i^2\}_{1 \leq i \leq k}$, for every $1 \leq i \leq k$, $S^1_i \equiv S^2_i$.
\end{definition}

Note that for each life cycle $L$, the number of dynamic segments within $L$ is bounded by $|\bar{s}^T|$ 
since within $L$, each variable in $\bar{s}^T$ is written at most once by returns of child tasks of $T$.
For a task $T$, as the number of $T$-isomorphism types is bounded, the number of services is bounded and the length of a segment is bounded
because each subtask can be called at most once, the number of equivalence classes of segments is bounded.
And since the number of dynamic segments is bounded within the same life cycle, the number of equivalence classes
of life cycles is also bounded. Thus,
\begin{lemma} \label{lem:equivalent}
The equivalence relation $\equiv$ on life cycles has finite index.
\end{lemma}

Our next step is to show that one can define a $\Ret$ function so that all life cycles have bounded timespans.
The timespan of a life cycle is defined as follows:
\begin{definition}
The timespan of a life cycle $L$, denoted by $sp(L)$, is an interval $[a, b]$ where $a$ is the index of the first symbolic instance 
of the first dynamic segment of $L$ and $b$ is the index of the last symbolic instance of the last dynamic segment.
\end{definition}

Consider an equivalence class $\call$ of life cycles. 
Suppose that for each $L \in \call$, the length of $sp(L)$ is bounded by some constant $m$.
Then we can further partition $\call$ into $m$ subsets $\call_0, \dots, \call_{m-1}$ of life cycles with disjoint timespan
by assigning each $L \in \call$ where $sp(L) = [a, b]$ to the subset $\call_{k}$ where $k = a \ \texttt{mod} \ m$.

We next show how to construct the function $\Ret$.
In particular, we construct a periodic $\Ret$ such that there is a short gap between each pair of inserting and retrieving instances. 
This is done in several steps, illustrated in Figure \ref{fig:retrieve}.
\begin{enumerate}
\item Initialize $\Ret$ to be an arbitrary one-to-one mapping 
with domain $\{I_i | I_i \in \cali^-, 0 \leq i \leq n\}$
such that for every $I_i = \Ret(I_j)$, $i < j$ and $\hat{\tau}_i = \hat{\tau}_j$
(recall that $\hat{\tau}_i = \tau_i | \bar{x}^T_{in} \cup \bar{s}^T$).
\item For every $j \in [n + 1, n + t]$, for $j' = j - t$ and for $i'$ being the index where $I_{i'} = \Ret(I_{j'})$,
\begin{itemize}\itemsep=0pt\parskip=0pt
\item [(i)] if $i' \in [n - t + 1, n]$, then for $i = i' + t$,
let $\Ret \leftarrow \Ret[I_{j + k \cdot t} \mapsto I_{i + k \cdot t} | k \geq 0]$, otherwise
\item  [(ii)] if $i' \in [0, n - t]$, then we pick $i \in [n - t + 1, n]$ satisfying that
$I_{i} \in \cali^+$, $\hat{\tau}_{i} = \hat{\tau}_j$ and $I_i$ is currently not in the range of $\Ret$. 
Then we let $\Ret \leftarrow \Ret[I_{j + k \cdot t} \mapsto I_{i + k \cdot t} | k \geq 0]$.
\end{itemize}
\end{enumerate}
At step 2 for the case $i' \in [0, n - t]$, the $i$ that we picked always exists for the following reason.
For every $TS$-isomorphism type $\hat{\tau}$, let 
\begin{itemize}\itemsep=0pt\parskip=0pt
\item $M_{\hat{\tau}}^-$ be the number of symbolic instances in $\cali^-$ with $TS$-isomorphism type $\hat{\tau}$ and indices in $[n - t + 1, n]$ 
that retrieves from symbolic instances with indices in $[0, n - t]$, and
\item $M_{\hat{\tau}}^+$ be the number of symbolic instances in $\cali^+$ with $TS$-isomorphism type $\hat{\tau}$ and indices in $[n - t + 1, n]$ 
that is NOT retrieved by symbolic instances with indices in $[n - t + 1, n]$.
\end{itemize}

We have $M_{\hat{\tau}}^+ - M_{\hat{\tau}}^- = \bar{c}_n(\hat{\tau}) - \bar{c}_{n - t}(\hat{\tau}) \geq 0$.
So for every $I_{i'} = \Ret(I_{j'})$ where $j' \in [n - t + 1, n]$ and $i' \in [0, n - t]$, 
we can always find a unique $i \in [n - t + 1, n]$ such that $I_i \in \cale^+$, $\hat{\tau}_i = \hat{\tau}_{j'} = \hat{\tau}_{i'}$ 
and $I_i$ is not retrieved by any retrieving instances with indices in $[n - t + 1, n]$.

\begin{figure}[!h]
\centering
\includegraphics[scale=0.55]{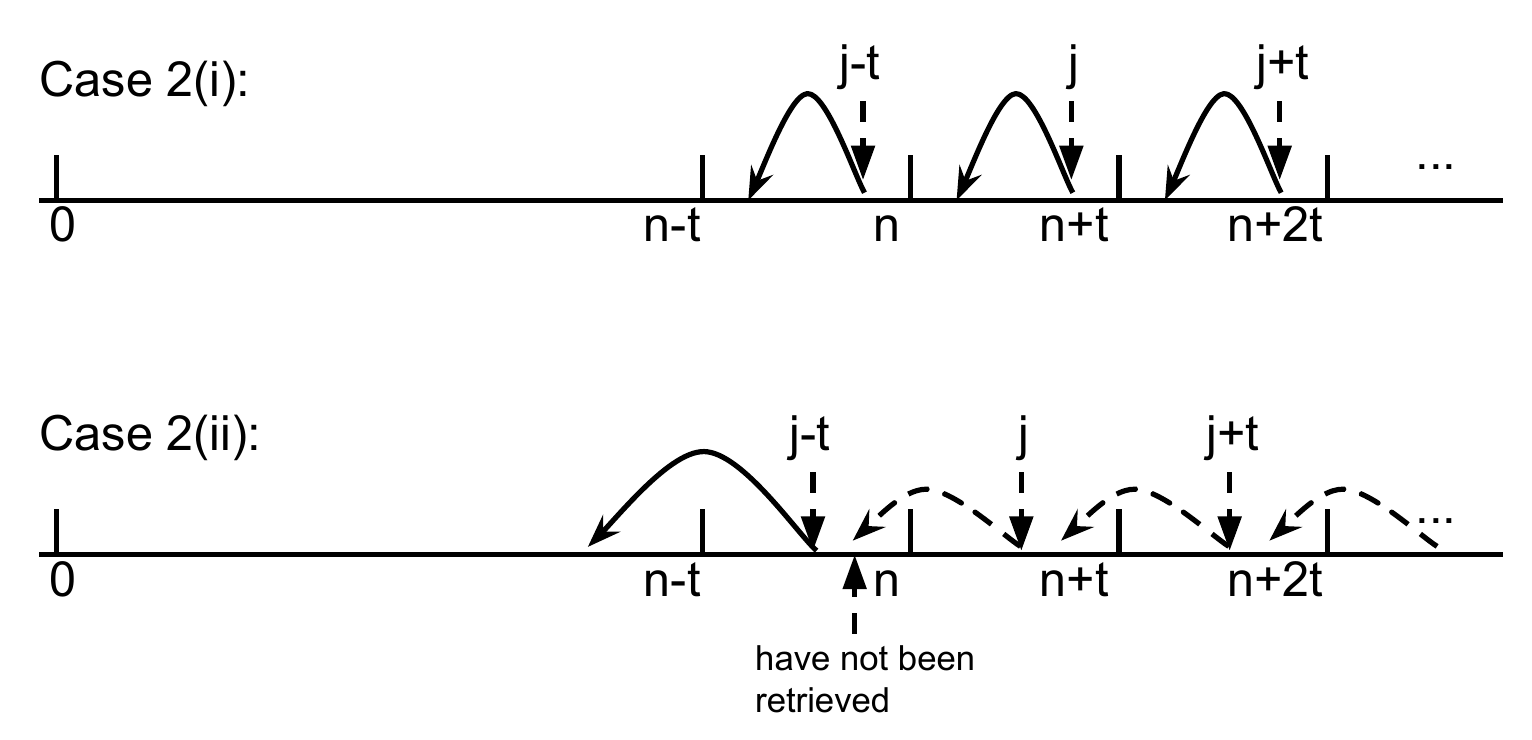}
\caption{Construction of $\Ret$}
\label{fig:retrieve}
\end{figure}

Let us fix the function $\Ret$ constructed above.
We first show the following:

\begin{lemma} \label{lem:periodic}
For every periodic $\trt$, and $j > n$, $I_i = \\ \Ret(I_j)$ implies that $j - i \leq 2t$ and
$I_{i+t} = \Ret \\ (I_{j+t})$.
\end{lemma}

\begin{proof}
By construction, for every $I_i = \Ret(I_j)$ where $j > i > n$, $I_{i + t} = \Ret(I_{j + t})$.
And it is also guaranteed that for the indices $i$ and $j$, either (1) $i$ and $j$ are both in the same range 
$[n + t k + 1, n + t (k + 1)]$ for
some $k \geq 0$, or (2) $i \in [n + t k + 1, n + t (k + 1)]$ and $j \in [n + t (k + 1) + 1, n + t (k + 2)]$ 
for some $k \geq 0$.
In both cases, $j - i \leq 2t$.
\end{proof}

For every life cycle $L$, for every pair of consecutive dynamic segments $S$ and $S'$, we denote by $gap(S, S')$
the number of static segments in between $S$ and $S'$. To show that $sp(L)$ is bounded, 
it is sufficient to show that $gap(S, S')$ is bounded for every pair of consecutive dynamic segments 
$S$ and $S'$. For every segment $S$, we denote by $a(S)$ the index of the first symbolic instance of $S$.
For every segment $S$ where $a(S) > n$, we let $p(S) = (a(S) - n - 1) \ \texttt{mod} \ t$. 

For every pair of consecutive dynamic segments $S$ and $S'$ and
by periodicity of $\Ret$, there are no two static segments $T$ and $T'$ in $L$ in between $S$ and $S'$
such that $a(S) < a(T) < a(T') < a(S') $ and $p(T) = p(T')$. Thus in $L$,
the number of static segments in between $S$ and $S'$ is at most $n + t$. 
Then by Lemma \ref{lem:periodic}, the number of symbolic instances in between any pair of consecutive segments is
bounded by $\max(2t, n)$ so $gap(S, S') \leq (n + t) \cdot \max(2t, n + t)$. And as the number of dynamic segments in $L$
is bounded by $|\bar{s}^T|$ and the length of each segment is at most $2 |child(T)|$, it follows that:

\begin{lemma} \label{lem:timespan}
For every periodic local symbolic run $\trt$ and life cycle $L$ of $\trt$,
$|sp(L)|$ is bounded by $m = (n + t) \cdot \max(2t, n + t) \cdot (|\bar{s}^T| + 1) \cdot 2 |child(T)|$.
\end{lemma}

So for a possibly infinite set of life cycles $\call$ where $|sp(L)| \leq m$ for each $L \in \call$, 
$\call$ can be partitioned into sets $\call_0, \dots, \call_{m - 1}$ by
assigning each life cycle $L \in \call$ where $sp(L) = [a, b]$ to the set $\call_{a \ \mathtt{mod} \ m}$.
So for every $\call_i$ and two distinct $L_1, L_2$ in $\call_i$
where $sp(L_1) = [a_1, b_1]$ and $sp(L_2) = [a_2, b_2]$, we have $a_1 \neq a_2$.
Assume $a_1 < a_2$. Then as $a_1 \equiv a_2 \pmod m$, $a_2 - a_1 \geq m$. And since $b_1 - a_1 + 1 < m$,
$L_1$ and $L_2$ are disjoint. 
Thus, given Lemma \ref{lem:equivalent} and Lemma \ref{lem:timespan}, we have
\begin{lemma}
Every local symbolic run $\trt$ can be partitioned into finitely many subsets of life cycles 
such that for each subset $\call$, if $L_1 \in \call$, $L_2 \in \call$ and $L_1 \neq L_2$ then 
$L_1 \equiv L_2$ and $sp(L_1) \cap sp(L_2) = \emptyset$.
\end{lemma}

Next, we show how we can construct the local run $\rho_T$ and finite database $D_T$ from $\trt$ using the partition.
We first construct global isomorphism type $\Lambda = (\cale, \sim)$ of $\trt$ using the approach for the finite case. 
Then we merge equivalent segments in $\Lambda$ as follows to obtain a new global isomorphism type with finitely many equivalence classes.
To merge two equivalent segments $S_1 = \{(I_i, \sigma_i)\}_{a_1 \leq i \leq a_1 + l}$ 
and  $S_2 = \\ \{(I_i, \sigma_i)\}_{a_2 \leq i \leq a_2 + l}$,
first for every $0 \leq i \leq l$ and for every $x \in \bar{x}^T$, 
we merge the equivalence classes $[(a_1 + i, x)]$ and $[(a_2 + i, x)]$ of $\sim$.
Then we apply the chase step (i.e. the FD-condition) to make sure the resulting database satisfies $\FD$.

The new $\Lambda$ is constructed as follows.
For every two segments $S_1 = \{(I_i, \sigma_i)\}_{a \leq i \leq b}$ and $S_2 = \{(I_i, \sigma_i)\}_{c \leq i \leq d}$,
we define that $S_1$ precedes $S_2$, denote by $S_1 \prec S_2$, if $b < c$.
For each subset $\call$ and for each pair of life cycles $L_1, L_2 \in \call$
where $dym(L_1) = \{S_i^1\}_{1 \leq i \leq k}$ and $dym(L_2) = \{S_i^2\}_{1 \leq i \leq k}$,
\begin{itemize}\itemsep=0pt\parskip=0pt
\item for $1 \leq i \leq k$, merge $S_i^1$ and $S_i^2$,
\item for $1 \leq i < k$, 
for every static segments $S_1 \subseteq L_1$ and $S_2 \subseteq L_2$ where
$S_i^1 \prec S_1 \prec S_{i+1}^1$, $S_i^2 \prec S_2 \prec S_{i+1}^2$ and $S_1 \equiv S_2$,
merge $S_1$ and $S_2$, and
\item for every pair of static segments $S_1 \subseteq L_1$ and $S_2 \subseteq L_2$
where $S_k^1 \prec S_1$, $S_k^2 \prec S_2$ and $S_1 \equiv S_2$, merge $S_1$ and $S_2$.
\end{itemize}
Finally, $\rho_T$ and $D_T$ are constructed following the same approach as in the finite case.
In the above construction, 
as the number of subsets of life cycles is finite, and for each $\call$,
the number of dynamic segments is bounded and the number of equivalence classes of static segments is bounded,
the number of equivalence classes of $\Lambda$ is also finite so $D_T$ is finite.

By an analysis similar to the finite case, 
we can show that $\rho_T$ and $D_T$ satisfy property (i)-(iii)
in Lemma \ref{lem:sym-to-loc-infinite} and $D_T \models \FD$.
In particular, to show property (ii), we can show the same invariant as in Lemma \ref{lem:invariant}, 
the invariant holds because every pair of merged segments are equivalent.

Finally, to show Lemma \ref{lem:sym-to-loc-infinite}, it remains to show
that $\rho_T$ is a valid local run. Similar to the finite case, it is sufficient to show that
\begin{lemma}\label{lem:sets-infinite}
For every $i \geq 0 $, if $\delta_i = \{-S^T(\bar{s}^T)\}$ then $\nu_i(\bar s^T) \in S_{i-1}$.
If $\delta_i = \{+S^T(\bar{s}^T), -S^T(\bar{s}^T)\}$ then either $\nu_i(\bar s^T) \in S_{i-1}$ or
$\nu_i(\bar s^T) = \nu_{i-1}(\bar s^T)$.
\end{lemma}

\begin{proof}
The following can be easily checked by the construction of $\Lambda$:
\begin{itemize}\itemsep=0pt\parskip=0pt
\item [(i)] for every pair of distinct life cycles $L$ and $L'$ where $sp(L) \cap sp(L') \neq \emptyset$, 
for every $I_k \in L$ and $I_{k'} \in L'$, if $\hat{\tau}_k$, $\hat{\tau}_{k'}$
are not input-bound then $\nu_k(\bar{s}^T) \neq \nu_{k'}(\bar{s}^T)$, and
\item  [(ii)]for every pair of life cycles $L$ and $L'$ where $sp(L) \cap sp(L') = \emptyset$, 
if $I_i,I_j \in L$, $I_j = \Ret(I_i)$, $\hat{\tau}_i$ is not input-bound, 
$I_k \in L'$ for $j < k < i$ and $\nu_k(\bar{s}^T) = \nu_i(\bar{s}^T) = \nu_j(\bar{s}^T)$,
then $I_k$ is contained in a static segment of $L'$.
\item [(iii)] for every $k,k' \geq 0$, if $\hat{\tau}_k$, $\hat{\tau}_{k'}$ are input-bound,
$\nu_k(\bar{s}^T) = \nu_{k'}(\bar{s}^T)$ iff $\hat{\tau}_k = \hat{\tau}_{k'}$.
\end{itemize}

Consider the case when $\delta_i = \{-S^T(\bar{s}^T)\}$ and $\hat{\tau_i}$ is not input-bound.
Let $I_j = \Ret(I_i)$ and $L$ be the life cycle that contains $I_i$. 
Consider $I_k$ where $j < k < i$ and let $L'$ be the life cycle containing $I_k$.
If $sp(L) \cap sp(L') \neq \emptyset$, by (i), $\nu_i(\bar{s}^T) \neq \nu_k(\bar{s}^T)$.
If $sp(L) \cap sp(L') = \emptyset$, by (ii), the segment containing $I_k$ 
is static, so it does not change $S^T$. Thus, for every segment $S$ between $I_j$ and $I_i$,
the tuple $\nu_i(\bar{s}^T)$ remains in $S^T$ after $S$. So $\nu_i(\bar{s}^T) \in S_{i-1}$.
The case when $\delta_i = \{-S^T(\bar{s}^T), +S^T(\bar{s}^T)\}$ is similar.

The proof for the case when $\hat{\tau_i}$ is input-bound is the same as the proof for Lemma \ref{lem:sets}.
\end{proof}

This completes the proof of Lemma \ref{lem:sym-to-loc-infinite}.

\subsubsection*{Symbolic Trees of Runs}

Finally, we show Lemma \ref{lem:sym-to-loc} by providing a recursive construction of a tree of runs 
$\Tree$ and database $D$ from any symbolic tree of runs $\Sym$ where all local symbolic runs 
are either finite or periodic, using Lemmas \ref{lem:sym-to-loc-finite} and \ref{lem:sym-to-loc-infinite}.
Intuitively, the construction simply applies the two lemmas to each node $\trt$ of $\Sym$ to obtain
a local run $\rho_T$ with a local database $D_T$. Then the local runs and databases are combined into
a tree of local runs recursively by renaming the values in each $\rho_T$ and $D_T$ 
in a bottom-up manner, reflecting the communication among local runs via input and return variables.

Formally, we first define recursively the construction function $F$ where $F(\Sym_T) = (\Tree_T, D_T)$ where
$\Sym_T$ is a subtree of $\Sym$ and $(\Tree_T, D_T)$ are the resulting subtree of local runs
and database instance. $F$ is defined as follows.

If $T$ is a leaf task, then $\Sym_T$ contains a single local symbolic run $\trt$.
We define that $F(\Sym_T) = F(\trt) = (\rho_T, D_T)$ where $\rho_T$ and $D_T$ are the
local run and database instance shown to exist in Lemmas \ref{lem:sym-to-loc-finite} and \ref{lem:sym-to-loc-infinite}
corresponding to $\trt$.

If $T$ is a non-leaf task where the root of $\Sym_T$ is $\trt = (\tauin, \tauout, \II)$, then we first let 
$(\rho_T, D_{\mathtt{root}}) = F(\trt)$. Next, let $J = \{i | \sigma_i = \sigma_{T_c}^o, T_c \in child(T)\}$.
For every $i \in J$, we denote by $\Sym_i$ the subtree rooted at the child of $\trt$ where
the edge connecting it with $\trt$ is labeled $i$ and let $\tilde{\rho}_i$ be the root of $\Sym_i$.
We denote by $(\Tree_i, D_i) = F(\Sym_i)$ and by $\rho_i$ the local run at the root of $\Tree_i$.
From the construction in Lemmas \ref{lem:sym-to-loc-finite} and \ref{lem:sym-to-loc-infinite},
we assume that the domains of $D_{\mathtt{root}}$ and the $D_i$'s are equivalence classes of local expressions.
We first define the renaming function $r$ whose domain is $\bigcup_{i \in J} \adom(D_i)$ 
as follows.
\begin{enumerate}
\item Initialize $r$ to be the identity function.
\item For every $i \in J$, for every expression $x_R.w$ where $x \in \bar{x}^{T_c}_{in}$ and $\nu_{in}^*(x_R.w)$ is defined,
for $y = f_{in}(x)$, let $r \leftarrow r[\nu_{in}^*(x_R.w) \mapsto \nu_i^*(y_R.w)]$.
Note that $\nu_{in}^*$ is defined wrt $\nu_{in}$ of $\rho_i$ and $D_i$ and $\nu_i^*$ is defined wrt $\nu_i$ of $\rho_T$ and $D_{\mathtt{root}}$.
And we shall see next that for every such $x_R.w$, if $\nu_{in}^*(x_R.w)$ is defined, then $\nu_i^*(y_R.w)$ is also defined.
\item For every $i \in J$ where $\tilde{\rho}_i$ is a returning local symbolic run where
the index of the corresponding $\sigma_{T_c}^c$ in $\trt$ is $j$, 
for every expression $x_R.w$ where $x \in \bar{x}^{T_c}_{ret}$ and $\nu_{out}^*(x_R.w)$ is defined,
for $y = f_{out}^{-1}(x)$, %$S_1 = \nu_{out}^*(x_R.w)$ and $S_2 = \nu_j^*(y_R.w)$,
let $r \leftarrow r[\nu_{out}^*(x_R.w) \mapsto \nu_j^*(y_R.w)]$.
\end{enumerate}
We denote by $r(D)$ the database instance obtained by replacing each value $v \in dom(r)$ in $D$ with $r(v)$
and denote by $r(\Tree)$ the tree of runs obtained by replacing each value $v \in dom(r)$ in $\Tree$ with $r(v)$.

Then if $\trt$ is finite, we define $F(\Sym_T) = (\Tree_T, D_T)$
where $D_T = D_{\mathtt{root}} \cup \bigcup_{i \in J} r(D_i)$ and
$\Tree_T$ is obtained from $\Sym_T$ by replacing the root of $\Sym_T$ with $\rho_T$ and each subtree
$\Sym_i$ with $r(\Tree_i)$. 

If $\trt$ is periodic where the period is $t$ and the loop starts with index $n$, we define $F(\Sym_T) = (\Tree_T, D_T)$
where $D_T = D_{\mathtt{root}} \cup \bigcup_{i \in J, i < n} r(D_i)$ and
$\Tree_T$ is obtained from $\Sym_T$ by replacing the root of $\Sym_T$ with $\rho_T$  and each subtree
$\Sym_i$ with $r(\Tree_{i'})$, where $i' = i$ if $i < n$ otherwise $i' = n + (i - n) \texttt{ mod } t$. 

To prove the correctness of the construction, we first need to show that for every $\Sym_T$ and $(\Tree_T, D_T) = F(\Sym_T)$,
$D_T$ is a finite database satisfying $\FD$ and $\Tree_T$ is a valid tree of runs over $D_T$.
Let $\trt$ and $\rho_T$ be the root of $\Sym_T$ and $\Tree_T$ respectively. 
We show the following:
\begin{lemma}
For every symbolic tree of runs $\Sym_T$ where $(\Tree_T, D_T) = F(\Sym_T)$,
$D_T$ is a finite database satisfying $\FD$, $\Tree_T$ is a valid tree of runs over $D_T$, and
$(\rho_T, D_T)$ satisfies properties (i)-(iii) in Lemma \ref{lem:sym-to-loc-finite} and \ref{lem:sym-to-loc-infinite}.
\end{lemma}

\begin{proof}
We use a simple induction. For the base case, where $T$ is a leaf task, the lemma holds trivially. 
For the induction step, assume that for each $i \in J$, $D_i$ is finite and satisfies $\FD$, $\Tree_i$ is a valid tree of runs over $D_i$, 
and $(\rho_i, D_i)$ satisfies property (i)-(iii). 

For each $i \in J$, where $\tilde{\rho}_i$ is a local symbolic run of task $T_c \in child(T)$, 
we first consider the connection between $\tilde{\rho}_i$ and $\trt$ via input variables. 
As $\rho_i$ satisfies properties (i) and (ii), for every expressions $x_R.w$ and $x'_{R'}.w'$ in 
the input isomorphism type $\tauin$ of $\tilde{\rho}_i$,
if $\nu_{in}^*(x_R.w)$ and $\nu_{in}^*(x'_{R'}.w')$ are defined, then
$\nu_{in}^*(x_R.w) = \nu_{in}^*(x'_{R'}.w')$ iff $x_R.w \sim_{\tauin} x'_{R'}.w'$.
And by definition of symbolic tree of runs,
we have that $\tauin = f_{in}^{-1}(\tau_i) | (\bar{x}^{T_c}_{in}, h(T_c))$.
So for $y = f_{in}(x)$ and $y' = f_{in}(x')$,
$\nu_{in}^*(x_R.w) = \nu_{in}^*(x'_{R'}.w')$ iff $y_R.w \sim_{\tau_i} y'_{R'}.w'$.
Then as $\rho_T$ satisfies (ii) and (iii),
$\nu^*_i(y_R.w)$ and $\nu^*_i(y'_{R'}.w')$ are defined and 
$\nu^*_i(y_R.w) = \nu^*_i(y'_{R'}.w')$ iff $y_R.w \sim_{\tau_i} y'_{R'}.w'$ so
$\nu^*_i(y_R.w) = \nu^*_i(y'_{R'}.w')$ iff $\nu_{in}^*(x_R.w) = \nu_{in}^*(x'_{R'}.w')$.

If $\tilde{\rho}_i$ is returning, using the same argument as above, we can show the following.
Let $j$ be the index of the corresponding returning service $\sigma_{T_c}^c$.
Let $f$ be the function where $f(x) = \begin{cases} f_{in}(x) \text{, } x \in \bar{x}^{T_c}_{in} 
\\ f_{out}^{-1}(x) , x \in \bar{x}^{T_c}_{ret} \end{cases}$ and let $\nu$ be the valuation where
$\nu(x) = \begin{cases} \nu_{in}(x) , x \in \bar{x}^{T_c}_{in} 
\\ \nu_{out}(x) , x \in \bar{x}^{T_c}_{ret} \end{cases}$, where $\nu_{in}$ and $\nu_{out}$ are the input and output
valuation of $\rho_i$.
For all expressions $x_R.w$ and $x'_{R'}.w'$ where
$x, x' \in \bar{x}^{T_c}_{ret} \cup \bar{x}^{T_c}_{in}$, if $\nu^*(x_R.w)$ and $\nu^*(x'_{R'}.w')$ are defined, 
then for $y = f(x)$ and $y' = f(x')$, $\nu_j^*(y_R.w)$ and $\nu_j^*(y'_{R'}.w')$ are also defined and
$\nu_j^*(y_R.w) = \nu_j^*(y'_{R'}.w')$ iff $\nu^*(x_R.w) = \nu^*(x'_{R'}.w')$.

Given this, after renaming, $D_{\mathtt{root}}$ and $r(D_i)$ can be combined consistently.
Also, one can easily check that $\Tree_T$ is a valid tree of runs where $(\rho_T, D_T)$ satisfies properties (i)-(iii) and $D_T \models \FD$.
And $D_T$ is a finite database because it is the union of $D_{\mathtt{root}}$ and finitely many $r(D_i)$'s and by the hypothesis,
$D_{\mathtt{root}}$ and the $D_i$'s are finite.
\end{proof}

Finally, to complete the proof of correctness of the construction, we note: 
\begin{lemma}
For every full symbolic tree of runs $\Sym$ where all local symbolic runs in $\Sym$ are either
finite or periodic, for $(\Tree, D) = F(\Sym)$ and
every $\hltlfo$ property $\varphi$, $\Sym$ is accepted by $\calb_\varphi$ 
iff $\Tree$ is accepted by $\calb_\varphi$ on $D$.
\end{lemma}

The above follows immediately from the fact that by construction, for every task $T$
and local symbolic run $\trt = \\ (\tauin, \tauout, \II)$ in $\Sym$ where the corresponding local run in $\Tree$
is $\rho_T = (\nu_{in}, \nu_{out}, \RS)$, for every condition $\pi$ over $\bar{x}^T$ and $0 \leq i < \gamma$,
$\tau_i \models \pi$ iff $D \models \pi(\nu_i)$.

This completes the proof of Lemma \ref{lem:sym-to-loc}, and the only-if part of Theorem \ref{thm:actual-symbolic}.

\subsection{Proof of Lemma \ref{lem:vass-correctness}}
%\yuliang{$\bar{c}_{ib}$ needs to be propagated here.}
%\victor{It looks like you did this, right ?}
\noindent
The proof is by induction on the task hierarchy $\calh$.

\vspace{4mm}
\noindent
{\bf Base Case}
Consider $\calr_T(\tauin, \tauout, \beta)$ where $T$ is a leaf task.
As $T$ has no subtask, $dom(\bar{o})$ is always empty so $\bar o$ can be ignored. Note that, by definition, there can be
no blocking path of $\calv(T,\beta)$.

For the \emph{if} part, consider $(\tauin, \tauout, \beta) \in \calr_T$.
Suppose first that $\tauout \neq \bot$.
By definition, there exists a finite local symbolic run $(\tauin, \tauout, \II)$ accepted by
$B(T, \beta)$, where $\gamma \in \mathbb{N}$ and $\sigma_{\gamma-1} = \sigma^c_T$.
Consider an accepting computation $\{q_i\}_{0 \leq i < \gamma}$ of $B(T, \eta)$ on $\II$, such that
$q_{\gamma-1} \in Q_{fin}$.  We can construct a returning path $P = \{(p_i, \bar{z}_i)\}_{0 \leq i < \gamma}$ of $\calv(T, \beta)$
where for each state  \\ $p_i = (\tau_i, \sigma_i, q_i, \bar{o}_i, \bar{c}_{ib}^i)$,
$(\tau_i, \sigma_i, q_i)$ is obtained directly from $\II$ and $\{q_i\}_{0 \leq i < \gamma}$,
$\bar{z}_i = \bar{c}_i$, and $\bar{c}_{ib}^i$ is the projection of $\bar{c}_i$ to 
input-bound $TS$-isomorphism types.

Now suppose $\tauout = \bot$. By definition, and since $T$ is a leaf task, there exists an infinite symbolic run 
$(\tauin, \tauout, \\ \IIw)$ accepted by $B(T, \beta)$. Consider
the sequence $\{q_i\}_{0 \leq i < \omega}$ of states in an accepting computation of $B(T, \eta)$ on $\IIw$. There must exist
$q_f \in Q_{inf}$ such that for infinitely many $i$, $q_i = q_f$. So
we can construct a path $P = \{(p_i, \bar{z}_i)\}_{0 \leq i < \omega}$ of $\calv(T, \beta)$
where for each state $p_i = (\tau_i, \sigma_i, q_i, \bar{o}_i, \bar{c}_{ib}^i)$ is obtained in the same way as in the case 
where $\tauout \neq \bot$.
It is sufficient to show that there exists a finite prefix $\{(p_i, \bar{z}_i)\}_{0 \leq i \leq n}$
of $P$ such that there exists $m < n$ such that $(\tau_m, \sigma_m, q_m, \bar{c}^m_{ib}) = (\tau_n, \sigma_n, q_n, \bar{c}^n_{ib})$, $q_m = q_n = q_f$,
and $\bar{z}_m \leq \bar{z}_n$. By the pigeonhole principle,
there exist $\tau$, $\sigma$, $\bar{c}_{ib}$ and an infinite $J \subseteq \mathbb{N}$
such that $(\tau_j, \sigma_j, \bar{c}^j_{ib} q_j) = (\tau, \sigma, \bar{c}_{ib}, q_f)$ 
for every $j \in J$. Consider the sequence $\{\bar z_j \mid j \in J\}$.
Next, there exists an infinite $J_1 \subseteq J$ such that $\{\bar z_j \mid j \in J_1\}$ is non-decreasing in the first dimension.
A straightforward induction shows that there exists an infinite $J_{|\bar z|} \subseteq J$ such that $\{\bar z_j \mid j \in J_{|\bar z|}\}$ is non-decreasing
in all dimensions. Now consider $m, n \in J_{|\bar z|}$, $m < n$. The sequence
$(p_0, \bar{z}_0), \ldots, (p_m, \bar{z}_m), \ldots, (p_n, \bar{z}_n)$ is a lasso path of $\calv(T,\beta)$.

For the \emph{only-if} direction, if there exists a returning path in $\calv(T, \beta)$, then by definition, $\tauin$ and $\tauout$
together with the sequence $\{(I_i,\sigma_i)\}_{0 \leq i \leq n}$ where each
$(I_i,\sigma_i)$ is obtained directly from $(p_i, \bar{z}_i)$ is a valid local symbolic run $\tilde{\rho}_T$.
And $\tilde{\rho}_T$ is accepted by $B(T, \beta)$ since $q_n$ is in $Q^{fin}$.
If there exists a lasso path in $\calv(T, \beta)$, then we can obtain a finite sequence \\ $\{(I_i, \sigma_i)\}_{0 \leq i \leq n}$ similar to
above. And we can construct \\ $\{(I_i,\sigma_i)\}_{0 \leq i < \omega}$ by repeating the subsequence from index $m+1$ to index $n$
infinitely many times.
As $q_n = q_f \in Q^{inf}$, $(\tauin, \bot, \{(I_i,\sigma_i)\}_{0 \leq i < \omega})$ is an infinite local symbolic run accepted by $B(T, \beta)$,
so $(\tauin, \bot, \beta) \in \calr_T$.

\vspace{4mm}
\noindent
{\bf Induction}
Consider a non-leaf task $T$, and suppose the statement is true for all its children tasks.

For the \emph{if} part, suppose $(\tauin, \tauout, \beta) \in \calr_T$. Then there exists an adorned symbolic tree of runs $\Sym_T$
with root $\trt = (\tauin, \tauout, \II)$
accepted by $\calb_{\bar \beta}$. We construct a path $P = \{(p_i, \bar{z}_i)\}_{0 \leq i < \gamma}$
of $\calv(T, \beta)$ as follows. The transitions in $\trt$ caused by internal services are treated as in the base case.
Suppose that $\sigma_i = \sigma^o_{T_c}$ for some child $T_c$ of $T$. Then there is an edge labeled $(i, \beta^{T_c})$
from $\trt$ to a symbolic tree of runs accepted by $\calb_{\bar \beta^{T_c}}$,
rooted at a run $\trtc$ of $T_c$ with input $\tauin^{T_c}$ and output
$\tauout^{T_c}$. Thus, $(\tauin^{T_c}, \tauout^{T_c}, \beta^{T_c}) \in \calr_{T_c}$ and
$\calv(T, \beta)$ can make the transition from $(p_{i-1}, \bar{z}_{i-1})$ to $(p_i, \bar{z}_i)$ as in its definition
(including the updates to $\bar o$).
If $\tauout^{T_c} \neq \bot$ then there exists a minimum $j > i$ for which $\sigma_j = \sigma^c_{T_c}$ and
once again $\calv(T, \beta)$ can make the transition from $(p_{j-1}, \bar{z}_{j-1})$ to $(p_j, \bar{z}_j)$ as in its definition,
mimicking the return of $T_c$ using the isomorphism type $\tauout^{T_c}$ stored in $\bar o(T_c)$.
Now consider the resulting path $P =  \{(p_i, \bar{z}_i)\}_{0 \leq i < \gamma}$.
By applying a similar analysis as in the base case, if $\gamma \neq \omega$ and $\tauout \neq \bot$, then $P$ is a returning path.
If $\gamma \neq \omega$ and $\tauout = \bot$, then $P$ is a blocking path.
If $\gamma = \omega$, then there exists a prefix $P'$ of $P$ such that $P'$ is a lasso path.

For the \emph{only-if} direction, let $P$ be a path of $\calv(T, \beta)$,
starting from a state
$p_0 = (\tau_0, \sigma_0, q_0, \bar{o}_0, \bar{c}^{0}_{ib})$ where $\tau_0 | \bar{x}^T_{in} = \tauin$.
If $P$ is a returning path, let $v_n = (\tau_n, \sigma_n, q_n, \bar{o}_n, \bar{c}^{n}_{ib})$ be its last state and
$\tauout = \tau_n | (\bar{x}^T_{in} \cup \bar{x}^T_{ret})$.  If $P$ is not a returning path, then $\tauout = \bot$.
From $P$ we can construct a adorned symbolic tree of runs $\Sym_T$ accepted by $\calb_{\bar \beta}$ as follows.
The root of $\Sym_T$ is a local symbolic run $\trt$
constructed analogously to the construction in the only-if direction in the base case.
Then for each $\sigma_i = \sigma_{T_c}^o$, by the induction hypothesis, there exists a symbolic tree
of runs $\Sym_{T_c}$ whose root has input isomorphism type $\tauin^{T_c}$, output isomorphism type
$\tauout^{T_c}$ and is accepted by $\calb_{\beta^{T_c}}$ (note that $\tauin^{T_c}$, $\tauout^{T_c}$ and
$\beta^{T_c}$ are uniquely defined by $P$ and $i$).
We connect $\Sym_T$ with $\Sym_{T_c}$ with an edge labeled $(i, \beta^{T_c})$.

If $P$ is a returning or blocking path, then $\Sym_T$ is accepted by $\calb_{\bar \beta}$.
If $P$ is a lasso path, then we first modify the root $\trt$ of $\Sym_T$ by repeating
the subsequence from $m + 1$ to $n$ infinitely, then for each integer $i$ such that
$m + 1 \leq i \leq n$ and $\Sym_T$ is connected with some $\Sym_{T_c}$
with edge labeled index $(i, \beta^{T_c})$, for each repetition $I_{i'}$ of symbolic instance $I_i$,
we make a copy of $\Sym_{T_c}$ and connect $\Sym_T$ with $\Sym_{T_c}$ with edge labeled $(i', \beta^{T_c})$.
The resulting $\Sym_T$ is accepted by $\calb_{\bar \beta}$. Thus, $(\tauin, \tauout, \beta) \in \calr_T$.

%\yuliang{Complexity analysis is moved back.}
%\subsection{Complexity analysis and proof of Theorem \ref{thm:complexity1}}
%
%\victor{More details on the complexity analysis should go here.}

\subsection{Complexity of Verification without Arithmetic} \label{sec:complexity-no-arith-app}

Let $\Gamma$ be a HAS and $\varphi$ an $\hltlfo$ formula over $\Gamma$. 
Recall the VASS $\calv(T,\beta)$ constructed for each task $T$ and assignment $\beta$ to $\Phi_T$.
According to the discussion of the complexity of verification in Section \ref{sec:verification},
checking whether $\Gamma \not\models \varphi$ can be done in $O(h \log n \cdot 2^{c \cdot d \log (d)})$ nondeterministic space, 
where $c$ is a constant, $h$ is the depth of $\calh$, and $n, d$ bound the number of states, resp. vector dimensions of $\calv(T,\beta)$
for all $T$ and $\beta$. We will estimate these bounds using the maximum number of $T$-isomorphism types, denoted $M$,
and the maximum number of $TS$-isomorphism types, denoted $D$. We also denote by $N$ the size of $(\Gamma, \varphi)$.
To complete the analysis, the specific bounds $M$ and $D$ will be computed for acyclic, linear-cyclic, 
and cyclic schemas, as well as with and without artifact relations. 

By our construction, the vector dimension of each $\calv(T,\beta)$ is the number of $TS$-isomorphism types, so bounded by $D$.
The number of states is at most the product of the number of distinct $T$-isomorphism types, 
the number states in $B(T, \beta)$, the number of all possible $\bar{o}$ and the number of possible states of $\bar{c}_{ib}$.
And since the number of $T_c$-isomorphism types is no more than the number of $T$-isomorphism types if $T_c$ is child of $T$,
the number of all possible $\bar{o}$ is at most $(3 + M)^{|child(T)|} \leq (3+M)^N$.
Note that the number of states in $B(T, \beta)$ 
is at most exponential in the size of the $\hltlfo$ property $\varphi$ (extending the classical construction \cite{WolperVardiSistla}).
Thus, $n = M \cdot 2^{O(N)} \cdot (3+M)^N \cdot 2^D$ bounds the number of states of all $\calv(T,\beta)$. 
It follows that $O(h \log n \cdot 2^{c \cdot d \log (d)}) = O(h \cdot N \cdot \log M \cdot 2^{c\cdot D \cdot \log D})$, yielding the complexity of 
checking $\Gamma \not\models \varphi$. 
Thus, checking whether $\Gamma \models \varphi$ can be done in
$O( h^2 \cdot N^2 \log^2 M \cdot 2^{c \cdot D \log D} )$ deterministic space by Savitch's Theorem \cite{sipser},
for some constant $c$.

For artifact systems with no artifact relation, the bounds degrade to 
$O( h \cdot N \log M )$ and $O( h^2 \cdot N^2 \log^2 M )$.
%since the entire VASS can be constructed on-the-fly (the states of $B(T, \beta)$ is constructed
%using the classic approach in \cite{VW:LICS:86}).
%\victor{It would be more uniform to use lower case notation for all parameters rather than having some be upper and others lower case.}
%\yuliang{I can fix that later. I use capital letters for $N$ and $D$ because they will be replaced later in the analysis.}

The number of $T$- and $TS$-isomorphism types depends on the type of the schema $\db$ of $\Gamma$, as described next.
In our analysis, we denote by $r$ the number of relations in $\db$ and $a$ the maximum arity of relations in $\db$. 
We also let $k = \max_{T \in \calh} |\bar{x}^T|$, $s = \max_{T \in \calh}|\bar{s}^T|$ and $h$ be the height of $\calh$.

% And we denote by $K$ the maximum number of numeric expressions in any $T$-isomorphism types.

\vspace{5mm}
\noindent
{\bf Acyclic Schema} if $\db$ is acyclic, then the length of each expression 
in the navigation set is bounded by the number of relations in $\db$. 
So the size of the navigation set of each 
$T$-isomorphism type is at most $a^r k$. 
The total number of $T$-isomorphism types is at most the product of
the number of possible navigation sets and the number of possible equality types.
So $M = (r + 1)^k \cdot (a^r k)^{a^r k}$ is a bound for the number of $T$-isomorphism types for every $T$.

For $TS$-isomorphism types, we note that within the same path in $\calv(T, \beta)$, all $TS$-isomorphism types have the same
projections on $\bar{x}^T_{in}$ since the input variables are unchanged throughout a local symbolic run. 
So within each query of (repeated) reachability, each $TS$-isomorphism type can be represented by (1) the equality connections
from expressions starting with $x \in \bar{x}^T_{in}$ to expressions starting with $x \in \bar{s}^T$ and (2) the equality connections
within expressions starting with $x \in \bar{s}^T$.
For (1), the total number of all possible connections is at most ${M_1}^{M_2}$ where $M_1$ is the number of expressions starting with 
$x \in \bar{x}^T_{in}$
and $M_2$ is the number of expressions starting with $x \in \bar{s}^T$. For (2), the total number of all possible connections is 
at most $M_2^{M_2}$. Note that $M_1 \leq a^r k$ and $M_2 \leq a^r s$.
So the total number of $TS$-isomorphism type is at most 
$D = (r+1)^s \cdot (a^r k \cdot a^r s)^{a^r s} =(r+1)^s \cdot (a^{2r} k \cdot s)^{a^r s}$.
So for $\db$ of fixed size and $S^T$ of fixed arity, 
the number of $T$-isomorphism type is exponential in $k$ and 
the number of $TS$-isomorphism type is polynomial in $k$.

By substituting the above values of $M$ and $D$ in 
the space bound $O( h^2 \cdot N^2 \log^2 M \cdot 2^{c \cdot D \log D} )$, we obtain:

\begin{theorem} \label{thm:acyclic}
For HAS $\Gamma$ with acyclic schema and $\hltlfo$ property $\varphi$ over $\Gamma$,
% if $\Gamma$ has fixed acyclic schema and set $S^T$ has fixed arity for every task $T$, then
$\Gamma \models \varphi$ can be checked in $O(\exp(N^{c_1}))$ deterministic space, where $c_1 = O(a^{r \log r} s)$.
If $\Gamma$ does not contain artifact relations, then 
$\Gamma \models \varphi$ can be checked in $c_2 \cdot N^{O(1)}$ deterministic space,
where $c_2 = O(a^{2r} \log^2 a^r)$.
\end{theorem}

Note that if $\db$ is a Star schema \cite{starschema1,starschema2}, which is a special case of acyclic schema, then the 
size of the navigation set is at most $a r k$ instead of $a^r k$.
So verification has the complexities stated in Theorem \ref{thm:acyclic}, 
with constants $c_1 = O(a r s)$ and $c_2 = O({ar}^2 \log^2 ar)$ respectively.

Note that with the simulation used in Lemma \ref{lem:simplification},
the number of variables is at most quadratic in the original number of variables.
This only affects the constants in the above complexities.

\vspace{4mm}
\noindent
{\bf Linearly-Cyclic Schema} 
Consider the case where $\db$ is linearly cyclic. To bound the number of $T$- and $TS$-isomorphism types,
it is sufficient to bound $h(T)$, which equals to $1 + k \cdot F(\delta)$ where 
$\delta = \max_{T_c \in child(T)} \{h(T_c)\}$ if $T$ is a non-leaf task and $\delta = 1$
if $T$ is a leaf. And recall that $F(\delta)$ is the maximum number of distinct
paths of length at most $\delta$ starting from any relation in the foreign key graph FK. 
If $\db$ is linearly cyclic, then by definition, the graph of cycles in FK form an acyclic graph
$G$ (each node in $G$ is a cycle in the FK graph and there is an edge from cycle $u$ to cycle $v$ 
iff there is an edge from some node in $u$ to some node in $v$ in FK).

Consider each path $P$ of length at most $\delta$ in FK.
$P$ can be decomposed into a list of subsequences of nodes, where each subsequence 
consists of nodes within the same cycle in FK (as shown in Figure \ref{fig:linearly-cyclic}).

\begin{figure}[!h]
\centering
\includegraphics[scale=0.7]{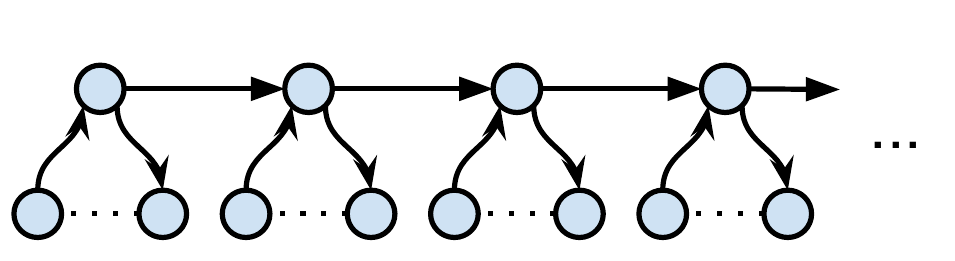}
\caption{Path in Linearly-Cyclic Foreign Key Graph}
\label{fig:linearly-cyclic}
\end{figure}

So $F(\delta)$ can be bounded by the product of (1) the number of distinct paths in $G$ starting from any cycle
and (2) the maximum number distinct paths of length at most $\delta$ formed using subsequences of nodes from cycles within the same path in $G$.
It is easy to see that (1) is at most $a^r$. And since the length of a path in $G$ is at most $r$,
(2) is at most $\delta^r$. Thus $F(\delta)$ is bounded by $a^r \cdot \delta^r = (a \cdot \delta)^r$.

%\yuliang{I added the above discussion and fixed a mistake in the following counting.}

So if $\db$ is linearly cyclic, then 
$h(T)$ is bounded by $1 + a^r k$ if $T$ is a leaf task and
$h(T)$ is bounded by $1 + (a \cdot \delta)^r   \cdot k$ if $T$ is non-leaf task
where $\delta = \max_{T_c \in child(T)} \{h(T_c)\}$. 
By solving the recursion, for every task $T$, we have that
$h(T) \leq c \cdot (a \cdot k)^{r \cdot h} $ for some constant $c$. 
So the size of the navigation set of each $T$-isomorphism type
is at most $c \cdot (a \cdot k)^{r(h+1)}$. Thus the number of
$T$- and $TS$-isomorphism types are bounded by $(r+1)^k \cdot (c \cdot (a \cdot k)^{r(h+1)})^{c \cdot (a \cdot k)^{r(h+1)}}$. 
By an analysis similar to that for acyclic schemas, we can show that
%\begin{lemma}
%$\calr_T$ can be computed in $O((n^h \log (n^h) t + \log q )2^{\exp((c \cdot n)^h)} ) $ space 
%for some constant $c$, if $\db$ is linearly-cyclic and has constant size.
%\end{lemma}
%So given Lemma \ref{lem:actual-symbolic} and \ref{lem:sym-to-loc}:
%\begin{theorem} \label{thm:linear}
%For hierarchical artifact system $\Gamma$ and $\hltlfo$ property $\varphi$,
%if $\Gamma$ has fixed linear-cyclic schema, then
%$\Gamma \models \varphi$ can be verified in 
%$O(2\hexp((c \cdot N)^{h+1}))$ nondeterministic space\footnote{$k\hexp$ is the tower of exponential functions of height $k$}
%where $N$ is the size of $(\Gamma, \varphi)$ and $c$ is a constant.
%\end{theorem}

\begin{theorem} \label{thm:linear}
For HAS $\Gamma$ with linearly-cyclic schema and $\hltlfo$ property $\varphi$ over $\Gamma$,
% if $\Gamma$ has fixed acyclic schema and set $S^T$ has fixed arity for every task $T$, then
$\Gamma \models \varphi$ can be checked in  \\ 
$O(2\hexp( N^{c_1 \cdot h}))$ deterministic space
where $c_1 = O(r)$. If $\Gamma$ does not contain artifact relations, then 
$\Gamma \models \varphi$ can be checked in $O(N^{c_2 \cdot h})$ deterministic space
where $c_2 = O(r)$.
\end{theorem}
%\yuliang{Any idea to make these result look better?}
%\yuliang{One alternative for complexity with set is $O(2\hexp( h \cdot (c_1 N)^{h + 1} \log(c_1 N) ))$ where $c_1 = O(a^r)$
%and the complexity without set is
%$O(h^4 (c_2 N)^{2h + 4} \log^2 (c_2 N))$ where $c_2 = O(a^r)$.}

\vspace{4mm}
\noindent
{\bf Cyclic Schema} If $\db$ is cyclic, then each relation in FK has at most $a$ outgoing edges so
$F(\delta)$ is bounded by $a^{\delta}$.
So $h(T) = O(k \cdot a^{\delta})$ 
%\victor{What is $n$?} \yuliang{Should be $k$. Fixed.} 
where $\delta = 1$ if $T$ is a leaf task and 
$\delta = \\ \max_{T_c \in child(T)} h(T_c)$ otherwise. Solving the recursion yields $h(T) = h\hexp(O(N))$. 
By pursuing the analysis similarly to the above, we obtain the following:
\begin{theorem} \label{thm:cyclic}
For HAS $\Gamma$ with cyclic schema and $\hltlfo$ property $\varphi$ over $\Gamma$,
% if $\Gamma$ has fixed acyclic schema and set $S^T$ has fixed arity for every task $T$, then
$\Gamma \models \varphi$ can be checked in $(h+2)\hexp(O(N))$ deterministic space.
If $\Gamma$ does not contain artifact relations, then 
$\Gamma \models \varphi$ can be checked in $h\hexp(O(N))$ deterministic space.
\end{theorem}

%\yuliang{I moved the discussion of different types of schema here.}
%\victor{OK}
To summarize, the schema type determines the size of the navigation set, and hence the complexity of verification, as follows
($h$ the height of the task hierarchy and $N$ the size of $(\Gamma, \varphi)$).
\begin{itemize}
\item Acyclic schemas are the least general, yet sufficiently expressive for many applications. A special case of
acyclic schema is the Star schema \cite{starschema1,starschema2} (or Snowflake schema) which is widely used in modeling business process data.
For fixed acyclic schemas, the navigation sets have constant depth.
\item Linearly-cyclic schemas extend acyclic schemas but yield higher complexity.
In general, the size of the navigation set is exponential in $h$ and polynomial in $N$.
Linearly-cyclic schemas allow very simple cyclic foreign key relations such as a single Employee-Manager relation.
They include important special cases such as schemas where each relation has at most one foreign key attribute.
\item Cyclic schemas allow arbitrary foreign keys but also come with much higher complexity (a tower of exponentials of height $h$),
as the size of navigation sets become hyper-\\exponential wrt $h$.
%\victor{What is the size of the navigation set?} \yuliang{fixed.}
\end{itemize}

%Without artifact relations, the complexity is in general logarithmic to the number of number of $T$-isomorphism types.
%With artifact relations, the complexity is exponential to the number of $TS$-isomorphism types.

\section{Verification with Arithmetic}\label{app:arithmetic}

\subsection{Review of Quantifier Elimination}\label{sec:qe}

The quantifier elimination ($\qe$) problem for the reals can be stated as follows.

\begin{definition}
For real variables $Y = \{y_i\}_{1 \leq i \leq l}$
and a formula $\Phi(Y)$ of the form $$(Q_1 x_1) \dots (Q_k x_k) F(y_1 \dots y_l, \\ x_1 \dots x_k)$$
where $Q_i \in \{\exists, \forall\}$ and $F(y_1 \dots y_l, x_1 \dots x_k)$ is a 
Boolean combination of polynomial inequalities with integer coefficients, 
the quantifier elimination problem is to output a quantifier-free formula $\Psi(Y)$
such that for every $Y \in \mathbb{R}^l$, $\Phi(Y)$ is true iff $\Psi(Y)$ is true.
\end{definition}

The best known algorithm for solving the \qe problem for the reals has 
time and space complexity doubly-exponential in the number of quantifier alternations and singly-exponential
in the number of variables. When applying \qe in verification of HAS, we are only
interested in formulas that are existentially quantified. According to Algorithm 14.6 of \cite{cad-book}, 
the result for this special case can be stated as follows:
\begin{theorem} \label{thm:qe}
For existentially quantified formula $\Phi(Y)$, an equivalent quantifier-free formula $\Psi(Y)$ can be computed in time and space
$(s \cdot d)^{O(k) O(l)} $, where $s$ is the number of polynomials in $\Phi$, $d$ is the maximum degree
of the polynomials, $k$ is the quantifier rank of $\Phi$ and $l = |Y|$.
\end{theorem}
%\yuliang{The complexity can also be written as $(s \cdot d)^{O(k \cdot l + k + l)} $. }
%\victor{Then should this be the stated complexity above?} \yuliang{fixed.}
%\victor{I modified the statement so that you get the corollary when you set $l = 0$}
%\yuliang{I changed it back to $(s \cdot d)^{O(k) O(l)} $. This should be a right way of writing the complexity.}
Note that in the special case when $l = 0$, quantifier elimination simply checks satisfiability.
Thus we have:
\begin{corollary} \label{cor:arith-sat}
Satisfiability over the reals of a Boolean combination $\Phi$ of polynomial inequalities with integer coefficients
can be decided in time and space $(s \cdot d)^{O(k)} $, where
$s$ is the number of polynomials in $\Phi$, $d$ is the maximum degree
of the polynomials, and $k$ is the number of variables in $\Phi$.
\end{corollary}
Also in \cite{cad-book}, it is shown that if the bit-size of coefficients in $\Phi$ is bounded by $\tau$, 
then the bit-size of coefficients in $\Psi$ is bounded by $\tau \cdot d^{O(k) O(l)}$.
%\victor{What is $l$ here? It should be zero, right? If so O(k)O(l) = 0.}
%\yuliang{When $l = 0$, $O(l)$ is constant but not zero. Right?}

\subsection{Review of General Real Algebraic Geometry}\label{sec:algebraic-geometry}
We next review a classic result in general real algebraic geometry.
For a given set of polynomials $\calp = \{P_1, \dots, P_s\}$ over $k$ variables $\{x_i\}_{1 \leq i \leq k}$, a sign condition of $\calp$
is a mapping $\sigma: \calp \mapsto \{-1, 0, +1\}$.
We denote by $c(\sigma, \calp)$ the semialgebraic set $\{x | x \in \mathbb{R}^k, 
sign(P(x)) = \sigma(P), \forall P \in \calp\}$
called the \emph{cell} of the sign condition $\sigma$ for $\calp$.

We use the following result from \cite{heinzt89, basu1996number}:
\begin{theorem} \label{thm:num-cells}
Given a set of polynomials $\calp$ with integer coefficients 
over $k$ variables $\{x_i\}_{1 \leq i \leq k}$, the number of distinct non-empty cells, 
namely $$\#\{\sigma : \calp \mapsto \{-1, 0, +1\} \ | c(\sigma, \calp)  \neq \\ \emptyset \},$$
is at most $(s \cdot d)^{O(k)}$, where $s = |\calp|$ and $d$ is the maximum degree of polynomials in $\calp$.
\end{theorem}
%\victor{We need to distinguish two things: the coefficients of the polynomial and the structure over which they are interpreted, presumably the reals.
%I assume the coefficients are always integers in this section, right? This should be made clear.}

Given a set of polynomials $\calp$, we can use the following naive approach to compute the set
of sign conditions resulting in non-empty cells.
We simply enumerate sign conditions of $\calp$ and discard sign conditions that results in empty cells or 
cells equivalent to any recorded sign conditions known to be non-empty. 
Checking whether a cell is empty and checking whether two cells are equivalent
can be reduced to checking satisfiability of a formula of polynomial inequalities.
By Corollary \ref{cor:arith-sat}, this naive approach takes space $(s \cdot d)^{O(k)}$.
\yuliang{The running time of the above naive algorithm is $3^s \cdot (s \cdot d)^{O(k)}$.
For linear arithmetic, there is an algorithm with running time $(s \cdot d)^{O(k)}$. 
I think for general polynomials, the upper bound for running time is still open.}

\begin{theorem} \label{thm:naive-approach}
Given a set of polynomials $\calp$ over $\{x_i\}_{1 \leq i \leq k}$,
the set of non-empty cells $\{\sigma : \calp \mapsto \{-1, 0, +1\} 
\ | \ c(\sigma, \calp) \neq \emptyset \}$ defined by $\calp$ can be computed in space $(s \cdot d)^{O(k)}$
where $s = |\calp|$ and $d$ is the maximum degree of polynomials in $\calp$.
\end{theorem}

\subsection{Cells for Verification}\label{sec:cells}

Intuitively, in order to handle arithmetic in our verification framework, 
we need to extend each isomorphism type $\tau$ with a set of polynomial inequality constraints over the set of numeric
expressions in the extended navigation set $\cale_T^+$. 

We say that an expression $e$ is numeric if $e = x$ for some numeric variable $x$
or $e = x_R.w$ and the last attribute of $w$ is numeric.
For each task $T$, we denote by $\calert$ the set of
numeric expressions of $T$ where for each $x_R.w \in \calert$, $|w| \leq h(T)$. 

The constraints over the numeric expressions are represented by a non-empty cell $c$ (formally defined below).
When a service is applied, the arithmetic parts of the conditions are evaluated against $c$.
And for every transition $I \goto{\sigma'} I'$ where $c, c'$ are the cells of $I, I'$ respectively, 
if any variables are modified by the transition, then the projection of $c'$ 
onto the preserved numeric expressions has to \emph{refine} the projection of 
$c$ onto the preserved numeric expressions. 
Similar compatibility checks are required when a child task returns to its parent.

We introduce some more notation.
For every $T \in \calh$, we consider polynomials in the polynomial ring $\mathbb{Z}[\calert]$.
For each polynomial $P$, we denote by $var(P)$ the set of numeric expressions mentioned in $P$
and for a set of polynomials $\calp$, we denote by $var(\calp)$ the set $\bigcup_{P \in \calp} var(P)$.
For $\calp \subset \mathbb{Z}[\calert]$ and $\cale \subseteq \calert$, 
we denote by $\calp | \cale$ the set of polynomials $\{P | P \in \calp, var(P) \subseteq \cale \}$.

We next define the cells used in our verification algorithm. 
At task $T$, for a set of numeric expressions $\cale \subseteq \calert$ and a set of
polynomials $\calp$ where $var(\calp) \subseteq \cale$, 
we define the cells over $(\cale, \calp)$ as follows.
\begin{definition} \label{def:cell}
A cell $c$ over $(\cale, \calp)$ is a subset of $\mathbb{R}^{|\cale|}$ for which
there exists a sign condition $\sigma$ of $\calp$ such that $c = c(\sigma, \calp)$.
\end{definition}

For $\calp \subset \mathbb{Z}[\calert]$, we 
denote by $\calk(\calp, \cale)$ the set of cells over $(\cale, \calp | \cale)$.
Namely, $\calk(\calp, \cale) = \{c(\sigma, \calp | \cale) | \sigma \in \calp | \cale \mapsto \{-1, 0, +1\}\}$.
And we denote by $\calk(\calp)$ the set of cells \\ $\bigcup_{\cale \subseteq \calert} \calk(\calp, \cale)$.

Compatibility between cells is tested using the notion of 
refinement. Intuitively, a cell $c$ refines another cell $c'$ if $c$ can be obtained by
adding extra numeric expressions and/or constraints to $c'$. Formally,
\begin{definition}
For cell $c$ over $(\cale, \calp)$ and cell $c'$ over $(\cale', \calp')$ where 
$c = c(\sigma, \calp)$ and $c' = c(\sigma', \calp')$, 
we say that $c$ refines $c'$, denoted by $c \sqsubseteq c'$,
if $\cale' \subseteq \cale$, $\calp' \subseteq \calp$ and $\sigma | \calp' = \sigma'$.
Note that if $\cale = \cale'$, then $c \sqsubseteq c'$ iff $c \subseteq c'$.
\end{definition}

We next define the projection of a cell onto a set of variables.
For each cell $c$ over $(\cale, \calp)$ where $\cale \subseteq \calert$ 
and variables $\bar{x} \subseteq \bar{x}^T$,
the projection of $c$ onto $\bar{x}$, denoted by $c | \bar{x}$,
is defined to be the projection of $c$ onto the expressions $\cale | \bar{x}$
where $\cale | \bar{x} = \{e \in \cale | e = x_R.w \lor e = x, x \in \bar{x} \}$.
By the Tarski-Seidenberg theorem \cite{tarski-seidenberg}, $c | \bar{x}$ is a union of disjoint cells.
Also, the projections $c | \bar{x}$ can be obtained by quantifier elimination.
Let $\Phi(c)$ be the conjunctive formula defining $c$ using polynomials in $\calp$.
Then by treating $\cale | \bar{x}$ as the set of free variables,
the formula $\Psi(c)$ obtained by eliminating $\cale - \cale | \bar{x}$ from $\Phi(c)$
defines $c | \bar{x}$. We denote by $\poly(c, \bar{x})$ the set of polynomials mentioned in $\Psi(c)$.
It is easy to see that $c | \bar{x}$ is a union of cells over $(\cale | \bar{x}, \poly(c, \bar{x}))$.

The following notation is useful for checking compatibility between a cell and the projection of another cell:
we define that a cell $c$ refines another cell $c'$ wrt to projection to $\bar{x}$, 
denoted as $c \sqsubseteq_{\bar{x}} c'$, if there exists a cell $\tilde{c} \subseteq c' | \bar{x}$
such that $c \sqsubseteq \tilde{c}$.

Finally, we introduce notations relative to variable passing between parent task and child task.
For each task $T$ and $T_c \in child(T)$, 
we denote by $\calertct$ the set of numeric expressions
$\{e | e \in \calertc, e = x \lor e = x_R.w, x \in \bar{x}^{T_c}_{in} \cup \bar{x}^{T_c}_{ret} \}$.
%$\calertc \cap \left( \bar{x}^{T_c}_{in} \cup \bar{x}^{T_c}_{ret} \cup 
%\{x_R.w | x \in \bar{x}^{T_c}_{in} \cup \bar{x}^{T_c}_{ret}\} \right)$.
In other words, $\calertct$ is the subset of expressions in $\calertc$ connected with expressions in $\calert$ by calls/returns of $T_c$.
Let $f_{in}, f_{out}$ be the input and output mapping between $T$ and $T_c$.
For each expression $e \in \calertct$, we define $e^{T_c \rightarrow T}$ to be an expression in $\calert$ as follows.
If $e = x$, then $e^{T_c \rightarrow T} = (f_{in} \circ f_{out}^{-1})(x)$.
If $e = x_R.w$, then $e^{T_c \rightarrow T} = ((f_{in} \circ f_{out}^{-1})(x))_R.w$.
For a set of variables $\cale \subseteq \calertct$, we define
$\cale^{T_c \rightarrow T}$ to be $\{e^{T_c \rightarrow T} | e \in \cale \}$.
For a polynomial $P$ over $\calertct$ where $T_c \in child(T)$, 
we denote by $P^{T_c \rightarrow T}$ the polynomial obtained by replacing in $P$ each
numeric expression $e$ with $e^{T_c \rightarrow T}$.
For a cell $c$ of $T_c$ where $c = c(\sigma, \calp)$ and $var(P) \subseteq \calertct$ for every $P \in \calp$,
we let $c^{T_c \rightarrow T}$ to be the cell of $T$ which equals 
$c(\sigma', \calp')$, where $\calp' = \{P^{T_c \rightarrow T} | P \in \calp\}$
and $\sigma'$ is a sign condition over $\calp'$ such that 
$\sigma'(P^{T_c \rightarrow T}) = \sigma(P)$ for every $P \in \calp$.

\subsection{Hierarchical Cell Decomposition}\label{sec:cell-decomp}

We now introduce the Hierarchical Cell Decomposition. 
Intuitively, for each task $T$, we would like to compute a set of polynomials $\calp$ and 
a set of cells $\calk_T$ such that for each subset $\cale$ of $\calert$, 
the set of cells over $(\cale, \calp | \cale)$ in $\calk_T$ is a partition of $\mathbb{R}^{|\cale|}$.

% \yuliang{do we need to talk about how we get from refinement to equality check? }

The set of cells $\calk_T$ satisfies the property that for the set of polynomials $\calp$ 
mentioned at any condition of $T$ in the specification $\Gamma$
and $\hltlfo$ property $\varphi$, each cell $c \in \calk_T$ uniquely defines the sign condition of $\calp$. 
This allows us to compute the signs of any polynomial in any condition in the local symbolic runs.
In addition, for each pair of cells $c, c' \in \calk_T$, we require that the projection of $c$ and $c'$ 
to the input variables $\bar{x}^T_{in}$ (and $\bar{x}^T_{in} \cup \bar{s}^T$) be disjoint or identical.
So to check whether two cells $c$ and $c'$ of two consecutive symbolic instances in a local symbolic run are compatible when applying
an internal service, we simply need to check whether their projections on $\bar{x}^T_{in}$ are equal (note that refinement is implied by
equality).
Finally, for each child task $T_c$ of $T$, for each cell $c \in \calk_T$ and $c' \in \calk_{T_c}$,
$c$ uniquely defines the sign condition for the set of polynomials that defines $c' | \bar{x}^{T_c}_{in} $
and $c' | (\bar{x}^{T_c}_{ret} \cup \bar{x}^{T_c}_{in})$. This reduces to cell refinement the problem of 
checking compatibility when child tasks
are called or return. 

The Hierarchical Cell Decomposition is formally defined as follows. 
\begin{definition}
The Hierarchical Cell Decomposition associated to an artifact system $\calh$ and property $\varphi$
is a collection $\{\calk_T\}_{T \in \calh}$ of sets of cells, 
such that for each $T \in \calh$, $\calk_T = \calk(\calp'_T)$, where 
the set of polynomials $\calp'_T$ is defined as follows. First, let $\calp_T$ consist of the following:
\begin{itemize}
\item all polynomials mentioned in any condition over $\bar{x}^T$ in $\Gamma$ and the property $\varphi$,
\item polynomials $\{e | e \in \calert\} \cup \{e - e' | e, e' \in \calert\}$, and
\item for every $T_c \in child(T)$ and subset $\bar{x} \subseteq \bar{x}^{T_c}_{ret}$, 
the set of polynomials $\{P^{T_c \rightarrow T} | P \in \poly(c, \bar{x}^{T_c}_{in} \cup \bar{x}), c \in \calk_{T_c} \}$.
\end{itemize}
Next, let $\calp^s_T = \calp_T \cup \bigcup_{c \in \calk(\calp_T)} \poly(c, \bar{x}^T_{in} \cup \bar{s}^T)$. Finally, 
$\calp'_T = \calp_T^s \cup \bigcup_{c \in \calk(\calp_T^s)} \poly(c, \bar{x}^T_{in})$.
\end{definition}
\victor{I reformulated a bit the above.}

The Hierarchical Cell Decomposition satisfies the following property, as desired.
%Intuitively,
%the first property says that for two cells $c$ and $c'$ that are over the same set of numeric expressions of $\bar{x}^T_{in}$, 
%then the projection $c | \bar{x}^T_{in}$ and $c' | \bar{x}^T_{in}$ are either equal or disjoint. The same holds $\bar{x}^T_{in} \cup \bar{s}^T$.
%The second property says that two 

\begin{lemma} \label{lem:cell-projection}
Let $T$ be a task and $\calp_T'$ as above. For every pair of cells $c_1, c_2 \in \calk_T$, and 
$\bar{x} = (\bar{x}^T_{in} \cup \bar{s}^T)$ or $\bar{x} = \bar{x}^T_{in}$,
if $c_1 \in \calk(\calp'_T, \cale_1)$ and $c_2 \in \calk(\calp'_T, \cale_2)$ where $\cale_1 | \bar{x} = \cale_2 | \bar{x}$,
then $c_1 | \bar{x}$ and $c_2 | \bar{x}$ are either equal or disjoint.
\end{lemma}

\begin{proof}
We prove the lemma for the case when $\bar{x} = \bar{x}^T_{in}$. The proof is similar for 
$\bar{x} = \bar{x}^T_{in} \cup \bar{s}^T$. 

Let $\tilde{\calp}_T^s = \bigcup_{c \in \calk(\calp_T^s)} \poly(c, \bar{x}^T_{in})$.
For each cell $c \in \calk(\calp', \\ \cale)$, since $\calp' | \cale =  (\calp_T^s | \cale) \cup (\tilde{\calp}_T^s | \cale)$ 
as $\calp' = \calp_T^s \cup \tilde{\calp}_T^s$,
there exist $c_1 \in \calk(\calp_T^s, \cale)$ and $c_2 \in \calk(\tilde{\calp}_T^s, \cale)$ such that $c = c_1 \cap c_2$.
Then consider $c | \bar{x}^T_{in}$. Since all polynomials in $\tilde{\calp}_T^s$ are over expressions of $\bar{x}^T_{in}$,
we have $c | \bar{x}^T_{in} = (c_1 \cap c_2) | \bar{x}^T_{in} = (c_1 | \bar{x}^T_{in}) \cap c_2$. 
And by definition, $\poly(c_1, \bar{x}^T_{in}) \subseteq \tilde{\calp}_T^s$, so 
$c_2$ uniquely defines the sign conditions for $\poly(c_1, \bar{x}^T_{in})$, which means that either $c_2 \cap c_1 | \bar{x}^T_{in} = \emptyset$
or $c_2 \subseteq c_1 | \bar{x}^T_{in}$. 
And as $c_2 \cap c_1 | \bar{x}^T_{in} = c | \bar{x}^T_{in}$ is non-empty, $c | \bar{x}^T_{in} = c_2$.

Therefore, for every $c_1 \in \calk(\calp_T', \cale_1)$ and $c_2 \in \calk(\calp_T', \cale_2)$ where
$\cale_1 | \bar{x}^T_{in} = \cale_2 | \bar{x}^T_{in} = \cale$, there exist cells $\tilde{c}_1, \tilde{c}_2 \in \calk(\calp_T', \cale)$
such that $c_1 | \bar{x}^T_{in} = \tilde{c}_1$ and $c_2 | \bar{x}^T_{in} = \tilde{c}_2$.
Since $\tilde{c}_1$ and $\tilde{c}_2$ are either disjoint or equal, $c_1 | \bar{x}^T_{in}$ and $c_2 | \bar{x}^T_{in}$ are also either disjoint or equal.
\end{proof}

From the above lemma, the following is obvious:
\begin{corollary} \label{cor:single-cell}
For every task $T$ and $c \in \calk_T$, $c | \bar{x}^T_{in}$ and $c | (\bar{x}^T_{in} \cup \bar{s}^T)$
are single cells in $\calk_T$.
\end{corollary}

In view of the corollary, we use the notations of single-cell operators (projection, refinement, etc.) 
on $c | \bar{x}^T_{in}$ and $c | (\bar{x}^T_{in} \cup \bar{s}^T)$ in the rest of our discussion.

To be able to connect with child tasks, we show the following property of $\calk_T$:
\begin{lemma} \label{lem:cell-containment}
For all tasks $T$ and $T_c$ where $T_c \in child(T)$,
and every cell $c_1 \in \calk_T$ and $c_2 \in \calk_{T_c}$
where $c_1 \in \calk(\calp'_T, \cale_1)$ and $c_2 \in \calk(\calp'_{T_c}, \cale_2)$,
for each set of variables $\bar{x} = \xttcup \cup \bar{y}$ where $\bar{y}$ is some subset of $\xttcdown$,
if $\cale_1 | \bar{x} = (\cale_2)^{T_c \rightarrow T} | \bar{x}$,
then either $(1)$ 
$c_1 \sqsubseteq_{\bar{x}} (c_2)^{T_c \rightarrow T}$
%there exists cell $c_2' \subseteq (c_2)^{T_c \rightarrow T} | \bar{x}$
%such that $c_1 \sqsubseteq c_2'$ 
or $(2)$
$c_1 | \bar{x}$ is disjoint from $(c_2)^{T_c \rightarrow T} | \bar{x}$.
\end{lemma} 

\begin{proof}
Denote by $\calp_{T_c}^{\bar{x}}$ the set of polynomials $\{P^{T_c \rightarrow T} | P \in \poly(c, \bar{x}), c \in \calk_{T_c} \}$.
For each cell $c_1 \in \calk(\calp'_T, \cale_1)$, there exists $\tilde{c}_1 \in \calk(\calp_{T_c}^{\bar{x}}, \cale_1)$ 
such that $c_1 \subseteq \tilde{c}_1$.
For each cell $c_2 \in \calk(\calp'_{T_c}, \cale_2)$, 
as $\cale_1 | \bar{x} = (\cale_2)^{T_c \rightarrow T} | \bar{x}$,
$(c_2)^{T_c \rightarrow T} | \bar{x}$ is a union of cells in $\calk(\calp_{T_c}^{\bar{x}}, \cale_1)$. 
So either $\tilde{c}_1$ is disjoint with or contained in 
$(c_2)^{T_c \rightarrow T} | \bar{x}$. 
If $\tilde{c}_1$ and $(c_2)^{T_c \rightarrow T} | \bar{x}$ are disjoint, then 
$(c_2)^{T_c \rightarrow T} | \bar{x}$ and $c_1 | \bar{x}$ are disjoint. 
If $\tilde{c}_1 \subseteq (c_2)^{T_c \rightarrow T} | \bar{x}$, then 
we have $c_1  \sqsubseteq \tilde{c}_1 \subseteq (c_2)^{T_c \rightarrow T} | \bar{x}$
so $c_1 \sqsubseteq_{\bar{x}} (c_2)^{T_c \rightarrow T}$.
\end{proof}

%
%By construction, $\calk_T$ satisfies the following property:
%\begin{lemma} \label{lem:arith-projection}
%For every task $T$ and $T_c \in child(T)$,
%for every cell $c \in \calk_T$ and $c' \in \calk_{T_c}$, 
%for $\tilde{c} = c | \xttcdown \cup \xttcup$ and 
%$\tilde{c}' = \left(c' | \bar{x}^T_{in} \cup \bar{x}^T_{ret}\right)^{T_c \rightarrow T}$, $\tilde{c}$ and $\tilde{c}'$ are either disjoint
%or $\tilde{c} \subseteq \tilde{c}'$.
%\end{lemma} \yuliang{Need to define the notation.}

\subsection{Extended Isomorphism Types} \label{sec:extended-tau}
Given the Hierarchical Cell Decomposition $\{\calk_T\}_{T \in \calh}$, we can extend our notion of isomorphism type to support arithmetic.
\begin{definition}
For navigation set $\cale_T$, equality type $\sim_{\tau}$ over $\cale^+_T$ and $c \in \calk_T$,
the triple $\tau = (\cale_T, \sim_{\tau}, c)$ is an extended $T$-isomorphism type if
\begin{itemize}
\item $(\cale_T, \sim_\tau)$ is a $T$-isomorphism type, and
\item $c = c(\sigma, \calp_T' | (\calert \cap \cale_T^+) )$ for some sign condition $\sigma$ of \\ $\calp_T' | (\calert \cap \cale_T^+)$
such that for every numeric expression $e, e' \in \cale_T^+$, $e \sim_\tau e'$ iff $\sigma(e - e') = 0$ and $e \sim_\tau 0$
iff $\sigma(e) = 0$.
\end{itemize}
\end{definition}

For each condition $\pi$ over $\bar{x}^T$ and extended $T$-isomorphism type $\tau$, 
$\tau \models \pi$ is defined as follows. For each polynomial inequality ``$P \circ 0$'' in $\pi$ where $\circ \in \{<, > , =\}$,
$P \circ 0$ is true iff $\sigma(P) \circ 0$ where $\sigma$ is the sign condition of $c$. 
The rest of the semantics is the same as in normal $T$-isomorphism type.

The projection of an extended $T$-isomorphism type $\tau$ on $\bar{x}^T_{in}$ and $\bar{x}^T_{in} \cup \bar{s}^T$ 
is defined in the obvious way.
For $\tau = (\cale_T, \sim_\tau, c)$, we define that $\tau | \bar{x} = (\cale_T | \bar{x}, \sim_\tau | \bar{x}, c | \bar{x})$
for $\bar{x} = \bar{x}^T_{in}$ or $\bar{x} = \bar{x}^T_{in} \cup \bar{s}^T$.
The projection of $\tau$ on $\bar{x}^T_{in}$ and $\bar{x}^T_{in} \cup \bar{s}^T$
up to length $k$ is defined analogously.
The projection of every extended $T$-isomorphism type on $\bar{x}^T_{in} \cup \bar{s}^T$ is an extended $TS$-isomorphism type.

To extend the definitions of local symbolic run and symbolic tree of runs,
we first replace $T$-isomorphism type with extended $T$-isomorphism type 
and $TS$-isomorphism type with extended $TS$-isomorphism type
in the original definitions.
The semantics is extended with the following rules.

For two symbolic instances $I$ and $I'$ where the cell of $I$ is $c$ and the cell of $I'$ is $c'$, 
$I'$ is a valid successor of $I$ by applying service $\sigma'$ 
if the following conditions hold in addition to the original requirements:
\begin{itemize}
\item if $\sigma'$ is an internal service, then $c | \bar{x}^T_{in} = c' | \bar{x}^T_{in}$.
\item if $\sigma'$ is an opening service of $T_c \in child(T)$ or closing service of $T$, then $c = c'$.
\item if $\sigma'$ is a closing service of $T_c \in child(T)$, then $c' \sqsubseteq c$.
%\yuliang{Perhaps $\subseteq$ needs to be formally defined in the notation part.}
\end{itemize}
The counters $\bar{c}$ are updated as in transitions between symbolic instances without arithmetic.
Each dimension of $\bar{c}$ corresponds to an extended $TS$-isomorphism type.

For each local symbolic run $\trt = (\tauin, \tauout, \II)$, the following are additionally satisfied:
\begin{itemize}
\item $c_{in} = c_0 | \bar{x}^T_{in}$, where $c_{in}$ is the cell of $\tauin$ and
$c_0$ is the cell of $\tau_0$;
\item if $\tauout \neq \bot$, then $c_{out} \sqsubseteq_{\bar{x}^T_{in} \cup \bar{x}^T_{ret}} c_{\gamma - 1}$,
where $c_{out}$ is the cell of $\tauout$ and $c_{\gamma - 1}$ is the cell of $\tau_{\gamma - 1}$. 
\end{itemize}

In a symbolic tree of runs $\Sym$,
for every two local symbolic runs $\trt = (\tauin, \tauout, \II)$ and 
$\trtc = (\tauin', \tauout', \\ \{(I_i', \sigma_i')\}_{0 \leq i < \gamma'})$ where $T_c \in child(T)$,
if $\trtc$ is connected to $\trt$ by an edge labeled with index $i$, then 
the following conditions must be satisfied in addition to the original requirements:
\begin{itemize}
\item for the cell $c_i$ of symbolic instance $I_i$ and the cell $c_{in}$ of $\tauin'$, $c_i \sqsubseteq c^{T_c \rightarrow T}_{in}$.
\item if $\trtc$ is a returning local symbolic run, then for the cells $c_{out}$ of $\tauout'$ and 
$c_j$ of $I_j$ where $j$ is the smallest index such that $\sigma_j = \sigma_{T_c}^c$ and $j > i$,
we have that $c_j \sqsubseteq_{\bar{x}_\anull} c^{T_c \rightarrow T}_{out}$, 
where $\bar{x}_{\anull} = \{x | x \in \xttcup, x \sim_{\tau_{j-1}} \anull \}$. 
\end{itemize} 

\subsection{Actual Runs versus Symbolic Runs}\label{sec:connecting-runs}

We next show that the connection between actual runs and symbolic runs established in Theorem \ref{thm:actual-symbolic}
still holds for the extended local and symbolic runs. 
The structure of the proof is the same, so we only state the necessary modifications needed to handle arithmetic.

\subsubsection{From Trees of Local Runs to Symbolic Trees of Runs}
Given a tree of local runs $\Tree$, the construction of a corresponding symbolic tree of runs $\Sym$ can be done as follows.
We first construct $\Sym$ from $\Tree$ without the cells following the construction described in 
the proof of the only-if part of Theorem \ref{thm:actual-symbolic}.
Then for each task $T$ and symbolic instance $I$ with extended isomorphism type $\tau$ in some local symbolic run of $T$, 
let $\cale$ be the set of numeric expressions in $\tau$ and $v: \cale \mapsto \mathbb{R}$ the valuation of $\cale$ at $I$.
Then the cell $c$ of $I$ is chosen to be the unique cell in $\calk(\calp_T', \cale)$ that contains $v$.
For cells $c$ and $c'$ of two consecutive symbolic instances $I$ and $I'$ where the service that leads to $I'$ is $\sigma'$,
\begin{itemize}
\item if $\sigma'$ is an internal service, by Lemma \ref{lem:cell-projection}, as $c | \bar{x}^T_{in}$ and $c' | \bar{x}^T_{in}$ 
overlaps, we have $c | \bar{x}^T_{in} = c' | \bar{x}^T_{in}$,
\item if $\sigma'$ is an opening service, $c = c'$ is obvious, and
\item if $\sigma'$ is a closing service, let $\cale$ be the numeric expressions of $c$ and $\cale'$ be the numeric expressions of $c'$.
We have $\cale \subseteq \cale'$ so $\calp_T' | \cale \subseteq \calp_T' | \cale'$. So $c'$ can be written as $c_1 \cap c_2$
where $c_1 \in \calk(\calp_T', \cale)$ and $c_2 \in \calk(\calp_T', \cale' - \cale)$. 
As the values of the preserved numeric expressions are equal in the two consecutive instances, we have $c_1 = c$
so $c \sqsubseteq c'$.
\end{itemize}
Thus, each local symbolic run in $\Sym$ is valid. 
Following a similar analysis, one can verify that for every two connected local symbolic runs $\trt$ and $\trtc$,
the conditions for symbolic tree of runs stated in Appendix \ref{sec:extended-tau} are satisfied due to Lemma \ref{lem:cell-containment}.

\subsubsection{From Symbolic Trees of Runs to Trees of Local Runs}
Given a symbolic tree of runs $\Sym$, we construct the tree of local runs $\Tree$ as follows.
Recall that in the original proof, for each local symbolic run $\trt$, we construct the global isomorphism type $\Lambda$
of $\trt$ and use $\Lambda$ to construct the local run $\rho_T$ and database instance $D_T$.
With arithmetic, the construction of $\Lambda$ remains unchanged
but we use a different construction for $\rho_T$ and $D_T$.

To construct $\rho_T$ and $D_T$, we first define a sequence of mappings 
$\{p_i\}_{0 \leq i < \gamma}$ from the sequence of cells $\{c_i\}_{0 \leq i < \gamma}$ of $\trt$
where each $p_i$ is a mapping from $\cale^+_T \cap \calert$ to $\mathbb{R}$ and $\cale^+_T$ 
is the extended navigation set of $\tau_i$.
Note that each $p_i$ can be also viewed as a point in $c_i$.
The sequence of mappings $\{p_i\}_{0 \leq i < \gamma}$ determines the values of numeric expressions, as we shall see next.
For each mapping $p$ whose domain is the set of numeric expressions $\cale$, 
we denote by $p | \bar{x}$ the projection of $p$ to $\cale \cap (\bar{x} \cup \{x_R.w | x \in \bar{x}\})$. 
Then $\{p_i\}_{0 \leq i < \gamma}$ is constructed as follows:
\begin{itemize}
\item First, we pick an arbitrary point (mapping) $p_{in}$ from $c_{in}$ where $c_{in}$ is the cell of the 
input isomorphism type of $\trt$. 
\item Then, for each equivalence class $\call$ of life cycles in $\trt$,
let $c_{\call}$ be the cell 
of the last symbolic instances in the last dynamic segments
of life cycles in $\call$. Pick a mapping $p_{\call} \in c_{\call}$
such that $p_{\call} | \bar{x}^T_{in} = p_{in}$.
Such a mapping always exists because, by Lemma \ref{lem:cell-projection}, for each $0 \leq i < \gamma$, $c_i | \bar{x}^T_{in} = c_{in}$.
\item Next, for each equivalence class $\cals$ of segments in $\call$,
let $c_{\cals}$ be the cell of the last symbolic instance in segments in $\cals$.
Pick a mapping $p_{\cals}$ from $c_{\cals}$ such that 
$p_{\cals} | (\bar{x}^T_{in} \cup \bar{s}^T) = p_{\call} | (\bar{x}^T_{in} \cup \bar{s}^T)$.
Such a mapping always exists because for each life cycle $L \in \call$ and $I_i$ in $L$,
$c_{\call} | (\bar{x}^T_{in} \cup \bar{s}^T) \sqsubseteq c_i | (\bar{x}^T_{in} \cup \bar{s}^T)$.
\item Finally, for each segment $S = \{(I_i, \sigma_i)\}_{a \leq i \leq b} \in \cals$,
let $p_b = p_\cals$, and for $a \leq i < b$, let $p_i = p_{i + 1} | \bar{x}$
where $\bar{x} = \{x | x \not\sim_{\tau_i} \anull\}$ are the preserved variables from $I_i$ to $I_{i+1}$.
Such mappings always exist because for each $a \leq i < b$, $c_{i+1} \sqsubseteq c_i$.
\end{itemize}

For the sequence of mappings $\{p_i\}_{0 \leq i < \gamma}$ constructed above, the following is easily shown:
\begin{lemma} \label{lem:numeric-equality}
For all local expressions $(i, e)$ and $(i', e')$ in the global isomorphism type $\Lambda$, where $e$ and $e'$ are numeric, 
$(i, e) \sim (i', e')$ implies that $p_i(e) = p_{i'}(e')$.
\end{lemma}

Given the above property, we can construct $\rho_T$ and $D_T$ as follows.
We first construct $\rho_T$ and $D_T$ as in the case without arithmetic.
Then for each equivalence class $[(i, e)]$, we replace the value $[(i, e)]$ in $\rho_T$ and $D_T$ with
the value $p_i(e)$. It is clear that Lemmas \ref{lem:sym-to-loc-finite} and \ref{lem:sym-to-loc-infinite}
still hold since the global equality type in $\Lambda$ remains unchanged.

%The sequence of mappings $\{p_i\}_{0 \leq i < \gamma}$ satisfies the following properties:
%\begin{itemize}
%\item $p_i | \bar{x}^T_{in} = p_j | \bar{x}^T_{in}$ for every $0 \leq i, j < \gamma$,
%\item for $I_i$ and $I_j$ within the same segment $S$ and $i < j$,
%$p_i \preceq p_j$, 
%\item for $I_i$ and $I_j$ within the same life cycle $L$ and $i < j$,
%$p_i | \bar{s}^T \preceq p_j | \bar{s}^T$, and
%\item for two segments $S_a = \{I_i\}_{a \leq i \leq a + k}$ and $S_b = \{I_i\}_{b \leq i \leq b + k}$
%where $S_a \equiv S_b$, it satisfies that $p_{a + i} = p_{b + i}$ for every $0 \leq i \leq k$.
%\end{itemize}

To construct the full tree of local runs $\Tree$ from the symbolic tree of runs, we perform the above construction
in a top-down manner. For each local symbolic run $\trt$, 
%we let $p_{in} = p_0 | \bar{x}^T_{in}$ and if $\trt$ is a returning run, 
%let $p_{out} = p_{\gamma - 1} | (\bar{x}^T_{in} \cup \bar{x}^T_{ret})$.
we first construct $\{p_i\}_{0 \leq i < \gamma}$ for the root $\tilde{\rho}_{T_1}$ of $\Sym$ using the above construction.
Then recursively for each $\trt \in \Sym$ and child $\trtc$ connected to $\trt$ by an edge labeled with index $i$,
we pick a mapping $p_{in}$ from $c_{in}$ of $\trtc$ such that $p_{in}^{T_c \rightarrow T} = p_i | \xttcdown$.
And if $\trtc$ is a returning run, we pick $p_{out}$ from $c_{out}$ of $\trtc$ such that 
$p_{out}^{T_c \rightarrow T} | \bar{x}_{\anull} = 
p_j | \bar{x}_{\anull}$ where $j$ is index of the corresponding closing service $\sigma_{T_c}^c$ at $\trt$, 
and $\bar{x}_{\anull}$ is defined as above. 

We next construct $\{p_i\}_{0 \leq i < \gamma}$ of $\trtc$ similarly to above, except that (1) $p_{in}$ is given, and
(2) if $\trtc$ is a returning run, then
for the equivalence class $\call$ of life cycles where $I_{\gamma - 1}$ is contained in some life cycle $L \in \call$,
we pick $p_{\call}$ such that $p_{\call} | \bar{x}^{T_c}_{in} \cup \bar{x}^{T_c}_{ret} = p_{out}$.
Then $\rho_{T_c}$ and $D_{T_c}$ are constructed following the above approach. 
The tree of local runs $\Tree$ is constructed as described in the proof of Theorem \ref{thm:actual-symbolic}.
Following the same approach, we can show: 
\begin{theorem} \label{thm:arithmetic-correctness}
For every HAS $\Gamma$ and $\hltlfo$ property $\varphi$ with arithmetic,
there exists a symbolic tree of runs $\Sym$ accepted by $\calb_\varphi$ iff
there exists a tree of local runs $\Tree$ and database $D$ such that $\Tree$ is accepted by $\calb_\varphi$ on $D$.
\end{theorem}

\subsection{Complexity of Verification with Arithmetic} \label{sec:complexity-arithmetic}

Similarly to the analysis in Appendix \ref{sec:complexity-no-arith-app}, it is sufficient to
upper-bound the number of $T$-and $TS$-isomorphism types.
To do so, we need to bound the size of $\{\calk_T\}_{T \in \calh}$. 
By the construction of each $\calk_T$ and by Theorem \ref{thm:num-cells}, 
it is sufficient to bound the size of each $\calp_T'$.

We denote by $l$ the number of numeric expressions, $s$ the number of polynomials in $\Gamma$ and $\varphi$,
$d$ the maximum degree of these polynomials, $t$ the maximum bitsize of the coefficients, 
and $h$ the height of the task hierarchy $\calh$.
For each task $T$, we denote by $s(T)$ the number of polynomials in $\calp_T'$
and $d(T)$ the maximum degree of polynomials in $\calp_T'$.

If $T$ is a leaf task, then $|\calp_T| \leq s + l^2$. 
The number of polynomials in $\calp^s_T$ is no more than the product of 
(1) the number of subsets of $\calert$, (2) the maximum number of non-empty cells over $(\cale, \calp_T | \cale)$ and (3)
the maximum number of polynomials in each $\poly(c, \bar{x}^T_{in} \cup \bar{s}^T)$. 
By Theorem \ref{thm:qe}, the number of polynomials is no more than the running time, 
which is bounded by $( (s + l^2) \cdot d )^{O(l^2)}$.
Then by Theorem \ref{thm:num-cells},
the number of non-empty cells over $(\cale, \calp_T | \cale)$ is at most
$((s + l^2) \cdot d)^{O(l)}$. 
Thus, $|\calp^s_T| \leq ( (s + l^2) \cdot d )^{O(l^2)}$.
By the same analysis, we obtain that for $\calp'_T$, $s(T) = |\calp'_T| \leq ( (s + l^2) \cdot d )^{O(l^4)}$.
Similarly, $d(T)$ can be upper-bounded by $( (s + l^2) \cdot d )^{O(l^4)}$.

Next, if $T$ is a non-leaf task, we denote by $s'$ the size of $\calp_T$ and by $d'$ 
the maximum degree of polynomials in $\calp_T$.
We have that $s' \leq (s + l^2) + \sum_{T_c \in child(T)} 2^l (s(T_c) \cdot d(T_c))^{O(l^2)} \cdot (s(T_c) \cdot d(T_c))^{O(l)} 
\leq (s + l^2) + (s(T_c) \cdot d(T_c))^{O(l^2)}$, and
$d' \leq \max_{T_c \in child(T)} (s(T_c) \cdot d(T_c))^{O(l^2)}$.

Following the same analysis as above, we have that 
both $s(T)$ and $d(T)$ are at most $( (s' + l^2) \cdot d' )^{O(l^4)}$.
By solving the recursion, we obtain that 
$s(T), d(T) \leq ( (s + l^2) \cdot d )^{(c \cdot l^6)^h}$ for some constant $c$.
Then by Theorem \ref{thm:num-cells}, $|\calk_T|$ is at most $(s(T) \cdot d(T))^{O(k)}$. So we have
\begin{lemma}\label{lem:size-kt}
For each task $T$, the number of cells in $\calk_T$ is at most 
$( (s + l^2) \cdot d )^{(c \cdot l^6)^{h}} $ for some constant $c$.
\end{lemma}

The space used by the verification algorithm with arithmetic is no more than the space needed
to pre-compute $\{\calk_T\}_{T \in \calh}$ plus the space for the VASS (repeated) reachability for each task $T$.
By Theoream \ref{thm:naive-approach}, for each task $T$, the set $\calk_T$ can be computed in
space $O\left(( (s + l^2) \cdot d )^{(c \cdot l^6)^{h}}\right) $.

For VASS (repeated) reachability, according to the analysis in Appendix \ref{sec:complexity-no-arith-app}, 
state (repeated) reachability can be computed in 
$O( h^2 \cdot N^2 \log^2 M \cdot 2^{c \cdot D \log D} )$ space ($O( h^2 \cdot N^2 \log^2 M)$ w/o. artifact relation), where 
$h$ is the height of $\calh$, $N$ is the size of $(\Gamma, \varphi)$, $M$ is the number of extended $T$-isomorphism types
and $D$ is the number of extended $TS$-isomorphism types.
With arithmetic, $M$ and $D$ are
the products of number of normal $T$-and $TS$-isomorphism types multiplied by $|\calk_T|$ respectively.
As $l$ is less than the number of expressions whose upper bounds are obtained 
in Appendix \ref{sec:complexity-no-arith-app}, by applying Lemma \ref{lem:size-kt}, 
we obtain upper bounds for $M$ and $D$ for the different types of schema.

By substituting the bounds for $M$ and $D$, we have the following results.
Note that for $\Gamma$ without artifact relations, the complexity is dominated by the space for pre-computing
$\{\calk_T\}_{T \in \calh}$.

\begin{theorem} \label{thm:acyclic-arithmetic}
Let $\Gamma$ be a HAS with {\bf acyclic} schema and $\varphi$ an $\hltlfo$ property over $\Gamma$,
where arithmetic is allowed in $\Gamma$ and $\varphi$.
$\Gamma \models \varphi$ can be verified in $2\hexp(N^{O(h + r)})$ deterministic space.
If $\Gamma$ does not contain artifact relation, then $\Gamma \models \varphi$ can be verified in 
$\exp(N^{O(h + r)})$ deterministic space.
\end{theorem}

\begin{theorem} \label{thm:linear-arithmetic}
Let $\Gamma$ be a HAS with {\bf linearly-cyclic} schema and $\varphi$ an $\hltlfo$ property over $\Gamma$,
where arithmetic is allowed in $\Gamma$ and $\varphi$.
$\Gamma \models \varphi$ can be verified in $O(2\hexp(N^{c_1 \cdot h^2}))$ deterministic space, where $c_1 = O(r)$.
If $\Gamma$ does not contain artifact relation, then $\Gamma \models \varphi$ can be verified in 
$O(\exp(N^{c_2 \cdot h^2}))$ deterministic space, where $c_2 = O(r)$.
\end{theorem}

\begin{theorem} \label{thm:cyclic-arithmetic}
Let $\Gamma$ be a HAS with {\bf cyclic} schema and $\varphi$ an $\hltlfo$ property over $\Gamma$,
where arithmetic is allowed in $\Gamma$ and $\varphi$.
$\Gamma \models \varphi$ can be verified in $(h+2)\hexp(O(N))$ deterministic space.
If $\Gamma$ does not contain artifact relation, then $\Gamma \models \varphi$ can be verified in 
$(h+1)\hexp(O(N))$ deterministic space.
\end{theorem}

\section{Undecidability Results} \label{app:undecidability}

We provide a proof of Theorem \ref{thm:restrictions} for relaxing restriction (2). Recall that HAS$^{(2)}$ allows 
subtasks of a given task to overwrite non-null ID variables.
The same proof idea can be used for restrictions (1) to (7).

\begin{proof}
We show undecidability by reduction from the Post Correspondence Problem (PCP) \cite{Post47,sipser}.
Given an instance $P = \{(a_i, b_i)\}_{1 \leq i \leq k}$ of PCP, 
where each $(a_i, b_i)$ is a pair of non-empty strings over $\{0,1\}$, we show how to construct a HAS$^{(2)}$ 
$\Gamma$ and $\hltlfo$ formula $\varphi$ such that there is a solution to $P$ iff 
there exists a run of $\Gamma$ satisfying $\varphi$ (i.e., $\Gamma \not\models \neg \varphi$).

The database schema of $\Gamma$ contains a single relation 
$$G(\underline{id}, \mathtt{next, label})$$ 
where $\mathtt{next}$ is a foreign-key attributes referencing attribute $id$ and
$\mathtt{label}$ is a non-key attribute. Let $\alpha, \beta$ be distinct id values in $G$.
A {\em path} in $G$ from $\alpha$ to $\beta$ is a sequence of IDs  $i_0, \ldots, i_n$ in $G$ where
$\alpha = i_0$, $\beta = i_n$, and for each $j, 0 \leq j < n$, $i_{j+1} = i_j.\mathtt{next}$.
It is easy to see that there is at most one path from $\alpha$ to $\beta$
for which $i_j \neq \alpha, \beta$ for $0 < j < n$, and the path must be simple ($i_0,i_1, \ldots,i_n$ are distinct).
If such a path exists, we denote by $w(\alpha, \beta)$ the sequence of labels $i_0.\mathtt{label}, \ldots, i_n.\mathtt{label}$ 
(a word over $\{0,1\}$, assuming the values of $\mathtt{label}$ are $0$ or $1$).
Intuitively, $\Gamma$ and $\varphi$ do the following given database $G$:
\begin{enumerate}
\item non-deterministically pick two distinct ids $\alpha, \beta$ in $G$ 
\item check that there exists a simple path from $\alpha$ to $\beta$ and that $w(\alpha,\beta)$ witnesses a solution to $P$;
the uniqueness of the simple path from $\alpha$ to $\beta$ is essential to ensure that 
$w(\alpha,\beta)$ is well defined. 
\end{enumerate} 

\noindent
Step 2 requires simultaneously parsing $w(\alpha,\beta)$ as $a_{s_1} \dots a_{s_m}$  and $b_{s_1} \dots b_{s_m}$
for some $s_i \in [1,k], 1 \leq i \leq m$, by synchronously walking the path from $\alpha$ to $\beta$ with two pointers
$P_a$ and $P_b$. More precisely, $P_a$ and $P_b$ are initialized to $\alpha$. 
Then repeatedly, an index $s_j \in [1, k]$ is picked
non-deterministically, and $P_a$ advances $|a_{s_j}|$ steps to a new position $P_a'$, such that 
the sequence of labels along the path from $P_a$ to $P_a'$ is $a_{s_j}$ and no id along the path equals $\alpha$ or $\beta$.
Similarly, $P_b$ advances $|b_{s_j}|$ steps to a new position $P_b'$, such that 
the sequence of labels along the path from $P_b$ to $P_b'$ is $b_{s_j}$ and  no id along the path equals $\alpha$ or $\beta$.
This step repeats until $P_a$ and $P_b$ simultaneously reach $\beta$ (if ever).
The property $\varphi$ checks that eventually $P_a = P_b = \beta$,
so $w(\alpha,\beta)$ witnesses a solution to $P$. 

In more detail, we use two tasks $T_p$ and $T_c$ where $T_c$ is a child task of $T_p$ (see Figure \ref{fig:non-null}).

\vspace*{-3mm}
\begin{figure}[!ht]
\centering
\includegraphics[scale=0.8]{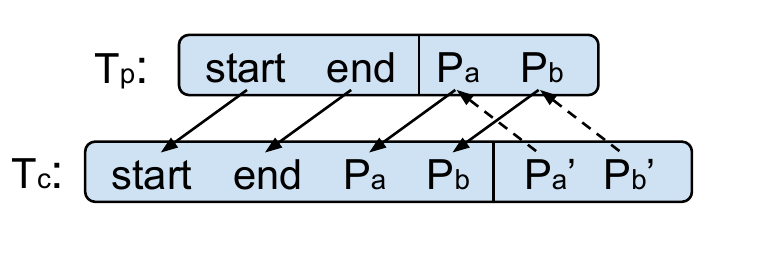}
\vspace*{-6mm}
\caption{Undecidiability for HAS$^{(2)}$}
\vspace*{-2mm}
\label{fig:non-null}
\end{figure}

\vspace{4mm}
\noindent
Task $T_p$ has two input variables $\emph{start}, \emph{end}$ (initialized to distinct ids $\alpha$ and $\beta$ by the global precondition), and two artifact variables $P_a$ and $P_b$ (holding the two pointers). 
$T_p$ also has a binary artifact relation $S$ whose set variables are $(P_a,P_b)$. 
At each segment of $T_p$, the subtask $T_c$ is called with $(P_a,P_b,\emph{start}, \emph{end})$ passed as input.
Then an internal service of $T_c$ computes $P_a'$ and $P_b'$, such that
$P_a, P_a', P_b$ and $P_b'$ satisfy the condition stated above for some $s_j \in [1,k]$.
Then $T_c$ closes and returns $P_a'$ and $P_b'$ to $T_p$, overwriting $P_a$ and $P_b$ 
(note that this is only possible because restriction $(2)$ is lifted).
At this point we would like to call $T_c$ again, but multiple calls to a subtasks are disallowed between internal transitions.
To circumvent this, we equip $T_p$ with an internal service that simply propagates $(P_a, P_b, \emph{start},\emph{end})$.
The variables $\emph{start},\emph{end}$ are automatically propagated as input variables of $T_p$.
Propagating $(P_a,P_b)$
is done by inserting it into $S$ and retrieving it in the next configuration (so $\delta= \{+S(P_a,P_b), -S(P_a,P_b)\}$). 
Now we are allowed to call again $T_c$, as desired.

It can be shown that there exists a solution to $P$ iff there exists a run of the above system that reaches a configuration in which
$P_a = P_b = \emph{end}$. This can be detected by a second internal service \emph{success} of $T_p$ with pre-condition 
$P_a = P_b = \emph{end}$. Thus, the $\hltlfo$ property $\varphi$ is simply [{\bf F} (\emph{success})]$_{T_p}$. 
Note that this is in fact an HLTL formula.
Thus, checking $\hltlfo$ (and indeed HLTL) properties of HAS$^{(2)}$ systems is undecidable.
\end{proof}

%We prove undecidability by reduction from the 2-counter machine state-reachability problem \cite{Minsky-computation67}
%and from the Post-correspondence problem.

%If ID variables can be overwritten when child tasks return, then
%undecidability can be proved by reduction from the state-reachability problem of 2-counter machines \cite{Minsky-computation67}.
%Informally, given a 2-counter machine $M$, we can construct a HAS $\Gamma$ as follows.
%The database schema of $\Gamma$ consists of a single table $R(\underline{id}, c_1, c_2, q)$ 
%where intuitively each $(c_1, c_2, q)$ associated with an ID of $R$ represents a possible configuration of $M$.
%Also $\Gamma$ has two tasks $T_p$ and $T_c$ where $T_p$ is the parent task of $T_c$. 
%$T_p$ contains a single ID variable $i$, which is passed to $T_c$ as input. 
%$T_c$ has a single internal service which computes an ID variable $i'$, such that
%for $(i, c_1, c_2, q)$ and $(i', c_1', c_2', q')$ in $R$, $(c_1', c_2', q')$ is a successor configuration of $(c_1, c_2, q)$ in $M$.
%Then $i'$ is returned to $T_p$ and overwrites $i$. In $T_p$, there is an internal service which
%inserts $i$ to an artifact relation $S$ and at the same time retrieves from $S$,
%so that the value of $i$ remains unchanged as the service is applied. 
%It is clear that whether a state $q_f$ is reachable 
%starting from counter values $(0, 0)$ iff there exists a run of $\Gamma$ reaching $q = q_f$,
%which can be expressed as a $\hltlfo$ property, thus verification is undecidable.

\end{document}